\documentclass[11pt,twoside]{amsart}
\usepackage{epic}
\usepackage{amsmath}
\usepackage{amssymb}  \usepackage{color}
\usepackage{amscd}
\usepackage{amsfonts}
\usepackage{amsthm} \usepackage[all]{xy}
\usepackage[colorlinks,linkcolor=blue,anchorcolor=blue,citecolor=blue]{hyperref}

\newcommand{\be}{\begin{equation}}
\newcommand{\ee}{\end{equation}}
\newcommand{\bea}{\begin{eqnarray}}
\newcommand{\eea}{\end{eqnarray}}
\newcommand{\ben}{\begin{eqnarray*}}
\newcommand{\een}{\end{eqnarray*}}

\newcommand{\bC}{\mathbb{C}} \newcommand{\bL}{\mathbb{L}} 
\newcommand{\bM}{\mathbb{M}} \newcommand{\bP}{\mathbb{P}} 
\newcommand{\bZ}{\mathbb{Z}} \newcommand{\bR}{\mathbb{R}}
\newcommand{\cS}{\mathcal{S}} \newcommand{\cX}{\mathcal{X}}

\newcommand{\cor}[1]{\langle {#1} \rangle} 

\DeclareMathOperator{\SL}{SL} \DeclareMathOperator{\Img}{Im} \DeclareMathOperator{\diag}{diag}

\newtheorem{prop}{Proposition}[section]
\newtheorem{thm}[prop]{Theorem}
\newtheorem{lemma}[prop]{Lemma}

\newtheorem{example}[prop]{Example}
\newtheorem{conj}[prop]{Conjecture}

\definecolor{A}{rgb}{.75,1,.75}

\begin{document}

\title[Quantum McKay Correspondence for Disc Invariants]
{Quantum McKay Correspondence for Disc Invariants of orbifold vertex}
\author{Hua-Zhong Ke, Jian Zhou}

\address{Department of Mathematical Sciences, Tsinghua University, Beijing, 100084, China}
\email{kehuazh@mail.sysu.edu.cn}
\address{Department of Mathematical Sciences, Tsinghua University, Beijing, 100084, China}
\email{jzhou@math.tsinghua.edu.cn}

\begin{abstract}
For a finite abelian subgroup $G \subset SL(3,\bC)$, we describe a a systematic procedure to find toric crepant resolutions of orbifold vertex $[\bC^3/G]$, and show that the generating series of certain disc invariants of the orbifold vertex can be suitably identified with  the generating series of certain disc invariants of its toric crepant resolutions.
\end{abstract}

\maketitle
\tableofcontents

\section{Introduction}

Let $G \subset SL(n, \bC)$ be a finite subgroup,
and let $\mathcal{X}$ be the orbifold $[\bC^n/G]$.
When $n = 2$ or $3$,
the coarse moduli space $\bC^n/G$ of $\mathcal{X}$ always admits a crepant resolution $Y$
(when $n=3$, $Y$ may not be unique).
McKay correspondence predicts that invariants of $Y$ should coincide with orbifold invariants of $\mathcal{X}$.
See \cite{Reid} for an exposition of such results for classical invariants.
After the introduction of orbifold Gromov-Witten invariants,
Ruan \cite{Ru} made the Crepant Resolution Conjecture that the small quantum cohomology of $Y$ is related to the orbifold small quantum cohomology of $\mathcal{X}$
in a suitable way.
There have recently appeared much work and some refinements of this conjecture.
See e.g. \cite{BG,CCIT,Z08} and references therein.
While most of the work in the existing literature focus on such quantum McKay correspondence for closed string invariants,
in this paper we will study the quantum McKay correspondence for some open string invariants.

We will be content with the case when $n=3$ and $G$ is abelian in this work.
In this case it is well-known that $\bC^3/G$ admits a toric crepant resolution \cite{MOP}.
Toric geometry then provides a unified approach to compute the disc invariants of both $[\bC^3/G]$ and 
the crepant resolutions of $\bC^3/G$.
This is compatible with the physical point of view of
gauged linear sigma model (GLSM) used by Witten to study phase transitions in string theory \cite{W}.
Mathematically,
we understand both the orbifold $[\bC^3/G]$ and its toric crepant resolution as  symplectic reduction of a Hamiltonian $(S^1\big)^k$-action on $\mathbb{C}^{3+k}$ by diagonal matrices,
with moment map $\mu: \mathbb{C}^{3+k} \to \bR^k$.
The regular values of $\mu$ form various open chambers in $\bR^k$,
when $\epsilon \in \bR^k$ moves in one of these chambers,
$Y_\epsilon: = \mu^{-1}(\epsilon)/(S^1)^k$ is a toric crepant resolution of $Y_0:=\mu^{-1}(0)/(S^1)^k$,
which we identify with $\bC^3/G$.
Gromov-Witten invariants of $Y_\epsilon$ should define analytic functions on the complexifications of these chambers near the infinity,
and orbifold Gromov-Witten invariants should define an analytic function on the complexification of a neighborhood of the origin.
The quantum McKay correspondence should then be understood as suitable identification of these functions.
At least in the case of genus zero Gromov-Witten invariants,
this can be checked in may cases because one can use the Picard-Fuchs system (GKZ system)
associated to the GLSM to compute genus zero Gromov-Witten invariants of $Y$ via mirror symmetry \cite{CKYZ,HKTY}.
We will treat this in a separate paper.
In this paper we extend this picture to the open string case when one considers a special Lagrangian submanifold of $Y$ as introduced in \cite{AV}.
For both $\cX = [\bC^3/G]$ and its toric crepant resolution $Y$ one can explicitly compute their disc invariants
by the methods of \cite{BC}  and \cite{FL} respectively.
One can verify that they arise from the same GKZ system, through different charge vectors,
hence they can be identified with each other by suitable change of variables.

The rest of the paper is organized as follows. In Section $2$, we first follow \cite{MOP} to find a projective toric resolution $Y$ of $\mathcal{X}$ and then give the GLSM description of $Y$ to write down the superpotential; in Section $3$, we follow \cite{BC,CCIT} to find the orbifold disc potential; in Section $4$, we give the change of variables and compare the two disc potentials; in Section $5$, we give some examples.

\section{Toric Crepant Resolutions of $\mathbb{C}^3/G$}\label{section 2}

In this section we will describe the toric crepant resolutions of $\bC^3/G$,
where $G$ is a nontrivial finite abelian subgroup of $\SL(3,\mathbb{C})$
acting on $\mathbb{C}^3$ diagonally.

\subsection{Fermionic shifts}

We will denote the three linear coordinates on $\bC^3$ by $z_0, z_1, z_2$ respectively.
For each $g\in G$,
there are unique rational numbers $F_g^{(0)},F_g^{(1)},F_g^{(2)}\in[0,1)$ (called the {\em fermionic shifts} of $g$)
such that
$$g.(z_0,z_1,z_2)=\big(z_0\cdot e^{2\pi\sqrt{-1}F_g^{(0)}} ,z_1\cdot e^{2\pi\sqrt{-1}F_g^{(1)}},z_2\cdot e^{2\pi\sqrt{-1}F_g^{(2)}}\big).$$
The {\em age} of $g$ is defined by
\begin{eqnarray}
a(g):=F_g^{(0)}+F_g^{(1)}+F_g^{(2)}.\label{age}
\end{eqnarray}
We shall use the additive notation for $G$, and in this notation, we have for $j=0,1,2$,
\begin{eqnarray}
 F^{(j)}_{g+h}=\langle F^{(j)}_{g}+F^{(j)}_{h}\rangle,\quad\forall g,h\in G.\label{adno}
\end{eqnarray}
Here for any $x\in\mathbb{R}$, we write $x=\lfloor x\rfloor+\langle x\rangle$,
where $\lfloor x\rfloor\in\mathbb{Z}$ and $0\leqslant\langle x\rangle<1$.
Since $G\neq\{0\}$, it follows that the ages take values in $\{0,1,2\}$.

For $j=0,1,2$,
let $G_j=\{g\in G\mid F^{(j)}_g=0\}$.
Then each $G_j$ is a subgroup of $G$.
Note that the quotient group $G/G_j$ has a natural action on the $z_j$-axis
which is effective,
and hence it is necessarily a cyclic group.
So for any $g \in G$,
$$0 \leq F_g^{(j)} \leq 1 - \frac{1}{|G/G_j|}.$$
When $G_j \neq G$,
one can find $\xi^{(j)}\in G$ such that $F^{(j)}_{\xi^{(j)}}=\frac{1}{\vert G/G_j\vert}$.
Note that $G\neq\{0\}$, and therefore
if both $G/G_1$ and $G/G_2$ are trivial,
i.e., for any $g \in G$,
$F_g^{(1)} = F^{(2)}_g = 0$,
then we also have $F^{(0)}_g = 0$,
and so $g$ is trivial.
Since we are working with nontrivial $G$,
we can assume without loss of generality that  $G/G_1\neq\{0\}$
and choose $\alpha \in G$ with
\begin{eqnarray}
F^{(1)}_{\alpha}=\frac{1}{\vert G/G_1\vert}.\label{0}
\end{eqnarray}
For $i=0,2$,
the $G$-action on the $z_i$-axis factors through $G/G_i$,
so there are integers $b, c$ such that
\ben
&& F^{(0)}_{\alpha} = \frac{b}{|G/G_0|}, \quad 0 \leq b < |G/G_0|, \\
&& F^{(2)}_{\alpha} = \frac{c}{|G/G_2|}, \quad 0 \leq c < |G/G_2|.
\een
In other words,
\be
\alpha = \begin{pmatrix}
e^{\frac{2\pi\sqrt{-1}b}{|G/G_0|}} & 0 & 0 \\
0 & e^{\frac{2\pi\sqrt{-1}}{|G/G_1|}} & 0 \\
0 & 0 & e^{\frac{2\pi\sqrt{-1}c}{|G/G_2|}}
\end{pmatrix}
\ee
Moreover, for each $g\in G_1$, if $F^{(2)}_g=0$,
then $F^{(0)}_g=0$ since we already have $F^{(1)}_g=0$, which implies that $g=0$.
Therefore $G_1$ acts effectively on the $z_2$-axis
and hence is a cyclic group.
So we can choose a generator $\beta$ of $G_1$ such that
\begin{eqnarray}
F^{(2)}_{\beta}=\frac{1}{\vert G_1\vert}-\delta_{\vert G_1\vert,1}.\label{5}
\end{eqnarray}
By the definition of $G_1$,
we have automatically
\be \label{eqn:F(beta_1)}
F^{(1)}_{\beta} = 0.
\ee
There is a unique integer $a$ such that
\be
F^{(0)}_{\beta} = \frac{a}{|G_1|}, \quad 0 \leq a < |G_1|, \quad (a, |G_1|) = 1.
\ee
In other words,
\be
\beta = \begin{pmatrix}
e^{\frac{2\pi\sqrt{-1}a}{|G_1|}} & 0 & 0 \\
0 & 1 & 0 \\
0 & 0 & e^{\frac{2\pi\sqrt{-1}}{|G_1|}}
\end{pmatrix}
\ee

For later use, we prove the following two Lemmas.

\begin{lemma}\label{lem}
Assume that $\vert G\vert>\vert G_0\vert$.
Then there exists $$h\in G_s:=\{g\in G\vert a(g)=1\},$$ such that $F^{(0)}_h=\frac{1}{\vert G/G_0\vert}$.
\end{lemma}

\begin{proof}
First  find $\xi\in G$ such that $F^{(0)}_{\xi}=\frac{1}{\vert G/G_0\vert}$.

If $G_0=\{0\}$, then $F^{(0)}_{\xi}=\frac{1}{\vert G\vert}$. Therefore,
\begin{eqnarray*}
0<a(\xi)&=&F^{(0)}_{\xi}+F^{(1)}_{\xi}+F^{(1)}_{\xi}\\
&\leqslant&\frac{1}{\vert G\vert}+\bigg(1-\frac{1}{\vert G/G_1\vert}\bigg)+\bigg(1-\frac{1}{\vert G/G_2\vert}\bigg)\\
&\leqslant&\frac{1}{\vert G\vert}+\bigg(1-\frac{1}{\vert G\vert}\bigg)+\bigg(1-\frac{1}{\vert G\vert}\bigg)\\
&=&2-\frac{1}{\vert G\vert}<2.
\end{eqnarray*}
Hence $a(\xi)=1$ and we can take $h=\xi$.

If $G_0\neq\{0\}$, then we have the following two cases:\\
\noindent (1) $a(\xi)=1$. Take $h=\xi$. \\
\noindent (2) $a(\xi)=2$. Note that $G_0$ acts effectively on the $z_1$-axis and hence is a cyclic group.
So we can find $\eta\in G_0$ such that
$$F^{(0)}_{\eta}=0,\quad F^{(1)}_{\eta}=\frac{1}{\vert G_0\vert},\quad
F^{(2)}_{\eta}=1-\frac{1}{\vert G_0\vert}.$$
Also note that
\begin{eqnarray*}
&&\vert G_0\vert\cdot F^{(2)}_{\xi}-\vert G_0\vert\cdot\Big(1-F^{(1)}_{\xi}\Big) \\
&=&\vert G_0\vert\cdot\bigg(F^{(0)}_{\xi}+F^{(1)}_{\xi}+F^{(2)}_{\xi}-1-F^{(0)}_{\xi}\bigg)\\
&=&\vert G_0\vert\cdot\bigg(1-\frac{1}{\vert G/G_0\vert}\bigg)\\
&\geqslant&\vert G_0\vert\cdot\bigg(1-\frac{1}{2}\bigg)\\
&=&\frac{\vert G_0\vert}{2}\geqslant 1.
\end{eqnarray*}
So we may find an integer $k_0$ such that
$$0<\vert G_0\vert\cdot(1-F^{(1)}_{\xi})\leqslant k_0\leqslant\vert G_0\vert\cdot F^{(2)}_{\xi}
<\vert G_0\vert.$$
With this $k_0$,
we have $$1\leqslant F^{(1)}_{\xi}+\frac{k_0}{\vert G_0\vert}<2,\quad 0\leqslant F^{(2)}_{\xi}
 -\frac{k_0}{\vert G_0\vert}<1,$$
and therefore
\ben
&& \biggl\langle F^{(1)}_{\xi}+\frac{k_0}{\vert G_0\vert}\biggr\rangle
= F^{(1)}_{\xi}+\frac{k_0}{\vert G_0\vert} -1,\\
&& \biggl\langle  F^{(2)}_{\xi} -\frac{k_0}{\vert G_0\vert} \biggr\rangle
=  F^{(2)}_{\xi} -\frac{k_0}{\vert G_0\vert}.
\een
Now we have
\begin{eqnarray*}
a(\xi+k_0\eta)&=&F^{(0)}_{\xi+k_0\eta}+F^{(1)}_{\xi+k_0\eta}+F^{(2)}_{\xi+k_0\eta}\\
&=&F^{(0)}_{\xi}+\langle F^{(1)}_{\xi}+\frac{k_0}{\vert G_0\vert}\rangle
+\langle F^{(2)}_{\xi}-\frac{k_0}{\vert G_0\vert}\rangle\\
&=& F^{(0)}_{\xi}+F^{(1)}_{\xi}+\frac{k_0}{\vert G_0\vert}-1+F^{(2)}_{\xi}-\frac{k_0}{\vert G_0\vert} \\
&=&2-1=1.
\end{eqnarray*}
So we can take $h=\xi+k_0\eta$.
\end{proof}

\begin{lemma}
We have a bijection of sets:
\begin{eqnarray}
 \{0,1,2\dots, |G/G_1|-1\}\times \{0,1, \dots, |G_1|-1\}&\xrightarrow{1:1}&G\nonumber\\
(k, l)&\mapsto& k\alpha + l \beta.\label{1}
\end{eqnarray}
\end{lemma}

\begin{proof}
We have a short exact sequence of abelian groups:
$$0 \to G_1 \stackrel{i}{\to} G \stackrel{\pi}{\to} G/G_1 \to 0,$$
where $i$ is the inclusion
and $\pi$ is the quotient map.
Furthermore,
both $G_1$ and $G/G_1$ are cyclic,
and we have chosen a generator $\beta$ of $G_1$,
and an element $\alpha$ of $G$ whose class in $G/G_1$
is a generator of $G/G_1$.
Given $g\in G$,
there is a unique $k$
such that $\pi(g) = k \cdot \pi(\alpha)$.
Because $\pi(g- k\alpha) = 0$,
$g-k\alpha \in \ker \pi = \Img i$,
so there is a unique $l$ such that
$$g - k\alpha = l \beta.$$
\end{proof}

Note we have
\be
k\alpha + l \beta = \begin{pmatrix}
e^{2\pi\sqrt{-1}(\frac{k b}{|G/G_0|}+\frac{la}{|G_1|})} & 0 & 0 \\
0 & e^{2\pi\sqrt{-1}\frac{k}{|G/G_1|}} & 0 \\
0 & 0 & e^{2\pi\sqrt{-1}(\frac{kc}{|G/G_2|}+\frac{l}{|G_1|})}
\end{pmatrix}
\ee
I.e.,
\bea
&& F^{(0)}_{k\alpha + l \beta} =
\biggl\langle \frac{k b}{|G/G_0|}+\frac{la}{|G_1|} \biggr\rangle, \label{eqn:FS0}\\
&& F^{(1)}_{k\alpha + l \beta} = \frac{k}{|G/G_1|}, \label{eqn:FS1} \\
&& F^{(2)}_{k\alpha + l \beta} = \biggl\langle \frac{kc}{|G/G_2|}+\frac{l}{|G_1|} \biggr \rangle. \label{eqn:FS2}
\eea

\subsection{Lattice of invariants of $G$}

Consider the lattice of invariants of $G$ defined by
$$M:=\{(m_0,m_1,m_2)\in\mathbb{Z}^3\mid m_0F_g^{(0)}+m_1F_g^{(1)}+m_2F_g^{(2)}\in\mathbb{Z},\forall g\in G\},$$ and let $$\tau:=\{(x_0,x_1,x_2)\in\mathbb{R}^3\mid x_i\geqslant 0, i=0,1,2\}.$$ Then $M\cap\tau$ is a semigroup, closed under addition. Its semigroup ring $\mathbb{C}[M\cap\tau]$ consists of elements of the form $\sum\limits_{i=1}^ma_ie^{\gamma_i}$ with $a_i\in\mathbb{C}$ and $\gamma_i\in M\cap\tau$. The ring structure is given by
\begin{eqnarray*}
\sum\limits_{i=1}^ma_ie^{\gamma_i}+\sum\limits_{i=1}^mb_ie^{\gamma_i}&=&\sum\limits_{i=1}^m(a_i+b_i)e^{\gamma_i},\\
e^{\gamma_i}\cdot e^{\gamma_i'}&=&e^{\gamma_i+\gamma_i'}.
\end{eqnarray*}
Then we have \cite{MOP}: $$\mathbb{C}^3/G\cong\textrm{Spec}\mathbb{C}[M\cap\tau].$$

Let $\bZ_n(a,b,c)$ denote the subgroup generated by the diagonal matrix
$\diag(e^{2\pi\sqrt{-1}\frac{a}{n}}, e^{2\pi\sqrt{-1}\frac{b}{n}},e^{2\pi\sqrt{-1}\frac{c}{n}})$.

\begin{example}
For $G= \bZ_3(1,1,1)$, let $\xi$ be the generator, and we have
\be
F_{\xi}^{(0)} = F_{\xi}^{(1)} = F_{\xi}^{(2)} = \frac{1}{3}.
\ee
Therefore,
\be
M = \{(m_0,m_1,m_2) \in \bZ^3\;|\; m_0+m_1+m_2 \equiv 0 \pmod{3}\}.
\ee
\end{example}

\begin{example}
For $G= \bZ_4(2,1,1)$, let $\xi$ be the generator, and we have
\be
F_{\xi}^{(0)} =  \frac{1}{2}, \;\;\; F_{\xi}^{(1)} = F_{\xi}^{(2)} = \frac{1}{4}.
\ee
Therefore,
\be
M = \{(m_0,m_1,m_2) \in \bZ^3\;|\; 2m_0+m_1+m_2 \equiv 0 \pmod{4}\}.
\ee
\end{example}

\begin{example}
For $G= \bZ_5(3,1,1)$, let $\xi$ be the generator, and we have
\be
F_{\xi}^{(0)} =  \frac{3}{5}, \;\;\; F_{\xi}^{(1)} = F_{\xi}^{(2)} = \frac{1}{5}.
\ee
Therefore,
\be
M = \{(m_0,m_1,m_2) \in \bZ^3\;|\; 3m_0+m_1+m_2 \equiv 0 \pmod{5}\}.
\ee
\end{example}

\begin{example}
For $G= \bZ_2(1,0,1) \times \bZ_2(1,1,0)$, let $\alpha_1,\alpha_2$ be the generators of the first and the second subgroup respectively, and we have
\begin{align}
F_{\alpha_1}^{(0)} & =  \frac{1}{2}, & F_{\alpha_1}^{(1)} & = 0, & F_{\alpha_1}^{(2)} & = \frac{1}{2}, \\
F_{\alpha_2}^{(0)} & =  \frac{1}{2}, & F_{\alpha_2}^{(1)} & = \frac{1}{2}, & F_{\alpha_2}^{(2)} & = 0.
\end{align}
Therefore,
\be
M = \{(m_0,m_1,m_2) \in \bZ^3\;|\; m_0+m_1\equiv 0 \pmod{2}, \;\; m_0+m_2\equiv 0 \pmod{2}\}.
\ee
\end{example}

\subsection{An explicit integral basis of $M$}

Let $\{e_0,e_1,e_2\}$ be the standard basis of $\mathbb{Z}^3$. To determine the charge vectors of the crepant resolutions of $\mathbb C^3/G$, the first step is to find a basis of $M$. We construct a basis $\{\varepsilon_0,\varepsilon_1,\varepsilon_2\}$ of $M$ explicitly as follows.

Since the integers $\frac{|G_1|}{\gcd\big(|G_1|,c|G_2|\big)}$ and $c\vert G_2|$ are relatively prime,
one can find $m^*_1,m^*_2\in\mathbb{Z}$ such that
\begin{eqnarray}
m^*_1\cdot\frac{|G_1|}{\gcd\big(\vert G_1\vert,c\vert G_2\vert\big)}+m_2^*\cdot c\vert G_2\vert=1.\label{3}\end{eqnarray}

\begin{prop}
Define three vectors $\varepsilon_0,\varepsilon_1,\varepsilon_2$ by:
\begin{eqnarray}
[\varepsilon_0,\varepsilon_1,\varepsilon_2]:=[e_0,e_1,e_2]\left[\begin{array}{ccc}0&0&1\\\\\frac{m^*_1\cdot\vert G\vert}{\textrm{gcd}\big(\vert G_1\vert,c\vert G_2\vert\big)}&-c\vert G_2\vert&1\\\\m^*_2\cdot\vert G\vert&\vert G_1\vert&1\end{array}\right]\label{4}
\end{eqnarray}
Then $\{\varepsilon_0,\varepsilon_1,\varepsilon_2\}$ is an integral basis of $M$.
\end{prop}

\begin{proof}
First it is clear that $\varepsilon_0,\varepsilon_1,\varepsilon_2$ are linearly independent,
because by using \eqref{3},\eqref{4},
we have
\begin{eqnarray*}
\det \left[\begin{array}{ccc}0&0&1
\\\\ \frac{m^*_1\cdot|G|}{\gcd\big(|G_1|,c|G_2|\big)}&-c|G_2|&1\\\\
m^*_2\cdot |G| &\vert G_1\vert&1\end{array}\right]=\vert G\vert>0.\label{6}
\end{eqnarray*}

Now let us  show that $\varepsilon_0,\varepsilon_1,\varepsilon_2\in M$.
It is clear that $\varepsilon_2 =e_0 + e_1  +e_2 \in M$,
since $G\subset \SL(3,\mathbb{C})$.
For $\varepsilon_0,\varepsilon_1$,
by using $(12),(13)$,
we have for any $(p_1,p_2)\in\mathbb{Z}^2$
and $0 \leq k < |G/G_1|$, $0 \leq l < |G_1|$,
\begin{eqnarray*}
&&p_1\cdot F^{(1)}_{k\alpha + l \beta}
 +p_2\cdot F^{(2)}_{k \alpha + l \beta}\\
&=&p_1\cdot\frac{k}{\vert G/G_1\vert}+p_2\cdot\bigg(k\cdot\frac{c}{\vert G/G_2\vert}
+\frac{l}{\vert G_1\vert}\bigg)\mod{\mathbb{Z}}\\
&=&\frac{k}{\vert G\vert}\bigg(p_1\cdot\vert G_1\vert+p_2\cdot c\vert G_2\vert\bigg)
+l\cdot\frac{p_2}{\vert G_1\vert}\mod{\mathbb{Z}}.
\end{eqnarray*}
In particular, for $(p_1,p_2)=(-c|G_2|, |G_1|)$ so that
$p_1e_1+p_2e_2=\varepsilon_1$,
$$-c|G_2|\cdot F^{(1)}_{k\alpha+l\beta}
+|G_1|\cdot F^{(2)}_{k\alpha+l\beta} = l \in\mathbb{Z}.$$
Therefore, $\varepsilon_1\in M$.
Similarly,
for $(p_1,p_2)=\biggl(\frac{m^*_1\cdot|G|}{\gcd\big(|G_1|,c|G_2|\big)},
m^*_2\cdot |G| \biggr)$ so that
$p_1e_1+p_2e_2=\varepsilon_0$, we have, from \eqref{3},\eqref{4},
\ben
&&\frac{m^*_1\cdot|G|}{\gcd\big(|G_1|,c|G_2|\big)}
\cdot F^{(1)}_{k\alpha+l\beta}
+ m^*_2\cdot |G| \cdot F^{(2)}_{k\alpha+l\beta} \\
& = & k\biggl(m_1^*\cdot \frac{|G_1|}{\gcd(|G_1|, c|G_2|)}
+ m_2^* \cdot c|G_2| \biggr) + lm_2^* \cdot |G/G_1| \in\mathbb{Z}.
\een
Therefore, $\varepsilon_0 \in M$.

Finally, let us prove that $\{\varepsilon_0,\varepsilon_1,\varepsilon_2\}$ generates $M$ over $\mathbb{Z}$.
For any $(n_0,n_1,n_2)\in M$,
we have $\sum\limits_{j=0}^2n_je_j=\sum\limits_{j=0}^2x_j\varepsilon_j$,
where
  \begin{eqnarray*}
  \left[\begin{array}{c}x_0\\x_1\\x_2\end{array}\right]&=&\left[\begin{array}{ccc}0&0&1\\\\\frac{m^*_1\cdot\vert G\vert}{\textrm{gcd}\big(\vert G_1\vert,c\vert G_2\vert\big)}&-c\vert G_2\vert&1\\\\m^*_2\cdot\vert G\vert&\vert G_1\vert&1\end{array}\right]^{-1}\left[\begin{array}{c}n_0\\n_1\\n_2\end{array}\right]\\
&=& \left[\begin{array}{c} (n_1-n_0)\cdot\frac{1}{\vert G/G_1\vert}+(n_2-n_0)\cdot\frac{c}{\vert G/G_2\vert}\\\\-(n_1-n_0)\cdot m^*_2+(n_2-n_0)\frac{m^*_1}{\textrm{gcd}\big(\vert G_1\vert,c\vert G_2\vert\big)}\\\\n_0\end{array}\right].
  \end{eqnarray*}
We need to show that $x_0,x_1,x_2 \in \bZ$.
We already have $x_2 = n_0 \in \bZ$.
Note that
\begin{eqnarray*}
&& \sum_{i=0}^2 n_i \cdot F^{(0)}_{k\alpha+l\beta}
- n_0 \cdot a(k\alpha+l\beta) \\
&=&(n_1-n_0)\cdot F^{(1)}_{k\alpha+l\beta}
+(n_2-n_0)\cdot F^{(2)}_{k\alpha+l\beta} \\
&=&k\bigg((n_1-n_0)\cdot\frac{1}{|G/G_1|}
+(n_2-n_0)\cdot\frac{c}{|G/G_2|}\bigg)
+l\cdot\frac{n_2-n_0}{\vert G_1\vert}
\end{eqnarray*}
is an integer for any $(k,l)\in\bZ^2$.
Choose $(k,l)=(1,0)$, we see that
$$x_0 =(n_1-n_0)\cdot\frac{1}{|G/G_1|}
+(n_2-n_0)\cdot\frac{c}{|G/G_2|} \in\mathbb{Z}.$$
If $\vert G_1\vert=1$, then it is clear that
$$x_1 = -(n_1-n_0)\cdot m^*_2+(n_2-n_0)\cdot m^*_1\in\mathbb{Z};$$
otherwise,
we may choose $(k,l)=(0,1)$ to see that $\frac{n_2-n_0}{\vert G_1\vert}\in\mathbb{Z}$, this implies that
$$x_1=-(n_1-n_0)\cdot m^*_2+m^*_1\frac{(n_2-n_0)}{\gcd(|G_1|,c|G_2|)} \in\mathbb{Z}.$$
\end{proof}

\begin{example}
For $G= \bZ_3(1,1,1)$,
we have $G_0=G_1=G_2 = \{0\}$.
In this case $c = 1$.
We take $m_1^* = 0$, $m_2^* = 1$ to get:
\ben
[\varepsilon_0,\varepsilon_1,\varepsilon_2]
=[e_0,e_1,e_2]
\begin{bmatrix} 0&0&1 \\ 0 &-1&1 \\ 3 & 1 &1
\end{bmatrix}
\een
\end{example}

\begin{example}
For $G= \bZ_4(2,1,1)$,
we have $G_0 \cong \bZ_2$, $G_1 = G_2 = \{0\}$.
In this case $c = 1$.
We take $m_1^* = 0$, $m_2^* = 1$ to get:
\ben
[\varepsilon_0,\varepsilon_1,\varepsilon_2]
=[e_0,e_1,e_2]
\begin{bmatrix} 0&0&1 \\ 0 &-1&1\\ 4 & 1 &1
\end{bmatrix}
\een
\end{example}

\begin{example}
For $G= \bZ_5(3,1,1)$,
we have $G_0=G_1=G_2 = \{0\}$.
In this case $c = 1$.
We take $m_1^* = 0$, $m_2^* = 1$ to get:
\ben
[\varepsilon_0,\varepsilon_1,\varepsilon_2]
=[e_0,e_1,e_2]
\begin{bmatrix} 0&0&1 \\ 0 &-1&1 \\ 5 & 1 & 1
\end{bmatrix}
\een
\end{example}

\begin{example}
For $G= \bZ_2(1,0,1) \times \bZ_2(1,1,0)$,
we have $G_0 = \cor{ \alpha_1+\alpha_2}$,
$G_1 = \cor{\alpha_1}$, $G_2 = \cor{\alpha_2}$.
We take $\alpha = \alpha_2$, $\beta = \alpha_1$, and then $c=1$.
We take $m_1^* =1$ and $m_2^* =0$
to get:
\ben
[\varepsilon_0,\varepsilon_1,\varepsilon_2]
=[e_0,e_1,e_2]
\begin{bmatrix} 0 & 0 & 1 \\ 2 & -2 & 1 \\ 0 & 2 & 1
\end{bmatrix}
\een
\end{example}

\subsection{Toric resolutions of $\mathbb{C}^3/G$}

In order to find toric resolutions of $\mathbb{C}^3/G$, it is necessary to consider the dual picture.
Following \cite{MOP}, write $V=\mathbb{R}^3$. Let $V^\vee$ be the dual space of $V$, $M^\vee$ the dual lattice of $M$ and $\tau^\vee$ the dual cone of $\tau$. Then $M^\vee\subset V^\vee$, and for the basis \eqref{4}, we have in the dual picture
\begin{eqnarray}
[e^\vee_0,e^\vee_1,e^\vee_2]=[\varepsilon^\vee_0,\varepsilon^\vee_1,\varepsilon^\vee_2]\left[\begin{array}{ccc}0&\frac{m^*_1\cdot\vert G\vert}{\textrm{gcd}\big(\vert G_1\vert,c\vert G_2\vert\big)}&m^*_2\cdot\vert G\vert\\\\0&-c\vert G_2\vert&\vert G_1\vert\\\\1&1&1\end{array}\right].\label{7}
\end{eqnarray}
Now the fan describing the affine toric variety $\mathbb{C}^3/G$ is given by the lattice cone $\sigma=\tau^\vee\cap M^\vee$, which is generated by $\{e^\vee_0,e^\vee_1,e^\vee_2\}$ in $M^\vee$. A crepant resolution of $\mathbb{C}^3/G$ can be described by a fan obtained by subdividing $\sigma$ into basic cones by adding more edges lying in $M^\vee$. Recall that a cone is basic if it is spanned by a basis of $M^\vee$. For more details and discussions, see \cite{MOP} and references therein.

From \eqref{7}, we can describe the toric affine variety $\mathbb{C}^3/G$ by a cone in $\mathbb{Z}^3$ given by
\begin{eqnarray} \label{eqn:Cone}
&& \tilde{v}_0:=\left[\begin{array}{c}0\\\\0\\\\1\end{array}\right],\tilde{v}_1:=\left[\begin{array}{c}\frac{m^*_1\cdot\vert G\vert}{\textrm{gcd}\big(\vert G_1\vert,c\vert G_2\vert\big)}\\\\-c\vert G_2\vert\\\\1\end{array}\right],\tilde{v}_2:=\left[\begin{array}{c}m^*_2\cdot\vert G\vert\\\\\vert G_1\vert\\\\1\end{array}\right].
\end{eqnarray}
Let $v_i$ be the integral point in $\mathbb{Z}^2$ by forgetting the last coordinate of $\tilde{v}_i$.
Then $\mathbb{C}^3/G$ can be described by the triangle $\triangle_{v_0v_1v_2}$ with vertices $v_0,v_1,v_2$.
A crepant resolution of $\mathbb{C}^3/G$ is given by a maximal triangulation of $\triangle_{v_0v_1v_2}$
with vertices being the integral points contained in $\triangle_{v_0v_1v_2}$.
So we first need to find all the integral points.
The following method was conjectured in \cite{Z11} by the second author.

Define the ``small'' part of the group as the set
$$G_{s}:=\{g\in G\mid a(g)=1\}\neq\emptyset.$$
We set $s:=\vert G_{s}\vert$ and impose an ordering $G_{s}=\{g_1,\cdots,g_s\}$.
For each $g\in G_s$, set 
\begin{eqnarray} 
\tilde{v}_g:=F_g^{(0)}\cdot \tilde{v}_0+F_g^{(1)}\cdot \tilde{v}_1
+F_g^{(2)}\cdot \tilde{v}_2. \label{8}
\end{eqnarray} 
Then we can check that $\tilde{v}_g$ is an integral point in $\mathbb{Z}^3$. Let $v_g$ be the integral point in $\mathbb{Z}^2$ by forgetting the last coordinate of $\tilde{v}_g$. We can see that $v_g$ is contained in $\triangle_{v_0v_1v_2}$.
Note that $G$ acts effectively on $\mathbb{C}^3$, and so it is easy to see that $v_0,v_1,v_2,v_{g_1},\cdots,v_{g_s}$ are mutually distinct.

\begin{prop}\label{prop:IntegralPoints}
The triangle $\triangle_{v_0v_1v_2}$ contains no integral points other than $v_0,v_1,v_2,v_{g_1},\cdots,v_{g_s}$.
\end{prop}

\begin{proof}
First we give a partition of $G\setminus\{0\}$ as follows:
\begin{eqnarray*}
I_1&:=&\{g\in G\setminus\{0\}\mid a(g)=1\textrm{ and }F_g^{(0)}\cdot F_g^{(1)}\cdot F_g^{(2)}\neq 0\},\\
I_2&:=&\{g\in G\setminus\{0\}\mid a(g)=2\},\\
B&:=&\{g\in G\setminus\{0\}\mid a(g)=1\textrm{ and }F_g^{(0)}\cdot F_g^{(1)}\cdot F_g^{(2)}=0\}.
\end{eqnarray*}
Then we can check the following:
\begin{itemize}
  \item $G\setminus\{0\}=I_1\cup I_2\cup B$, and for any $g\in G\setminus\{0\}$, we have $$g\in I_1\Longleftrightarrow -g\in I_2.$$ Therefore, $\vert I_1\vert=\vert I_2\vert$.
  \item $G_s=I_1\cup B$ and for any $g\in G_s$, we have
  \begin{eqnarray*}
  g\in I_1&\Longleftrightarrow&v_g\in\textrm{int}\triangle_{v_0v_1v_2},\\
  g\in B&\Longleftrightarrow&v_g\in\partial\triangle_{v_0v_1v_2}.
  \end{eqnarray*}
\end{itemize}
Recall by Pick's theorem \cite{MT},
for a triangle $\Delta$ with vertices in $\bZ^2$,
\be
\textrm{area of }\Delta = 1+ i + \frac{b}{2},
\ee
where $i$ is the number of lattice points in the interior of $\Delta$,
and $b$ is the number of lattice points on the boundary of $\Delta$.
Hence  we have
$$\textrm{area of }\triangle_{v_0v_1v_2}\geqslant\frac{3+\vert B\vert}{2}+\vert I_1\vert-1.$$
Note that the left hand side of this inequality $=\frac{\vert G\vert}{2}$ from \eqref{6},
and the right-hand side is equal to
$$\frac{1}{2}\big(3+\vert B\vert+2\vert I_1\vert-2\big)=\frac{1}{2}\big(1+\vert B\vert+\vert I_1\vert+\vert I_2\vert\big)=\frac{\vert G\vert}{2}.$$
So the inequality is actually an equality,
and therefore $\triangle_{v_0v_1v_2}$ contains no integral points other than $v_0,v_1,v_2,v_{g_1},\cdots,v_{g_s}$.
\end{proof}

We shall use $$\mathcal{S}:=\{0,1,2,g_1,\cdots,g_s\}$$ to index the lattice points in $\Delta_{v_0,v_1, v_2}$.
Now one can triangulate $\Delta_{v_0v_1v_2}$ into triangles with these lattice points as vertices whose areas are $1/2$,
i.e.,
there are no other lattice points in these smaller triangles.
This can be easily seen by Pick's theorem and induction.
Indeed,
by Pick's theorem a lattice triangle has area $1/2$ if and only if there is no lattice points other
than the vertices on this triangle.
Hence if a lattice triangle has area bigger than $1/2$,
one can use an extra lattice point as vertex to subdivide the triangle into two or three triangles with smaller areas.

Each of such triangulations corresponds to a subdivision of the lattice cone $\sigma = \tau^\vee \cap M^\vee$
into basic cones, and hence provides a crepant resolution of $\bC^3/G$.
Indeed,
suppose that $\Delta_{v_{i_1}v_{i_2}v_{i_3}}$ is one of the smaller triangles in the triangulation,
then the tetrad with vertices $(0,0,0)$, $\tilde{v}_{i_1}$, $\tilde{v}_{i_2}$ and $\tilde{v}_{i_3}$ has volume
$\frac{1}{6}$,
and so $\tilde{v}_{i_1}$, $\tilde{v}_{i_2}$ and $\tilde{v}_{i_3}$ span a basic cone.
In general, the triangulations are not unique and hence crepant resolutions are not unique,
they are related to each other by {\em flops}.
Combinatorially,
a flop is two different ways to triangulate a parallelogram with integral vertices of area one into two triangles:
$$\xy
(-10, 10); (0,0),**@{-};(10,0), **@{-}; (0, 10), **@{-}; (-10, 10), **@{-}; (10, 0), **@{-};
(30, 10); (20,10),**@{-};(30,0), **@{-}; (40, 0), **@{-}; (30, 10), **@{-}; (30, 0), **@{-};
(15,-10)*+{\text{Figure 1}}; \endxy $$

\begin{example}
For $G= \bZ_3(1,1,1)$,
\ben
[e_0^\vee,e_1^\vee,e_2^\vee]= [\varepsilon_0^\vee,\varepsilon_1^\vee,\varepsilon_2^\vee]
\begin{bmatrix} 0&0&3 \\ 0 &-1&1 \\ 1 & 1 &1
\end{bmatrix}
\een
Only $\xi$ has age $1$,
and
$$\tilde{v}_{\xi} = \begin{bmatrix} 1 \\ 0 \\ 1 \end{bmatrix}.$$
There is only one triangulation of $\Delta_{v_0v_1v_2}$ with $v_{\xi}$ added:
$$\xy
(0, -10); (0,0),**@{-};(10,0), **@{-};
(0, -10), **@{-}; (30, 10), **@{-};(10,0), **@{-};
(0, 0), **@{-}; (30,10), **@{-};
(0,0)*+{\bullet};(10,0)*+{\bullet};(0,-10)*+{\bullet};(30,10)*+{\bullet};
(-3,0)*+{v_0}; (33,10)*+{v_2}; (13,-1)*+{v_{\xi}}; (-3,-10)*+{v_1};
(5,-15)*+{\text{Figure 2}}; \endxy $$
\end{example}

\begin{example}
For $G= \bZ_4(2,1,1)$,
\ben
[e_0^\vee,e_1^\vee,e_2^\vee]= [\varepsilon_0^\vee,\varepsilon_1^\vee,\varepsilon_2^\vee]
\begin{bmatrix} 0&0& 4 \\ 0 &-1&1 \\ 1 & 1 &1
\end{bmatrix}
\een
There are two elements of age $1$: $\xi$ and $2\xi$,
and
\begin{align*}
\tilde{v}_{\xi} & = \begin{bmatrix} 1 \\ 0 \\ 1 \end{bmatrix}, &
\tilde{v}_{2\xi} & = \begin{bmatrix} 2 \\ 0 \\ 1 \end{bmatrix}.
\end{align*}
There is only one triangulation of $\Delta_{v_0v_1v_2}$ with $v_{\xi}$ and $v_{2\xi}$ added:
$$\xy
(0, -10); (0,0),**@{-};(20,0), **@{-};
(0, -10), **@{-}; (40, 10), **@{-};(20,0), **@{-};
(0, 0), **@{-}; (40,10), **@{-};
(0, -10); (10,0),**@{-};(40,10), **@{-};
(0,0)*+{\bullet};(10,0)*+{\bullet};(20,0)*+{\bullet};(0,-10)*+{\bullet};(40,10)*+{\bullet};
(-3,0)*+{v_0}; (43,10)*+{v_2}; (13,-1)*+{v_{\xi}}; (-3,-10)*+{v_1}; (24,-1)*+{v_{2\xi}};
(5,-15)*+{\text{Figure 3}}; \endxy $$

\end{example}

\begin{example}
For $G= \bZ_5(3,1,1)$,
\ben
[e_0^\vee,e_1^\vee,e_2^\vee]= [\varepsilon_0^\vee,\varepsilon_1^\vee,\varepsilon_2^\vee]
\begin{bmatrix} 0&0& 5 \\ 0 &-1&1 \\ 1 & 1 &1
\end{bmatrix}
\een
There are two elements of age $1$: $\xi$ and $2\xi$,
and
\begin{align*}
\tilde{v}_{\xi} & = \begin{bmatrix} 1 \\ 0 \\ 1 \end{bmatrix}, &
\tilde{v}_{2\xi} & = \begin{bmatrix} 2 \\ 0 \\ 1 \end{bmatrix}.
\end{align*}
There is only one triangulation of $\Delta_{v_0v_1v_2}$ with $v_{\xi}$ and $v_{2\xi}$ added:
$$\xy
(0, -10); (0,0),**@{-};(20,0), **@{-};
(0, -10), **@{-}; (50, 10), **@{-};(20,0), **@{-};
(0, 0), **@{-}; (50,10), **@{-};
(0, -10); (10,0),**@{-};(50,10), **@{-};
(0,0)*+{\bullet};(10,0)*+{\bullet};(20,0)*+{\bullet};(0,-10)*+{\bullet};(50,10)*+{\bullet};
(-3,0)*+{v_0}; (53,10)*+{v_2}; (13,-1)*+{v_{\xi}}; (-3,-10)*+{v_1}; (24,-2)*+{v_{2\xi}};
(5,-15)*+{\text{Figure 4}}; \endxy $$
\end{example}

\begin{example}
For $G= \bZ_2(1,0,1) \times \bZ_2(1,1,0)$,
\ben
[e_0^\vee,e_1^\vee,e_2^\vee]= [\varepsilon_0^\vee,\varepsilon_1^\vee,\varepsilon_2^\vee]
\begin{bmatrix} 0 & 2 & 0 \\ 0 &-2& 2 \\ 1 & 1 &1
\end{bmatrix}
\een
There are three elements of age $1$: $\alpha_1$, $\alpha_2$, and $\alpha_1+\alpha_2$,
and
\begin{align*}
\tilde{v}_{\alpha_1} & = \begin{bmatrix} 0 \\ 1 \\ 1 \end{bmatrix}, &
\tilde{v}_{\alpha_2} & = \begin{bmatrix} 1 \\ -1 \\ 1 \end{bmatrix}, &
\tilde{v}_{\alpha_1+\alpha_2} & = \begin{bmatrix} 1 \\ 0 \\ 1 \end{bmatrix}.
\end{align*}
There are four triangulations of $\Delta_{v_0v_1v_2}$ in this case:
\begin{center}
\setlength{\unitlength}{1cm}
\begin{picture}(0,10)(1,-7)
\drawline(0,0)(0,2)(2,-2)(0,0)
\drawline(1,0)(0,1)
\drawline(1,0)(0,0)
\drawline(1,0)(1,-1)
\multiput(0,0)(0,1){3}{\circle*{0.15}}
\multiput(0,0)(1,-1){3}{\circle*{0.15}}
\put(1,0){\circle*{0.15}}
\put(-0.5,0){$v_0$}
\put(2,-1.9){$v_1$}
\put(0.1,2){$v_2$}
\put(-0.6,1){$v_{\alpha_1}$}
\put(0.4,-1){$v_{\alpha_2}$}
\put(1.1,0){$v_{\alpha_1+\alpha_2}$}
\put(1,-1.5){\vector(0,-1){1}} \put(1,-2.5){\vector(0,1){1}}

\drawline(0,-4)(0,-2)(2,-6)(0,-4)
\drawline(1,-4)(0,-3)
\drawline(0,-3)(1,-5)
\drawline(1,-4)(1,-5)
\multiput(0,-4)(0,1){3}{\circle*{0.15}}
\multiput(0,-4)(1,-1){3}{\circle*{0.15}}
\put(1,-4){\circle*{0.15}}
\put(1.5,-4){\vector(1,0){1}} \put(2.5,-4){\vector(-1,0){1}}
\put(-1.5,-4){\vector(1,0){1}} \put(-0.5,-4){\vector(-1,0){1}}

\drawline(3,-4)(3,-2)(5,-6)(3,-4)
\drawline(4,-4)(3,-3)
\drawline(4,-5)(3,-3)
\drawline(3,-3)(5,-6)
\multiput(3,-4)(0,1){3}{\circle*{0.15}}
\multiput(3,-4)(1,-1){3}{\circle*{0.15}}
\put(4,-4){\circle*{0.15}}

\drawline(-3,-4)(-3,-2)(-1,-6)(-3,-4)
\drawline(-2,-5)(-3,-2)
\drawline(-2,-5)(-3,-3)
\drawline(-2,-4)(-2,-5)
\multiput(-3,-4)(0,1){3}{\circle*{0.15}}
\multiput(-3,-4)(1,-1){3}{\circle*{0.15}}
\put(-2,-4){\circle*{0.15}} 
\put(-0.5,-6.5){Figure 5}
\end{picture}
\end{center}

\end{example}

\subsection{GLSM for toric crepant resolutions of $\mathbb{C}^3/G$}

In the above subsection we have described the fans that correspond to  
toric crepant resolutions of $\bC^3/G$.
It is well-known that one can describe toric manifolds by symplectic reduction \cite{CK},
and this corresponds to moduli spaces of gauged linear sigma models (GLSM) in physics literature \cite{W}. 
We now recall this in this subsection.
Consider the following $(S^1)^k$-action on $\mathbb{C}^{n+k}$: 
$$(t_1,\dots,t_k).(x_1,\cdots,x_{n+k})=(\prod\limits_{i=1}^k t_i^{l^{(i)}_1}\cdot x_1,\dots,
\prod\limits_{i=1}^k t_i^{l^{(i)}_{n+k}}\cdot x_{n+k}),$$ 
where $l^{(i)}_j\in\mathbb{Z}$. 
This action is Hamiltonian with respect to the standard K\"{a}hler form of $\mathbb{C}^{n+k}$. The vectors $$l^{(i)}=(l^{(i)}_1,\cdots,l^{(i)}_{n+k}),\quad i=1,\cdots,k$$ are called charge vectors of this action. They determine a moment map
\begin{eqnarray*}
\mu:\mathbb{C}^{n+k}&\rightarrow&\mathbb{R}^k\\
(x_1,\dots,x_{n+k})&\mapsto&(\frac{1}{2}\sum\limits_{j=1}^{n+k}l^{(i)}_j\vert x_j\vert^2,\dots,\frac{1}{2}\sum\limits_{j=1}^{n+k}l^{(k)}_j\vert x_j\vert^2).
\end{eqnarray*}
For $\vec{a}\in\mathbb{R}^k$, let $X_{\vec{a}}:=\mu^{-1}(\vec{a})/(S^1 )^k$. Then $X_{\vec{a}}$ is an orbifold for generic $\vec{a}$. 
More precisely, the regular values of $\mu$ form some open chambers. It is known that when $\vec{a}$ moves inside one chamber, 
the topology of $X_{\vec{a}}$ remains unchange, and when $\vec{a}$ crosses the wall from one chamber to another, $X_{\vec{a}}$ changes by a flop. 
If the charge vectors satisfy the Calabi-Yau condition 
$$\sum\limits_{j=1}^{n+k}l^{(i)}_j=0,\quad i=1,\cdots,k,$$ 
then generic $X_{\vec{a}}$'s are Calabi-Yau $n$-orbifolds.

We now recall  a procedure of finding charge vectors of a toric manifold  given in \cite{dOFS}.
Suppose that $X$ is a smooth toric $n$-fold described by a fan $\Sigma$,
whose one-dimensional cones are generated by $v_1, \dots,v_m $ in a lattice $\bM \subset \bR^n$, 
$n<m$. 
Suppose that $v_1, \dots, v_m$ also generate $\bZ^n$ over $\bZ$.
The lattice of relations over $\bZ$ among $v_1, \dots, v_m$ is defined by
$$\bL:=\{(l_1, \dots, l_m) \in \bZ^m: \; l_1v_1 + \cdots + l_m v_m = 0\}.$$
I.e.,
we have the following exact sequence
\be
0 \to \bL \to \bZ^m \to \bM \to 0.
\ee
The charge vectors correspond to a set of generators of $\bL$. 

Let us first construct some vectors in $\bL$.
The vectors $v_1, \dots, v_m$ correspond to toric divisors $D_1, \dots, D_m$, respectively.
Given any $w \in \bM^\vee$,
there is a linear equivalence:
\be \label{eqn:LinEquiv}
\langle v_1,w\rangle D_1 + \cdots + \langle v_m, w \rangle D_m \sim 0.
\ee
Toric curves on $X$ are determined by $(n-1)$-dimensional cones in $\Sigma$.
If a toric curve $C$ is compact, the corresponding cone $\tau$ is the
boundary between two $n$-dimensional cones $\sigma_1$, $\sigma_2$. 
Denote the integer generators
of $\tau$, $\sigma_1$ and $\sigma_2$
by $\{v_{i_1}, . . . , v_{i_{n-1}}\}$, $\{v_{i_1}, . . . , v_{i_{n-1}}, v_{i_n}\}$, 
$\{v_{i_1}, . . . , v_{i_{n-1}}, v_{i_{n+1}}\}$, respectively.
Because $X$ is smooth,
if we write $v_i$ as column vectors then we have
\ben
&& \det (v_{i_1}, . . . , v_{i_{n-1}}, v_{i_n}) = - \det (v_{i_1}, . . . , v_{i_{n-1}}, v_{i_{n+1}}) = \pm 1,
\een
and so 
$$
\det (v_{i_1}, . . . , v_{i_{n-1}}, v_{i_n}+v_{i_{n+1}}) = 0.
$$
So there exists a unique linear relation of the form 
$$l_{i_1}v_{i_1} + \cdots + l_{i_{n+1}}v_{i_{n+1}} = 0$$ 
with $l_{i_n} = l_{i_{n+1}} = 1$ and all $l_{i_j}$ integer.
In \cite{dOFS} the following geometric interpretation of the numbers $l_i$ is given.
The toric curve $C$ corresponding to $\tau$ is given by 
$C = D_{i_1} \cdots D_{i_{n-1}}$.
Then
$$l_{i_j} = C D_{i_j} = D_{i_1} \cdots D_{i_{n-1}} D_{i_j}.$$
These intersection numbers can be computed as follows:
The intersection numbers between $n$ distinct toric divisors $D_{j_1}, \dots, D_{j_n}$
is $1$ if their corresponding vertices $v_{j_1}, \dots, v_{j_n}$ generate a cone in $\Sigma$, and $0$ otherwise.
Hence one can see from this that when $j=n$ or $n+1$,
$l_{i_j} = 1$.
When $1 \leq j \leq n-1$,
say $j=1$,
choose a lattice point $w \in \bM^\vee$ such that
$\langle v_{i_1}, w\rangle = \delta_{j1}$,
then by \eqref{eqn:LinEquiv} one has
$$
D_{i_1} \cdots D_{i_{n-1}} D_{i_1}
= - \sum_{i \neq i_1, \dots, i_{n-1}} \langle v_i, w \rangle \cdot 
D_{i_1} \cdots D_{i_{n-1}} D_i,
$$
where $D_{i_1} \cdots D_{i_{n-1}} D_i$ is either $1$ or $0$,
depending on whether $v_{i_1}, \dots, v_{i_{n-1}}$ and $v_i$ 
generate a cone in $\Sigma$ or not.
Thus we have described the way to associate to a compact toric divisor $C$
a vector $\vec{l}_C \in \bL$.
To find the charge matrix,
take all those curves $C_i$ whose
$\vec{l}^{(i)}:=\vec{l}_{C_i}$ cannot be written as nonnegative linear combinations 
of the other $\vec{l}^{(j)}$.

Applying this procedure to the case of toric crepant resolutions of $\bC^3/G$,
one can obtain the corresponding charge vectors.
In this case we propose the following procedure to find the $C_i$'s.
Consider the dual graph of the regular triangulation of $\Delta_{v_0v_1v_2}$.
It consists of edges (of finite length) and half edges (that go to infinity).
The edges correspond to compact toric curves,
and the half edges correspond to noncompact toric curves.
For $i \in G_s\subset\cS$,
let $\overline{v_iv_{j_1}}, \dots, \overline{v_iv_{j_k}}$ be
all the edges in the triangulation of $\Delta_{v_0v_1v_2}$ emanating at $v_i$. 
The star of $v_i$ consists of those edges in the dual graph, 
one for each of $\overline{v_iv_{j_l}}$.
If $v_i$ lies in the interior of $\Delta_{v_0v_1v_2}$,
then the star of $v_i$ is a closed polygon;
if $v_i$ lies on the boundary of $\Delta_{v_0v_1v_2}$,
then the star of $v_i$ is an ``open polygon" with two parallel sides.
For each $i \in G_s$,
we will choose one edge in the star of $v_i$,
to obtain a collection of the corresponding curves $C_i$,
so that $\{[C_i]:\;i \in G_s\}$ form a basis of $H_2(Y;\bZ)$.

\begin{example} 
For the toric resolution of $\mathbb{C}^3/\mathbb{Z}_3(1,1,1)$, recall that the triangulation of $\Delta_{v_0v_1v_2}$ is 
$$\xy
(0, -10); (0,0),**@{-};(10,0), **@{-};
(0, -10), **@{-}; (30, 10), **@{-};(10,0), **@{-};
(0, 0), **@{-}; (30,10), **@{-};
(0,0)*+{\bullet};(10,0)*+{\bullet};(0,-10)*+{\bullet};(30,10)*+{\bullet};
(-3,0)*+{v_0}; (33,10)*+{v_2}; (13,-1)*+{v_{\xi}}; (-3,-10)*+{v_1};
(5,-15)*+{\text{Figure 7}}; \endxy $$
and its dual graph is
\begin{center}
\setlength{\unitlength}{1cm}
\begin{picture}(0,6.5)(0,-3.5)
\drawline(-1,0)(0,0)(1,-1)(0,1)(0,0)
\drawline(-0.5,2.5)(0,1)
\drawline(1.5,-1.75)(1,-1)
\dottedline[.]{0.2}(-2,0)(-1,0)
\dottedline[.]{0.2}(-0.5,2.5)(-0.75,3.25)
\dottedline[.]{0.2}(1.5,-1.75)(2,-2.5)
\put(-1,0.5){$D_0$}
\put(-1,-0.5){$D_1$}
\put(1,0){$D_2$}
\put(0,0){$D_\xi$}
\put(-0.5,-3){Figure 8}
\end{picture}
\end{center}
Let $D_i$ be the corresponding toric divisor for $i\in\{0,1,2,\xi\}$. 
Then from $(30)$,
we have the following linear equivalence relations:
\be
D_0 - D_1 \sim 0, \;\; D_0 - D_2 \sim 0, \;\;
D_0 + D_1 + D_2 + D_\xi \sim 0.
\ee 
So we have
\be
D_0 \sim D_1 \sim D_2, \;\; D_\xi \sim - 3 D_0 \sim -3 D_1 \sim -3 D_2.
\ee
From the two-dimensional profile above,
we see there are three compact toric curves: $D_\xi \cdot D_i$, $i=0,1,2$.
The intersection numbers are:
\begin{displaymath}
\begin{array}{c|cccc}
&D_0&D_1&D_2&D_{\xi}\\
\hline
D_\xi\cdot D_0&1&1&1&-3\\
D_\xi\cdot D_1&1&1&1&-3\\
D_\xi\cdot D_2&1&1&1&-3
\end{array}
\end{displaymath}
For example,
\ben
(D_\xi \cdot D_0) \cdot D_1 = (D_\xi \cdot D_0) \cdot D_2 = 1
\een
because $\{v_\xi, v_0, v_1\}$ and $\{v_\xi, v_0, v_2\}$ are the the vertices
of two triangles in the triangulation.
Using the linear equivalence as above,
we also have
\ben
&& (D_\xi \cdot D_0) \cdot D_0 = (D_\xi \cdot D_0) \cdot D_1 = 1, \\
&& (D_\xi \cdot D_0) \cdot D_\xi = -3 (D_\xi \cdot D_0) \cdot D_1 = -3.
\een
The star of $v_\xi$ is a triangle.
It corresponds to the fact that $D_\xi \cong \bP^2$.
The three edges of the star of $v_\xi$ corresponds to three coordinate lines
on $\bP^2$,
and they have the same homology classes.
Choose any one of them as  $C_\xi$,
one get the charge vector is 
$$l^{(\xi)}=(C_\xi\cdot D_0, C_\xi \cdot D_1, C_\xi \cdot D_2, C_\xi \cdot D_\xi)
= (1,1,1,-3).$$
\end{example}

For more examples,
see Section \ref{Appendix}.

\subsection{Summary}

In this subsection we summarize the main results of this section.
First of all,
when $G \subset SL(3, \bC)$ is a finite abelian subgroup,
then $\bC^3/G$ is a toric variety described by a cone with generators
$\tilde{v}_0, \tilde{v}_1, \tilde{v}_2$ given by \eqref{eqn:Cone}.
Secondly,
toric crepant resolution of $\bC^3/G$ can be obtained by adding vectors 
$\{\tilde{v}_g: \;g\in G, a(g) = 1\}$ to this cone
and use them to subdivide the cone into simplicial cones.
Thirdly the subdivision determines some charge vectors
which realize the toric crepant resolutions as symplectic quotients
of $T^s$-actions on $\bC^{3+s}$.

\section{Disc potential for crepant resolutions of $\mathbb{C}^3/G$}

In this section we compute the potential function of disc invariants of 
crepant resolutions of $\bC^3/G$ with some special D-branes.
Our main results of this section are contained Proposition \ref{prop:OCMirror} and
\S \ref{sec:SuperPot}. 

\subsection{Charge vectors for of some special D-branes}

Let $Y$ be a crepant resolution of $\mathbb{C}^3/G$.
In the last section we have shown that they can be described as symplectic quotients with charge vectors
\begin{eqnarray*}
l^{(g_1)}&=&\big(l_0^{(g_1)},l_1^{(g_1)},l_2^{(g_1)},l_{g_1}^{(g_1)},\cdots,l_{g_s}^{(g_1)}\big),\\
&\cdots&\\
l^{(g_s)}&=&\big(l_0^{(g_s)},l_1^{(g_s)},l_2^{(g_s)},l_{g_1}^{(g_s)},\cdots,l_{g_s}^{(g_s)}\big).
\end{eqnarray*}
They form a basis of the lattice of relations \cite{S}:$$\mathbb{L}:=\{(l_0,l_1,l_2,l_{g_1},\cdots,l_{g_s})\in\mathbb{Z}^{3+s}\mid \sum\limits_{j=0}^2l_j\tilde{v}_j+\sum\limits_{i=1}^sl_{g_i}\tilde{v}_{g_i}=0\}.$$ From \eqref{8}, the following vectors form a natural $\mathbb{Q}$-basis of $\mathbb{L}\otimes_\mathbb{Z}\mathbb{Q}$:
\begin{eqnarray}
\begin{array}{ccccccc}
(\vert G\vert F_{g_1}^{(0)},&\vert G\vert F_{g_1}^{(1)},&\vert G\vert F_{g_1}^{(2)},&-\vert G\vert,&0,&\cdots,&0),\nonumber\\
(\vert G\vert F_{g_2}^{(0)},&\vert G\vert F_{g_2}^{(1)},&\vert G\vert  F_{g_2}^{(2)},&0,&-\vert G\vert,&\cdots,&0),\nonumber\\
\cdots&\cdots&&&&&\label{9}\\
(\vert G\vert F_{g_s}^{(0)},&\vert G\vert F_{g_s}^{(1)},&\vert G\vert  F_{g_s}^{(2)},&0,&0,&\cdots,&-\vert G\vert).\nonumber
\end{array}
\end{eqnarray}

For $i\in\mathcal{S}=\{0,1,2,g_1,\cdots,g_s\}$,
we shall write $l_i=\big(l_i^{(g_1)},\cdots,l_i^{(g_s)}\big)^T$.
Since the charge vectors are linear combinations of those in \eqref{9},
it follows that the vectors $l_{g_l},i=1,\cdots,s,$ form a basis of $\mathbb{R}^s$, and
\begin{eqnarray}
\sum\limits_{i\in\mathcal{S}}l_i^{(g_l)}=0,\quad l=1,\cdots,s,&&\label{10'}\\
\sum\limits_{g\in G_s}F_g^{(j)}l_g+l_j=0,\quad j=0,1,2.\label{10}
\end{eqnarray}

Consider an outer D-brane $L_0$ in $Y$, with a framing $f\in\mathbb{Z}$, which intersects the non-compact toric curve given by $\overline{v_{i_1}v_{i_2}}\subset\overline{v_1v_2}$, as illustrated below. Here $v_{i_0}$ is the unique integral point such that $\vartriangle_{v_{i_0}v_{i_1}v_{i_2}}$ is an anticlockwise directed triangle in the triangulation. It is easy to check that $i_0\in\{0,g_1,\cdots,g_s\}, i_1\in\{1,g_1,\cdots,g_s\},i_2\in\{2,g_1,\cdots,g_s\}$.
\begin{center}
\setlength{\unitlength}{0.5cm}
\begin{picture}(10,8)(-7.5,-4.5)
\drawline(-2,0)(0,-1)(1,1)(-2,0)
\drawline(-1,-3)(-7,0)(2,3)
\dottedline{0.1}(1,1)(2,3)
\dottedline{0.1}(0,-1)(-1,-3)
\put(-7,0){\circle*{0.1}}
\put(-2,0){\circle*{0.1}}
\put(-1,-3){\circle*{0.1}}
\put(0,-1){\circle*{0.1}}
\put(1,1){\circle*{0.1}}
\put(2,3){\circle*{0.1}}
\put(-2.8,0){$v_{i_0}$}
\put(0,-1.3){$v_{i_1}$}
\put(1.1,1){$v_{i_2}$}
\put(-7.7,0){$v_0$}
\put(-1,-3.3){$v_1$}
\put(2.1,3){$v_2$}
\put(-0.5,0){$\circlearrowleft$}
\put(0.5,0){$\backslash$}
\put(-3,-4.5){Figure 9}
\end{picture}
\end{center}
Then the geometry of $(Y,L_0)$ can be described by the following extended charge vectors:
\begin{displaymath}
\begin{array}{cccccc|cc}
\tilde{l}^{(g_1)}=\big(l_0^{(g_1)},&l_1^{(g_1)},&l_2^{(g_1)},&l_{g_1}^{(g_1)},&\cdots,&l_{g_s}^{(g_1)},&0,&0\big),\\
\cdots&\cdots&&&&&&\\
\tilde{l}^{(g_s)}=\big(l_0^{(g_s)},&l_1^{(g_s)},&l_2^{(g_s)},&l_{g_1}^{(g_s)},&\cdots,&l_{g_s}^{(g_s)},&0,&0\big),\\
\tilde{l}^{(0)}=\big(l^{(0)}_0,&l^{(0)}_1,&l^{(0)}_2,&l^{(0)}_{g_1},&\cdots,&l^{(0)}_{g_s},&1,&-1\big),
\end{array}
\end{displaymath}
where
\begin{eqnarray}
l_{i}^{(0)}=\left\{\begin{array}{cl}1,&i=i_0,\\f,&i=i_1,\\-f-1,&i=i_2,\\0,&i\in\mathcal{S}\setminus\{i_0,i_1,i_2\}.\end{array}\right.\label{11}
\end{eqnarray}

We remark that if $i_0=0$, then $G_0=G$ and $G\subset \SL(2,\mathbb{C})$ acts effectively on the plane given by $z_0=0$. Therefore, we can see that $$l_0^{(0)}=\delta_{i_0,0}=\frac{\vert G_0\vert}{\vert G\vert}\delta_{i_0,0}.$$

\begin{lemma}\label{lemma 2}
If $i_0\neq0$, then $i_0\in G_s$ and $F_{i_0}^{(0)}=\frac{\vert G_0\vert}{\vert G\vert}$.
\end{lemma}
\begin{proof}
It is easy to see that $i_0\in G_s$ with $F^{(0)}_{i_0}\neq 0$. So $i_0\in G\setminus G_0$ and hence $\frac{\vert G\vert}{\vert G_0\vert}>1$. From Lemma \ref{lem}, we can find $h\in G_s$ such that $F^{(0)}_h=\frac{1}{\vert G/G_0\vert}$. Then the affine line $$l:=\{xv_1+(1-\frac{\vert G_0\vert}{\vert G\vert}-x)v_2\mid x\in\mathbb{R}\}$$ contains $v_h$ and is parallel to $\overline{v_1v_2}$.

If $F_{i_0}^{(0)}\neq\frac{1}{\vert G/G_0\vert}$, then we have $F_{i_0}^{(0)}>\frac{1}{\vert G/G_0\vert}$ and therefore $v_{i_0}$ lies on the different side of $l$ from $\overline{v_1v_2}$. So it is clear that $$\textrm{area of }\triangle_{v_{i_0}v_{i_1}v_{i_2}}>\textrm{area of }\triangle_{v_hv_{i_1}v_{i_2}},$$ as illustrated below. This is absurd since the area of the triangle is minimal.

\begin{center}
\setlength{\unitlength}{1cm}
\begin{picture}(3.8,3.5)(-2.5,-1.5)
\drawline(-2,0)(0,-1)(1,1)(-2,0)
\drawline(0.8,1.6)(-0.8,-1.6)
\put(0.9,1.6){$l$}
\put(0,1){$v_{h}$}
\put(1.1,1){$v_{i_2}$}
\put(-2.5,0){$v_{i_0}$}
\put(0,-1.3){$v_{i_1}$}
\put(0.5,1){\circle*{0.1}}
\put(-2,0){\circle*{0.1}}
\put(0,-1){\circle*{0.1}}
\put(1,1){\circle*{0.1}}
\put(0.5,0){$\backslash$}
\put(-0.5,-2){Figure 10}
\end{picture}
\end{center}
\end{proof}

\subsection{Extended Picard-Fuchs system for $(Y,L_0)$}

Let $z_0$, $z_1$, $z_2$, $z_{g_1}$, $\cdots$, $z_{g_s}$, $z_+,z_-$ be some variables and set
\begin{eqnarray}
&&q_{g_l}:=\prod_{i\in\mathcal{S}}z_i^{l^{(g_l)}_i},\quad l=1,\cdots,s,\label{12}\\
&&q_0:=\prod_{i\in\mathcal{S}}z_i^{l^{(0)}_i}\cdot\frac{z_+}{z_-}.\nonumber
\end{eqnarray}
We shall write $\vec{q}=(q_{g_1},\cdots,q_{g_s})$.
Then $\vec{q}$ gives the closed moduli parameters
and $q_0$ gives the open moduli parameter.
In physics literature, $\vec{q}$ is referred to as bulk moduli and $q_0$ is referred to as boundary moduli.  The extended Picard-Fuchs system associated to the extended charge vectors $\tilde{l}^{(0)},\tilde{l}^{(g_1)},\cdots,\tilde{l}^{(g_s)}$ are the following system of partial differential equations:
\begin{eqnarray}
&&\bigg(\prod_{i\in\mathcal{S}:l^{(g_l)}_i>0}\partial_{z_i}^{l^{(g_l)}_i}-\prod_{i\in\mathcal{S}:l^{(g_l)}_i<0}\partial_{z_i}^{-l^{(g_l)}_i}\bigg)\Omega=0,\quad l=1,\cdots,s,\label{13}\\
&&\bigg(\partial_{z_+}\cdot\prod_{i\in\mathcal{S}:l_i^{(0)}>0}\partial_{z_i}^{l_i^{(0)}}-\partial_{z_-}\cdot\prod_{i\in\mathcal{S}:l_i^{(0)}<0}\partial_{z_i}^{-l_i^{(0)}}\bigg)\Omega=0.\label{14}
\end{eqnarray}
Assume that $\Omega=\Omega(q,q_0)$. Then we can rewrite \eqref{13},\eqref{14} as
\begin{eqnarray}
&&\mathcal{D}_l\Omega=0,\quad l=1,\cdots,s,\label{15}\\
&&\mathcal{D}_0\Omega=0,\label{16}
\end{eqnarray}
where
\begin{eqnarray*}
\mathcal{D}_l&=&\prod_{i\in\mathcal{S}:l_i^{(g_l)}>0}\bigg(\sum\limits_{g\in G_s}l_i^{(g)}\Theta_{q_g}+l_i^{(0)}\Theta_{q_0}\bigg)_{l_i^{(g_l)}}\\
&&\qquad-q_{g_l}\cdot\prod_{i\in\mathcal{S}:l_i^{(g_l)}<0}\bigg(\sum\limits_{g\in G_s}l_i^{(g)}\Theta_{q_g}+l_i^{(0)}\Theta_{q_0}\bigg)_{-l_i^{(g_l)}},\quad l=1,\cdots,s,\\
\mathcal{D}_0&=&\Theta_{q_0}\cdot\prod_{i\in\mathcal{S}:l_i^{(0)}>0}\bigg(\sum\limits_{g\in G_s}l_i^{(g)}\Theta_{q_g}+l_i^{(0)}\Theta_{q_0}\bigg)_{l_i^{(0)}}\\
&&\qquad+q_{0}\Theta_{q_0}\cdot\prod_{i\in\mathcal{S}:l_i^{(0)}<0}\bigg(\sum\limits_{g\in G_s}l_i^{(g)}\Theta_{q_g}+l_i^{(0)}\Theta_{q_0}\bigg)_{-l_i^{(0)}}
\end{eqnarray*}
Here we use the logarithmic derivative $\Theta_x=x\partial_x$ and the Pochhammer symbol: $$(x)_n:=\frac{\Gamma(x+1)}{\Gamma(x-n+1)}=x(x-1)(x-2)\cdots(x-n+1),\quad n\in\mathbb{Z}_{\geqslant 0}$$

By Frobenius method,
the fundamental solution to the Picard-Fuchs system \eqref{15},\eqref{16} is:
\be
\Omega(\vec{q},q_0;\vec{r},r_0)
=\sum\limits_{\vec{m},m_0}A_{\vec{m},m_0}(\vec{r},r_0)\cdot \vec{q}^{\vec{m}+\vec{r}}\cdot q_0^{m_0+r_0},
\ee
where
\begin{eqnarray}
&& m_0\in\mathbb{Z}_{\geqslant 0},\quad \vec{m}=(m_{g})_{g\in G_s}\in\mathbb{Z}^s_{\geqslant 0},\nonumber\\
&& r_0\in\mathbb{C},\quad \vec{r}=(r_g)_{g\in G_s}\in\mathbb{C}^s,\nonumber\\
&& \vec{q}^{\vec{m}+\vec{r}}:=\prod\limits_{g\in G_s}q_g^{m_g+r_g},\nonumber
\end{eqnarray}
and
\be \label{eqn:A}
\begin{split}
A_{\vec{m},m_0}(r,r_0)
=&\prod\limits_{i\in\mathcal{S}}
\frac{\Gamma(1+r_0l_i^{(0)}+\langle \vec{r},l_i\rangle)}{\Gamma\Big(1+(r_0+m_0)l_i^{(0)}+\langle \vec{r}+\vec{m},l_i\rangle\Big)} \\
& \cdot\frac{\Gamma(1+r_0)}{\Gamma(1+r_0+m_0)}
\cdot\frac{\Gamma(1-r_0)}{\Gamma(1-r_0-m_0)},
\end{split}
\ee
where $\langle\cdot,\cdot\rangle$ is the standard inner product.
Here we use the following convention:
For $z_1,z_2\in\mathbb{C}$ with $z_1-z_2\in\mathbb{Z}_{\geqslant 0}$, we set
\begin{eqnarray}
\frac{\Gamma(z_1)}{\Gamma(z_2)}:=\lim_{x\rightarrow 0}\frac{\Gamma(x+z_1)}{\Gamma(x+z_2)}.\label{gamma}
\end{eqnarray}
With this convention, we see that for any $z_1,z_2\in\mathbb{Z}$, $\frac{\Gamma(z_1)}{\Gamma(z_2)}$ makes sense when either $z_1\in\mathbb{Z}_{>0}$ or $z_2\in\mathbb{Z}_{\leqslant 0}$. In particular, for $n\in\mathbb{Z}_{\geqslant 0},n'\in\mathbb{Z}$, we have:
\begin{eqnarray}
\frac{1}{\Gamma(-n)}&=&0,\label{n'=0}\\
\frac{\Gamma(1+n')}{\Gamma(-n)}&=&(-1)^{n+n'+1}\frac{\Gamma(1+n)}{\Gamma(-n')}.\label{n'}
\end{eqnarray}

Recall that $\Gamma(z)$ is a meromorphic function in $\mathbb{C}$ with simple poles $z=0,-1,-2,-3,\cdots$,
and takes values in $\mathbb{C}\setminus\{0\}$.
So we can check, with convention \eqref{gamma},
that for each pair $(\vec{m},m_0)$,
the function $A_{\vec{m},m_0}$ is holomorphic around $(\vec{r},r_0)=(\vec{0},0)$.
In particular, $A_{\vec{0},0}$ is a constant function with
\begin{eqnarray}
A_{\vec{0},0}(\vec{r},r_0)=1,\quad\forall (\vec{r},r_0)\in\mathbb{C}^s\times\mathbb{C}.\label{A00}
\end{eqnarray}

\begin{lemma}\label{Frobenius lemma 1}
For $(\vec{m},m_0)\neq( \vec{0},0)$, we have $A_{\vec{m},m_0}(\vec{0},0)=0.$
\end{lemma}
\begin{proof}
 We have the following two cases to consider:
\begin{itemize}
\item When $m_0>0$, we have $$\lim_{r_0\rightarrow 0}\frac{\Gamma(1-r_0)}{\Gamma(1-r_0-m_0)}=0\Rightarrow
A_{\vec{m},m_0}(\vec{0},0)=0.$$
\item When $m_0=0$, we have $\vec{m}\neq \vec{0}$.
If $\langle \vec{m},l_i\rangle=0$ for each $i\in\mathcal{S}$,
then one must have $\vec{m}=0$ since the vectors $l_{g_l},l=1,\cdots,s,$ form a basis of $\mathbb{R}^s$,
which is a contradiction.
So we may find some $i'\in\mathcal{S}$ such that
$\langle \vec{m},l_{i'}\rangle\neq 0$.
Furthermore, note that, from \eqref{10'},
\begin{eqnarray}
\sum\limits_{i\in\mathcal{S}}\langle \vec{m},l_i\rangle
=\langle \vec{m},\sum\limits_{i\in\mathcal{S}}l_i\rangle=0,\nonumber
\end{eqnarray}
and so we may assume that $\langle m,l_{i'}\rangle<0$.
Therefore
$$\lim_{\vec{r} \rightarrow \vec{0}}
\frac{\Gamma(1+\langle \vec{r},l_{i'}\rangle)}{\Gamma\Big(1+\langle \vec{r}+\vec{m},l_{i'}\rangle\Big)}
=0\Rightarrow A_{\vec{m},0}(\vec{0},0)=0.$$
\end{itemize}
\end{proof}

\subsection{Open-closed mirror map}

In this subsection, we derive the open-closed mirror map \eqref{17},\eqref{18} by careful study of the charge vectors. The explicit form of the open-closed mirror map \eqref{19},\eqref{20} is crucial to the determination of the change of variables \eqref{28}.

The open-closed mirror map is given by
\begin{eqnarray}
\Omega_{0}(\vec q,q_0)&=&\frac{\partial\Omega}{\partial r_{0}}(\vec q,q_0;0,0) \label{open-closed open} \\
& = & \log q_{0}+\sum\limits_{\vec{m},m_0\geqslant 0}
\frac{\partial A_{\vec{m},m_0}}{\partial r_{0}}(\vec{0},0)\vec{q}^{\vec{m}}q_0^{m_0}, 
\nonumber \\
\Omega_{g_l}(\vec{q},q_0)
&=&\frac{\partial\Omega}{\partial r_{g_l}}(\vec q,q_0;0,0) \label{open-closed closed} \\
& = & \log q_{g_l}+\sum\limits_{\vec{m},m_0\geqslant 0}
\frac{\partial A_{\vec{m},m_0}}{\partial r_{g_l}}(0,0)\vec{q}^{\vec{m}}q_0^{m_0},\qquad
 \nonumber
\end{eqnarray}
for $l=1,\cdots,s$.

Set $\Psi:=\big(\log\Gamma\big)'$.
Then direct calculations by \eqref{eqn:A} yield:
\begin{eqnarray}
&& \frac{\partial A_{\vec{m},m_0}}{\partial r_{g_l}}(\vec{r},r_0)
=A_{\vec{m},m_0}(\vec{r},r_0) \nonumber\\
&& \cdot \sum\limits_{i\in\mathcal{S}}l_i^{(g_l)}
\bigg[\Psi(1+r_0l_i^{(0)}+\langle \vec{r},l_i\rangle)
-\Psi\Big(1+(r_0+m_0)l_i^{(0)}+\langle \vec{r}+\vec{m},l_i\rangle\Big)\bigg],\nonumber\\&&\qquad\qquad\qquad\qquad\qquad\qquad\qquad\qquad\qquad\qquad\qquad l=1,\cdots,s,\nonumber\\
&& \frac{\partial A_{\vec{m},m_0}}{\partial r_{0}}(\vec{r},r_0)
=A_{\vec{m},m_0}(\vec{r},r_0)\nonumber\\
&& \cdot\Bigg\{\sum\limits_{i\in\mathcal{S}}l_i^{(0)}
\bigg[\Psi(1+r_0l_i^{(0)}+\langle \vec{r},l_i\rangle)
-\Psi\Big(1+(r_0+m_0)l_i^{(0)}
+\langle \vec{r}+\vec{m},l_i\rangle\Big)\bigg] \nonumber \\
&&\quad+\big[\Psi(1+r_0)-\Psi(1+r_0+m_0)\big]-\big[\Psi(1-r_0)-\Psi(1-r_0-m_0)\big]\Bigg\}\nonumber
\end{eqnarray}

Note that $\Psi(z)$ is a meromorphic function on $\mathbb{C}$ with simple poles $z=0,-1,-2,-3,\cdots$.
For convenience, for $n=0,1,2,3,\cdots,$ we write
\begin{eqnarray}
\quad \frac{\Psi}{\Gamma}(-n):=\lim_{x\rightarrow 0}\frac{\Psi}{\Gamma}(x-n).\label{Frobenius convention}
\end{eqnarray}

\begin{lemma}\label{Frobenius lemma 2}
For $n=0,1,2,3,\cdots$, we have
\begin{eqnarray}
\lim\limits_{x\rightarrow 0}\frac{\Psi}{\Gamma}(x-n)&=&(-1)^{n+1}\Gamma(n+1).\nonumber
\end{eqnarray}
\end{lemma}

\begin{proof}
Recall that the Gamma function has the following property:
$$\Gamma(z) = \frac{1}{z} \Gamma(z+1).$$
So we have
\ben
\Psi(x-n)
& = & \frac{d}{dx} \log \Gamma(x-n) \\
& = & \frac{d}{dx} \log \biggl( \frac{1}{x-n} \cdot \frac{1}{x-n+1}
\cdots \frac{1}{x} \cdot \Gamma(x+1) \biggr) \\
& = &  \frac{1}{x-n} + \frac{1}{x-n+1} +
\cdots + \frac{1}{x} + \frac{d}{dx} \log \Gamma(x+1),
\een
and so
\ben
\frac{\Psi}{\Gamma}(x-n)
= \sum_{i=0}^n \prod_{0 \leq j \leq n, j \neq i} (x-j) +
\prod_{i=0}^n (x-i) \cdot \frac{1}{\Gamma(x+1)} \frac{d}{dx} \log \Gamma(x+1).
\een
The proof is completed by taking $x$ to $0$.
\end{proof}

By \eqref{n'=0},\eqref{Frobenius convention}, we can write
\begin{eqnarray}
&&\frac{\partial A_{\vec m,m_0}}{\partial r_{g_l}}(0,0)\nonumber\\
&=&\sum\limits_{i\in\mathcal{S}}l_i^{(g_l)}\frac{\Phi(1)-\Phi(1+m_0l_i^{(0)}+\langle \vec m,l_i\rangle)}{\Gamma(1+m_0l_i^{(0)}+\langle \vec m,l_i\rangle)}\cdot\frac{1}{\Gamma(1+m_0)\Gamma(1-m_0)}\nonumber\\
&&\quad\cdot\prod\limits_{i'\in\mathcal{S}\setminus\{i\}}\frac{1}{\Gamma(1+m_0l_{i'}^{(0)}+\langle \vec m,l_{i'}\rangle)},\quad l=1,\cdots,s,\nonumber\\
&&\frac{\partial A_{\vec m,m_0}}{\partial r_{0}}(0,0)\nonumber\\
&=&\sum\limits_{i\in\mathcal{S}}l_i^{(0)}\frac{\Phi(1)-\Phi(1+m_0l_i^{(0)}+\langle \vec m,l_i\rangle)}{\Gamma(1+m_0l_i^{(0)}+\langle \vec m,l_i\rangle)}\cdot\frac{1}{\Gamma(1+m_0)\Gamma(1-m_0)}\nonumber\\
&&\quad\quad\cdot\prod\limits_{i'\in\mathcal{S}\setminus\{i\}}\frac{1}{\Gamma(1+m_0l_{i'}^{(0)}+\langle \vec m,l_{i'}\rangle)}\nonumber\\
&&\quad+\frac{\Phi(1)-\Phi(1+m_0)}{\Gamma(1+m_0)}\cdot\frac{1}{\Gamma(1-m_0)}\cdot\prod\limits_{i\in\mathcal{S}}\frac{1}{\Gamma(1+m_0l_i^{(0)}+\langle \vec m,l_i\rangle)}\nonumber\\
&&\quad-\frac{\Phi(1)-\Phi(1-m_0)}{\Gamma(1-m_0)}\cdot\frac{1}{\Gamma(1+m_0)}\cdot\prod\limits_{i\in\mathcal{S}}\frac{1}{\Gamma(1+m_0l_i^{(0)}+\langle \vec m,l_i\rangle)}.\nonumber
\end{eqnarray}
We observe that if $m_0>0$, then $\frac{\partial A_{\vec m,m_0}}{\partial r_{g_l}}(0,0)=0, l=1,\cdots,s$. The following Lemma \ref{Frobenius lemma 4} implies that this is also true for $\frac{\partial A_{\vec m,m_0}}{\partial r_{0}}(0,0)$, whose proof needs Lemma \ref{Frobenius lemma 3}.
\begin{lemma}\label{Frobenius lemma 3}
The vectors $l_i,i\in\mathcal{S}\setminus\{i_0,i_1,i_2\}$ form a basis of $\mathbb{R}^s.$
\end{lemma}
\begin{proof}
Recall that $\tilde{v}_{i_0},\tilde{v}_{i_1},\tilde{v}_{i_2}$ form a basis of $\mathbb{R}^3$. So for each $i\in\mathcal{S}\setminus\{i_0,i_1,i_2\}$, we can find unique $C^{(i_0)}_i,C^{(i_1)}_i,C^{(i_2)}_i\in\mathbb{Q}$ such that $$\tilde{v}_i=C^{(i_0)}_i\tilde{v}_{i_0}+C^{(i_1)}_i\tilde{v}_{i_1}+C^{(i_1)}_i\tilde{v}_{i_1},$$ and we associate to $i$ a vector $\alpha^{(i)}=(\alpha^{(i)}_0,\alpha^{(i)}_1,\alpha^{(i)}_2,\alpha^{(i)}_{g_1},\cdots,\alpha^{(i)}_{g_s})\in\mathbb{L}\otimes_{\mathbb{Z}}\mathbb{Q}$, defined by
\begin{eqnarray}
\alpha^{(i)}_j=\left\{\begin{array}{cl}1,&j=i,\\-C^{(i_0)}_i,&j=i_0,\\-C^{(i_1)}_i,&j=i_1,\\-C^{(i_2)}_i,&j=i_2,\\0,&j\in\mathcal{S}\setminus\{i_0,i_1,i_2\}.\end{array}\right.\label{Frobenius lemma 3'}
\end{eqnarray}
Then we can check that the vectors $\alpha^{(i)}, i\in\mathcal{S}\setminus\{i_0,i_1,i_2\}$ form a $\mathbb{Q}$-basis of $\mathbb{L}\otimes_{\mathbb{Z}}\mathbb{Q}$. Therefore the charge vectors $l^{g_l},l=1,\cdots,s$ are linear combinations of the vectors $\alpha^{(i)}, i\in\mathcal{S}\setminus\{i_0,i_1,i_2\}$. So from \eqref{Frobenius lemma 3'}, we have
\begin{eqnarray}
\sum\limits_{i\in\mathcal{S}\setminus\{i_0,i_1,i_2\}}C^{(i_0)}_il_i+l_{i_0}=0,\nonumber\\
\sum\limits_{i\in\mathcal{S}\setminus\{i_0,i_1,i_2\}}C^{(i_1)}_il_i+l_{i_1}=0,\nonumber\\
\sum\limits_{i\in\mathcal{S}\setminus\{i_0,i_1,i_2\}}C^{(i_2)}_il_i+l_{i_2}=0,\nonumber
\end{eqnarray}
Recall that the $s+3$ vectors $l_i,i\in\mathcal{S}$ span $\mathbb{R}^s$. So the $s$ vectors  $l_i,i\in\mathcal{S}\setminus\{i_0,i_1,i_2\}$ span $\mathbb{R}^s$ and form a basis of $\mathbb{R}^s$.
\end{proof}
\begin{lemma}\label{Frobenius lemma 4}
If $m_0>0$, then $$\prod\limits_{i\in\mathcal{S}}\frac{1}{\Gamma(1+m_0l_i^{(0)}+\langle \vec m,l_i\rangle)}=0.$$
\end{lemma}
\begin{proof}
Argue by contradiction and assume that $\prod\limits_{i\in\mathcal{S}}\frac{1}{\Gamma(1+m_0l_i^{(0)}+\langle \vec m,l_i\rangle)}\neq 0.$ Then for each $i\in\mathcal{S}$, we have $m_0l_i^{(0)}+\langle \vec m,l_i\rangle\geqslant 0$ since $m_0l_i^{(0)}+\langle \vec m,l_i\rangle\in\mathbb{Z}$.  From \eqref{10'} and \eqref{11}, we have $$\sum\limits_{i\in\mathcal{S}}\big[m_0l_i^{(0)}+\langle \vec m,l_i\rangle\big]=m_0\cdot\sum\limits_{i\in\mathcal{S}}l_i^{(0)}+\langle \vec m,\sum\limits_{i\in\mathcal{S}}l_i\rangle=0,$$ and therefore, $$m_0l_i^{(0)}+\langle\vec m,l_i\rangle=0,\quad i\in\mathcal{S.}$$ Then from \eqref{11}, $$\langle\vec m,l_i\rangle=0,\quad i\in\mathcal{S}\setminus\{i_0,i_1,i_2\}.$$ From Lemma \ref{Frobenius lemma 3}, the vectors $l_i,i\in\mathcal{S}\setminus\{i_0,i_1,i_2\}$ form a basis of $\mathbb{R}^s$, and hence $m=0$. From \eqref{11}, we have $$m_0=m_0l^{(0)}_{i_0}+\langle\vec m,l_{i_0}\rangle=0,$$ which contradicts our condition that $m_0>0$.
\end{proof}

Note that for $\vec m\neq 0,i\in\mathcal{S}$, if $\langle\vec m,l_i\rangle\geqslant 0$, then we may find $i'\in\mathcal{S}\setminus\{i\}$ such that $\langle\vec m,l_{i'}\rangle<0$ by \eqref{10'}. So from Lemma \ref{Frobenius lemma 2},
\begin{eqnarray}
&&\frac{\partial A_{\vec m,m_0}}{\partial r_{g_l}}(0,0)\nonumber\\
&=&\delta_{0,m_0}\sum\limits_{i\in\mathcal{S}}l_i^{(g_l)}\frac{\Phi(1)-\Phi(1+\langle\vec m,l_i\rangle)}{\Gamma(1+\langle\vec m,l_i\rangle)}\cdot\prod\limits_{i'\in\mathcal{S}\setminus\{i\}}\frac{1}{\Gamma(1+\langle\vec m,l_{i'}\rangle)}\nonumber\\
&=&\delta_{0,m_0}\sum\limits_{\substack{i\in\mathcal{S}\\\langle\vec m,l_i\rangle<0}}l_i^{(g_l)}\frac{(-1)^{1+\langle\vec m,l_i\rangle}\Gamma(-\langle\vec m,l_i\rangle)}{\prod\limits_{i'\in\mathcal{S}\setminus\{i\}}\Gamma(1+\langle\vec m,l_{i'}\rangle)},\quad l=1,\cdots,s\label{open-closed closed'}\\
&&\frac{\partial A_{\vec m,m_0}}{\partial r_{0}}(0,0)\nonumber\\
&=&\delta_{0,m_0}\sum\limits_{i\in\mathcal{S}}l_i^{(0)}\frac{\Phi(1)-\Phi(1+\langle\vec m,l_i\rangle)}{\Gamma(1+\langle\vec m,l_i\rangle)}\cdot\prod\limits_{i'\in\mathcal{S}\setminus\{i\}}\frac{1}{\Gamma(1+\langle\vec m,l_{i'}\rangle)}\nonumber\\
&=&\delta_{0,m_0}\sum\limits_{\substack{i\in\mathcal{S}\\\langle\vec m,l_i\rangle<0}}l_i^{(0)}\frac{(-1)^{1+\langle\vec m,l_i\rangle}\Gamma(-\langle\vec m,l_i\rangle)}{\prod\limits_{i'\in\mathcal{S}\setminus\{i\}}\Gamma(1+\langle\vec m,l_{i'}\rangle)}.\label{open-closed open'}
\end{eqnarray}

We may cut down the choice of $i\in\mathcal{S}$, as the following Lemma \ref{Frobenius lemma 5} shows.
\begin{lemma}\label{Frobenius lemma 5}
Assume that $i\in\{0,1,2\}$ and $\langle\vec m,l_i\rangle<0$. Then $$\prod\limits_{i'\in\mathcal{S}\setminus\{i\}}\frac{1}{\Gamma(1+\langle\vec m,l_{i'}\rangle)}=0.$$
\end{lemma}
\begin{proof}
Without loss of generality, we may assume that $i=0$. Argue by contradiction and assume that $\prod\limits_{i'\in\mathcal{S}\setminus\{0\}}\frac{1}{\Gamma(1+\langle\vec m,l_{i'}\rangle)}\neq 0$. Then for each $i'\in\mathcal{S}\setminus\{i\}$, we have $\langle\vec m,l_{i'}\rangle\geqslant 0$ since $\langle\vec m,l_{i'}\rangle\in\mathbb{Z}.$

For $j=0,1,2$, from \eqref{10}, we have $$\langle\vec m,l_j\rangle=-\sum\limits_{g\in G_s}F_g^{(j)}\langle\vec m,l_g\rangle\leqslant 0,$$ which implies that $\langle\vec m,l_1\rangle=\langle\vec m,l_2\rangle=0$, and that we may find $g\in G_s$ with $F_g^{(0)}>0$ and $\langle\vec m,l_g\rangle>0$. For this $g$, we have $F_g^{(1)}=F_g^{(2)}=0$, and so $F_g^{(0)}=F_g^{(0)}+F_g^{(1)}+F_g^{(2)}=1$, contradicting the condition that $F_g^{(0)}\in\mathbb{Q}\cap[0,1)$.
\end{proof}

By \eqref{open-closed closed}, \eqref{open-closed open}, \eqref{open-closed closed'}, \eqref{open-closed open'}, and Lemma \ref{Frobenius lemma 1}, Lemma \ref{Frobenius lemma 5}, 
We have

\begin{prop} \label{prop:OCMirror}
The open-closed mirror map for $(Y, L_0)$ is given by;
\begin{eqnarray}
&& \Omega_{g_l}(\vec q,q_0) = \log q_{g_l}+\mathcal{C}_l(\vec q),\quad l=1,\cdots,s,\label{17}\\
&& \Omega_0(\vec q,q_0) = \log q_0+\mathcal{C}_0(\vec q),\label{18}
\end{eqnarray}
with the quantum corrections ($\vec m=(m_g)_{g\in G_s}$):
\begin{eqnarray}
&& \mathcal{C}_l(\vec q)=-\sum\limits_{h\in G_s}l_{h}^{(g_l)}\cdot
\sum\limits_{\substack{m\in\mathbb{Z}^s_{\geqslant0}\\\langle\vec m,l_h\rangle<0}}\frac{(-1)^{\langle\vec m,l_{h}\rangle}\Gamma\big(-\langle\vec m,l_{h}\rangle\big)}{\prod\limits_{i\in\mathcal{S}\setminus\{h\}}\Gamma(1+\langle\vec m,l_i\rangle)}\prod_{g\in G_s}q_g^{m_g},\label{19} \\
&& \qquad \qquad \qquad \qquad \qquad \qquad   \qquad \qquad \qquad   \qquad \qquad    
l=1,\cdots,s, \nonumber\\
&& \mathcal{C}_0(\vec q)=-\sum\limits_{h\in G_s}l_{h}^{(0)}\cdot
\sum\limits_{\substack{m\in\mathbb{Z}^s_{\geqslant0}\\\langle\vec m,l_h\rangle<0}}\frac{(-1)^{\langle\vec m,l_{h}\rangle}\Gamma\big(-\langle\vec m,l_{h}\rangle\big)}{\prod\limits_{i\in\mathcal{S}\setminus\{h\}}\Gamma(1+\langle\vec m,l_i\rangle)}\prod_{g\in G_s}q_g^{m_g},\label{20}
\end{eqnarray}
\end{prop}

Notably, the quantum corrections do not involve $l_0,l_1,l_2$ and $q_0$.

\subsection{Superpotential} \label{sec:SuperPot}

The superpotential for $(Y,L_0)$ is \cite{FL}:
\begin{eqnarray}
W(q,q_0;f)&=&\sum\limits_{\substack{m_0\in\mathbb{Z}_{>0},\vec m=(m_g)_{g\in G_s}\in\mathbb{Z}^s_{\geqslant0}\\\textrm{if }l_i^{(0)}\geqslant 0,\textrm{ then }m_0l_i^{(0)}+\langle\vec m,l_i\rangle\geqslant 0,\\\textrm{if }l_i^{(0)}<0,\textrm{ then }m_0l_i^{(0)}+\langle\vec m,l_i\rangle<0,\\\forall i\in\mathcal{S}}}\frac{(-q_0)^{m_0}}{m_0}\prod_{g\in G_s}q_g^{m_g}\nonumber\\
&&\quad\cdot\frac{\prod\limits_{l_i^{(0)}<0}(-1)^{m_0l_i^{(0)}+\langle\vec m,l_i\rangle}\Gamma(-m_0l_i^{(0)}-\langle\vec m,l_i\rangle)}{\prod\limits_{l_i^{(0)}\geqslant 0}\Gamma(1+m_0l_i^{(0)}+\langle\vec m,l_i\rangle)}.\label{general f}
\end{eqnarray}
Depending on the sign of the framing $f\in\mathbb{Z}$, we have two cases:
\begin{itemize}

\item For $f\geqslant 0$, we have $$l_i^{(0)}<0\Longrightarrow i=i_2,$$ and so
\begin{eqnarray}
W(q,q_0;f)&=&\sum\limits_{\substack{m_0\in\mathbb{Z}_{>0},\\\vec m=(m_g)_{g\in G_s}\in\mathbb{Z}^s_{\geqslant 0},\\(f+1)m_0>\langle\vec m,l_{i_2}\rangle}}\prod\limits_{i\in\mathcal{S}\setminus\{i_0,i_1,i_2\}}\frac{1}{\Gamma(1+\langle\vec m,l_i\rangle)}\nonumber\\
&&\qquad\cdot\frac{1}{\Gamma(1+m_0+\langle\vec m,l_{i_0}\rangle)}\cdot\frac{\Gamma((f+1)m_0-\langle\vec m,l_{i_2}\rangle)}{\Gamma(1+fm_0+\langle\vec m,l_{i_1}\rangle)}\nonumber\\
&&\qquad\cdot\frac{[(-1)^fq_0]^{m_0}}{m_0}\prod_{g\in G_s}[(-1)^{l_{i_2}^{(g)}}q_g]^{m_g}.\label{f>0}
\end{eqnarray}
Here we have used \eqref{n'=0} to extend the domain of sum.
\item For $f<0$, we have $$l_i^{(0)}<0\Longrightarrow i=i_1,$$ and so
\begin{eqnarray}
W(q,q_0;f)&=&\sum\limits_{\substack{m_0\in\mathbb{Z}_{>0},\\\vec m=(m_g)_{g\in G_s}\in\mathbb{Z}^s_{\geqslant0},\\fm_0+\langle\vec m,l_{i_1}\rangle<0}}\prod\limits_{i\in\mathcal{S}\setminus\{i_0,i_1,i_2\}}\frac{1}{\Gamma(1+\langle\vec m,l_i\rangle)}\nonumber\\
&&\qquad\cdot\frac{1}{\Gamma(1+m_0+\langle\vec m,l_{i_0}\rangle)}\cdot\frac{\Gamma(-fm_0-\langle\vec m,l_{i_1}\rangle)}{\Gamma(1-m_0(f+1)+\langle\vec m,l_{i_2}\rangle)}\nonumber\\
&&\qquad\cdot\frac{[(-1)^{f+1}q_0]^{m_0}}{m_0}\prod\limits_{g\in G_s}[(-1)^{l_{i_1}^{(g)}}q_g]^{m_g}\nonumber\\
&=&(-1)\cdot\sum\limits_{\substack{m_0\in\mathbb{Z}_{>0},\\\vec m=(m_g)_{g\in G_s}\in\mathbb{Z}^s_{\geqslant0},\\fm_0+\langle\vec m,l_{i_1}\rangle<0}}\prod\limits_{i\in\mathcal{S}\setminus\{i_0,i_1,i_2\}}\frac{1}{\Gamma(1+\langle\vec m,l_i\rangle)}\nonumber\\
&&\qquad\cdot\frac{1}{\Gamma(1+m_0+\langle\vec m,l_{i_0}\rangle)}\cdot\frac{\Gamma((f+1)m_0-\langle\vec m,l_{i_2}\rangle)}{\Gamma(1+fm_0+\langle\vec m,l_{i_1}\rangle)}\nonumber\\
&&\qquad\cdot\frac{[(-1)^fq_0]^{m_0}}{m_0}\prod_{g\in G_s}[(-1)^{l_{i_2}^{(g)}}q_g]^{m_g}.\label{f<0}
\end{eqnarray}
\end{itemize}
Here in the last equality, we have used \eqref{n'}.

\subsection{Open mirror symmetry}

Note that $L_0$ is an outer brane of $Y$. 
See \cite{LLLZ} for a mathematical definition of open Gromov-Witten invariants of $(Y,L_0)$
in algebraic geometry.

Let $\beta_l$ be the toric curve class corresponding to $l^{(g_l)}$ for $l=1,\cdots,s$. For each $\beta\in H_2(Y,\mathbb{Z})$, we can write uniquely $$\beta=d_1\beta_1+\cdots+d_s\beta_s,\quad d_1,\cdots,d_s\in\mathbb{Z},$$ and we set $$Q^{\beta}:=Q_1^{d_1}\cdots Q_s^{d_s}.$$ Then the disc potential of $(Y,L_0)$ is $$\mathcal{F}_{0,1}^{(Y,L_0)}(Q,Q_0;f)=\sum\limits_{\beta\in H_2(Y,\mathbb{Z})}\sum\limits_{w\in\mathbb{Z}_{>0}}N_{0,\beta}^w(f)Q^\beta Q_0^w,$$ which packages the disc invariants $N^w_{0,\beta}$ of $(Y,L_0)$.

The open string mirror symmetry predicts that, up to the usual sign ambiguity, we have the equality $$\mathcal{F}_{0,1}^{(Y,L_0)}(Q,Q_0;f)=W(\vec q,q_0;f),$$ via the open-closed mirror map:
\begin{eqnarray*}
Q_l&=&e^{\Omega_{g_l}(\vec q,q_0)}=q_{g_l}\cdot e^{\mathcal{C}_l(\vec q)},\quad l=1,\cdots,s,\\
Q_0&=&e^{\Omega_0(\vec q,q_0)}=q_0\cdot e^{\mathcal{C}_0(\vec q)}.
\end{eqnarray*}

\section{Orbifold Gromov-Witten invariants of $[\mathbb{C}^3/G]$}

Orbifold Gromov-Witten theory has been developed both in symplectic geometry \cite{CR}
and in algebraic geometry \cite{AGV}.
In this section we follow the formalism of \cite{CCIT} to 
determine the orbifold mirror map, $I$-function and $J$-function for $\cX = [\bC^3/G]$.
From their explicit expressions,
we verify that they come from the Picard-Fuchs system associated with 
the charge vectors for $\cX$.

\subsection{Orbifold Gromov-Witten theory of $\mathcal{B}G$}

For orbifold Gromov-Witten theory of $\mathcal{B}G$, we refer to Jarvis-Kimura \cite{JK} and Zhou \cite{Z07}.

Recall that our group $G$ is abelian. 
So for each $g\in G$, the conjugacy class containing $g$ is simply $\{g\}$, 
and the centralizer of $g$ in $G$ is $$Z_G(g):=\{h\in\mid hg=gh\}=G.$$
Geometrically, the classifying stack $\mathcal{B}G$ is a point with a trivial $G$-action. 
Its inertia stack is $\mathcal{IB}G:=\coprod\limits_{g\in G}\mathcal{B}Z_G(g),$ 
and the orbifold cohomology group is 
$$H^*_{orb}(\mathcal{B}G,\mathbb{C})=H^*(\mathcal{IB}G,\mathbb{C})
=\bigoplus\limits_{g\in G}\mathbb{C}\cdot e_g.$$ 
where $e_g$ denotes the class $1\in H^0(\mathcal{B}Z_G(g),\mathbb{C}).$ 
Then $\{e_g\}_{g\in G}$ is a basis of the orbifold cohomology group, called the class basis. 
In $H^*_{orb}(\mathcal{B}G,\mathbb{C})$, the orbifold Poincar\'{e} pairing is given by 
$$(e_g,e_h)_{orb}=\frac{1}{\vert G\vert}\delta_{g,-h},\quad\forall g,h\in G,$$ 
and the dual basis $\{e^g\}_{g\in G}$ of the class basis, 
with respect to $(\cdot,\cdot)_{orb}$, is given by
\begin{eqnarray}
e^g=\vert G\vert e_{-g},\quad\forall g\in G.\label{BG 1}
\end{eqnarray}
We shall not use the multiplicative structure of the orbifold cohomology.

Let $\overline{\mathcal{M}}_{g,m}(\mathcal{B}G)$ be the moduli space of orbifold stable maps from orbicurves of genus $g$ with $m$ marked points to $\mathcal{B}G$.
For our purpose, we are interested in two natural morphisms. The first one is the evaluation map \cite{JK}: $$ev_j:\overline{\mathcal{M}}_{g,m}(\mathcal{B}G)\rightarrow \mathcal{IB}G,$$ which is given by the local holonomy of the orbicurve around the $j$-th marked point. We can use these evaluation maps to give a stratification of $\overline{\mathcal{M}}_{g,m}(\mathcal{B}G)$: $$\overline{\mathcal{M}}_{g,m}(\mathcal{B}G)=\coprod\limits_{(h_1,\cdots,h_m)\in G^m}\overline{\mathcal{M}}_{g,m}(\mathcal{B}G;(h_1,\cdots,h_m))$$ where for each orbifold stable map in $$\overline{\mathcal{M}}_{g,m}(\mathcal{B}G;(h_1,\cdots,h_m)):=\bigcap\limits_{j=1}^m ev_j^{-1}(\mathcal{B}Z_G(h_j)),$$ we say that it is of topological type $(h_1,\cdots,h_m)$. The second natural morphism is the forgetting morphism: $$\phi:\overline{\mathcal{M}}_{g,m}(\mathcal{B}G)\rightarrow\overline{\mathcal{M}}_{g,m},$$ which is given by forgetting the orbifold structure on the orbicurves. It is known that $\phi$ is a finite morphism. Let $\Omega^G_{g}(h_1,\cdots,h_m)$ be the degree of $\phi$ restricting to $\overline{\mathcal{M}}_{g,m}(\mathcal{B}G;(h_1,\cdots,h_m))$. Then we have \cite{Z07}:
\begin{eqnarray}
\Omega^G_g(h_1,\cdots,h_m)=\left\{\begin{array}{cl}\vert G\vert^{2g-1},&h_1+\cdots+h_m=0,\\0,&h_1+\cdots+h_m\neq 0.\end{array}\right.\label{BG 2}
\end{eqnarray}
In particular, the degree of $\phi$ tells us that $\overline{\mathcal{M}}_{g,m}(\mathcal{B}G;(h_1,\cdots,h_m))$ is nonempty if and only if $h_1+\cdots+h_m=0$.

Let $\psi_j$ be the $\psi$-class on $\overline{\mathcal{M}}_{g,m}$. 
We define $\bar{\psi}_j:=\phi^*\psi_j.$ 
The correlators on $\mathcal{B}G$ are given by $$\langle e_{h_1}\psi_{k_1},\cdots,e_{h_m}\psi^{k_m}\rangle_{g,m}^{\mathcal{B}G}:=\int_{\overline{\mathcal{M}}_{g,m}(\mathcal{B}G)}\prod\limits_{j=1}^mev_j^*(e_{h_j})\bar{\psi}_j^{k_j}.$$ Then we have \cite{JK}:
\begin{eqnarray}
\langle e_{h_1}\psi_{k_1},\cdots,e_{h_m}\psi^{k_m}\rangle_{g,m}^{\mathcal{B}G}=\Omega^G_g(h_1,\cdots,h_m)\cdot\int_{\overline{\mathcal{M}}_{g,m}}\prod\limits_{j=1}^m\psi_j^{k_j}.\label{BG 3}
\end{eqnarray}
In particular, from \eqref{BG 2},\eqref{BG 3}, we have
\begin{eqnarray}
&&\langle e_{h_1},\cdots,e_{h_m},e_h\psi^k\rangle_{0,m+1}^{\mathcal{B}G}\nonumber\\
&=&\Omega^G_0(h_1,\cdots,h_m,h)\cdot\int_{\overline{\mathcal{M}}_{0,m+1}}\psi_{m+1}^k\nonumber\\
&=&\frac{1}{\vert G\vert}\cdot\delta_{h_1\cdots+h_m,-h}\cdot\delta_{k,m-2}.\label{BG 4}
\end{eqnarray}

\subsection{$J$-function of $\mathcal{B}G$}

To apply the formalism of \cite{CCIT} to compute the equivariant $J$-function of $\mathcal{X}=[\mathbb{C}^3/G]$, we need to determine the contribution to the $J$-function of $\mathcal{B}G$, which is defined by: $$J_{\mathcal{B}G}(x;z)=z\cdot e_0+x+\sum\limits_{h\in G}\sum\limits_{m=2}^\infty\frac{1}{m!}\langle \underbrace{x,\cdots,x,}_{m}\frac{e_h}{z-\psi}\rangle_{0,m+1}^{\mathcal{B}G}\cdot e^h,$$ with $x=\sum\limits_{h\in G}x_he_h\in H^*_{orb}(\mathcal{B}G,\mathbb{C})$, from each topological type. Using \eqref{BG 1},\eqref{BG 4}, we have
\begin{eqnarray}
&&J_{\mathcal{B}G}(x;z)\nonumber\\
&=&z\cdot e_0+\sum\limits_{h\in G}x_he_h+\sum\limits_{m=2}^\infty\sum\limits_{(h_1,\cdots,h_m)\in G^m}\frac{x_{h_1}\cdots x_{h_m}}{m!}\nonumber\\
&&\quad\cdot\sum\limits_{g\in G}\sum\limits_{k=0}^\infty\frac{1}{z^{k+1}}\langle e_{h_1},\cdots,e_{h_m},e_h\psi^{k}\rangle_{0,m+1}^{\mathcal{B}G}\cdot\vert G\vert e_{-g}\nonumber\\
&=&z\cdot e_0+\sum\limits_{h\in G}x_he_h+\sum\limits_{m=2}^\infty\sum\limits_{(h_1,\cdots,h_m)\in G^m}\frac{1}{z^{m-1}}\frac{x_{h_1}\cdots x_{h_m}}{m!}e_{h_1+\cdots+h_m}.\nonumber
\end{eqnarray}
So we see that for $m\geqslant 2$, a topological type $(h_1,\cdots,h_m,h)$ contributes to $J_{\mathcal{B}G}$ if and only if $h=-h_1-\cdots-h_m$. Define
\begin{eqnarray}
J_{\emptyset}(x;z)&:=&z\cdot e_0,\nonumber\\
J_{(h)}(x;z)&:=&x_he_h,\nonumber\\
J_{(h_1,\cdots,h_m)}(x;z)&:=&\frac{1}{z^{m-1}}\frac{x_{h_1}\cdots x_{h_m}}{m!}e_{h_1+\cdots+h_m}.\nonumber
\end{eqnarray}
Then we have
\begin{eqnarray}
J_{\mathcal{B}G}(x;z)=\sum\limits_{m=0}^\infty\sum\limits_{(h_1,\cdots,h_m)\in G^m}J_{(h_1,\cdots,h_m)}(x;z).\label{BG 5}
\end{eqnarray}

\subsection{Equivariant $J$-function of $[\mathbb{C}^3/G]$}

In this subsection, we follow the procedure of \cite{CCIT} to compute 
the equivariant $J$-function of $\mathcal{X}=[\mathbb{C}^3/G]$.

The inertia stack of $\mathcal{X}$ is $\mathcal{IX}:=\coprod\limits_{g\in G}(\mathbb{C}^3)^{(g)}/G,$ 
where $(\mathbb{C}^3)^{(g)}$ is the fixed point set of $g$. 
So the orbifold cohomology group of $\mathcal{X}$ is 
$$H^*_{orb}(\mathcal{X},\mathbb{C})=H^*(\mathcal{IX},\mathbb{C})
=\bigoplus\limits_{g\in G}\mathbb{C}\cdot e_g,$$ 
where by abuse of the notations, we denote by $e_g$  
the cohomology class class $1\in H^0((\mathbb{C}^3)^{(g)}/G,\mathbb{C})$. 
This is reasonable since $\mathcal{X}$ is an orbibundle over $\mathcal{B}G$. We also have another induced bundle $\widetilde{\mathcal{X}}$ over $\mathcal{IB}G$, whose ``fiber'' over $Z_G(g)$ is $(\mathbb{C}^3)^{(g)}$.

Let the $1$-torus $T=S^1$ act on $\mathbb{C}^3$ diagonally with weights $(\lambda_0,\lambda_1,\lambda_2)$. It is not difficult to see that the torus action descends to $\mathcal{X}$. The $T$-equivariant orbifold cohomology of $\mathcal{X}$ is $$H^*_{orb,T}(\mathcal{X})=\bigoplus\limits_{g\in G}\mathbb{C}(\lambda_0,\lambda_1,\lambda_2)\cdot e_g,$$ with the coefficient ring $\mathbb{C}(\lambda_0,\lambda_1,\lambda_2)$. The $T$-equivariant orbifold Poincar\'{e} pairing in $H^*_{orb,T}(\mathcal{X})$ is given by: $$(e_g,e_h)_{orb,T}=\frac{\delta_{g,-h}}{\vert G\vert}\mathbf{e}_T^{-1}\Big(\widetilde{\mathcal{X}}\mid_{Z_G(g)}\Big),$$ where the $T$-equivariant Euler class can be written out explicitly: $$\mathbf{e}_T\Big(\widetilde{\mathcal{X}}\mid_{Z_G(g)}\Big)=\mathbf{e}_T\Big((\mathbb{C}^3)^{(g)}\Big)=\prod\limits_{j=0}^2\lambda_j^{\delta_{0,F_g^{(j)}}}.$$ The dual basis $\{e^\vee_g\}_{g\in G}$ of $\{e_g\}_{g\in G}$ in $H^*_{orb,T}(\mathcal{X})$, with respect to $(\cdot,\cdot)_{orb,T}$, is given by $$e^\vee_g=\mathbf{e}_T\Big((\mathbb{C}^3)^{(g)}\Big)\cdot\vert G\vert e_{-g}.$$

The fixed locus of the torus action on $\mathcal{X}$ is a copy of $\mathcal{B}G$. Let $\overline{\mathcal{M}}_{g,m}(\mathcal{X})$ be the moduli space of orbifold stable maps from orbicurves of genus $g$ with $m$ marked points to the orbifold $\mathcal{X}$. The torus action on $\mathcal{X}$ induces a natural torus action on $\overline{\mathcal{M}}_{g,m}(\mathcal{X})$, and the fixed locus can be identified with $\overline{\mathcal{M}}_{g,m}(\mathcal{B}G)$.

Note that $\mathcal{X}$ is an orbibundle over $\mathcal{B}G$. Consider the universal family:
\[
\begin{CD}
\mathcal{C}_{g,m} @>ev>> \mathcal{B}G\\
@V \pi VV  \\
\overline{\mathcal{M}}_{g,m}(\mathcal{B}G)
\end{CD}
\]
Then we have a virtual bundle $\mathcal{N}_{g,m}$ given by:  $$\mathcal{N}_{g,m}=\pi_!ev^*\mathcal{X}\in K^0(\overline{\mathcal{M}}_{g,m}(\mathcal{B}G)).$$ Here $\pi_!$ is the $K$-theoretic push-forward and the ``fiber'' of the virtual bundle $\mathcal{N}_{g,m}$ at the point represented by $\Sigma\xrightarrow{f}\mathcal{B}G$ is $$H^0(\Sigma,f^*\mathcal{X})-H^1(\Sigma,f^*\mathcal{X}).$$ Using virtual localization\cite{GP}, the $T$-equivariant correlators on $\mathcal{X}$ are given by: $$\langle e_{h_1}\psi^{k_1},\cdots,e_{h_m}\psi^{k_m}\rangle_{g,m}^\mathcal{X}:=\int_{\overline{\mathcal{M}}_{g,m}(\mathcal{B}G)}\frac{\prod\limits_{j=1}^mev_j^*(e_{h_j})\bar{\psi}_j^{k_j}}{\mathbf{e}_T(\mathcal{N}_{g,m})}.$$

The genus $0$, $T$-equivariant nondescendant Gromov-Witten invariants are packaged into a generating function: 
$$\mathcal{F}_0^{\mathcal{X}}(X;\lambda_0,\lambda_1,\lambda_2)
:=\sum\limits_{m=3}^\infty\frac{1}{m!}\langle\underbrace{X,\cdots,X}_{m}\rangle^{\mathcal{X}}_{0,m},$$ with $X=\sum\limits_{g\in  G}X_ge_g\in H^*_{orb}(\mathcal{X},\mathbb{C})$. $\mathcal{F}_0^{\mathcal{X}}$ is called the Gromov-Witten potential, and is encoded in the $T$-equivariant $J$-function $$J_\mathcal{X}(X;z;\lambda_0,\lambda_1,\lambda_2):=z\cdot e_0+X+\sum\limits_{g\in G}\sum\limits_{m=2}^\infty\frac{1}{m!}\langle\underbrace{X,\cdots,X}_{m},\frac{e_g}{z-\psi}\rangle_{0,m+1}^{\mathcal{X}}\cdot e_g^{\vee},$$ 
as follows: $$J_\mathcal{X}(X;z;\lambda_0,\lambda_1,\lambda_2)
=z\cdot e_0+X+\frac{1}{z}\sum\limits_{g\in G}
\frac{\partial\mathcal{F}_0^{\mathcal{X}}}{\partial X_g}(X;\lambda_0,\lambda_1,\lambda_2)\cdot e_g^{\vee}+O(\frac{1}{z^2}).$$ 
Here $\lambda_0,\lambda_1,\lambda_2$ are considered to be parameters, 
which will take special values \eqref{weight} to compute the disc invariants of $\mathcal{X}$.

We are interested in ``small'' $J_{\mathcal{X}}$, 
i.e. restricting to small variables 
$X=\sum\limits_{g\in G_s}X_ge_g\in H^*_{orb}({\mathcal{X},\mathbb{C}})$. 
We now apply the formalism of \cite{CCIT} to compute the ``small'' $J_\mathcal{X}$ as follows. 
First, by modifying $J_{\mathcal{B}G}$ \eqref{BG 5}, 
we obtain the $T$-equivariant $I$-function of $\mathcal{X}$: 
$$I_\mathcal{X}(x;z;\lambda_0,\lambda_1,\lambda_2)
:=\sum\limits_{m=0}^\infty\sum\limits_{(h_1,\cdots,h_m)\in G^m}
J_{(h_1,\cdots,h_m)}(x;z)\cdot M_{(h_1,\cdots,h_m)}(z),$$ 
where $x=\sum\limits_{g\in G}x_ge_g\in H^*_{orb}(\mathcal{X},\mathbb{C})$, 
and the modification factors are given by:

\begin{eqnarray*}
M_\emptyset(z)&:=&1\\
M_{(h_1,\cdots,h_m)}(z)&:=&(-z)^{\sum\limits_{j=0}^2\lfloor F^{(j)}_{h_1}+\cdots+F^{(j)}_{h_m}
\rfloor}\cdot\prod\limits_{j=0}^2\frac{\Gamma(F^{(j)}_{h_1}+\cdots+F^{(j)}_{h_m}-\frac{\lambda_j}{z})}
{\Gamma(\langle F^{(j)}_{h_1}+\cdots+F^{(j)}_{h_m}\rangle-\frac{\lambda_j}{z})}.
\end{eqnarray*}
The modification factors are computed directly as in \cite{CCIT}, 
using the data from the orbibundle $\mathcal{X}\rightarrow\mathcal{B}G$ and the invertible multiplicative class $\mathbf{e}^{-1}_T(\cdot)$. Here we have used \eqref{gamma} to write the modification factors in terms of $\Gamma$ function. Then direct calculation gives:
\begin{eqnarray*}
&&I_\mathcal{X}(x;z;\lambda_0,\lambda_1,\lambda_2)\\
&=&z\sum\limits_{k_g\in\mathbb{Z}_{\geqslant 0},\forall g\in G}(-z)^{\sum\limits_{j=0}^2\lfloor\sum\limits_{g\in G}k_gF_g^{(j)}\rfloor-\sum\limits_{g\in G}k_g}\\
&&\quad\cdot\prod\limits_{j=0}^2\frac{\Gamma\Big(\sum\limits_{g\in G}k_gF_g^{(j)}-\frac{\lambda_j}{z}\Big)}{\Gamma\Big(\langle\sum\limits_{g\in G}k_gF_g^{(j)}\rangle-\frac{\lambda_j}{z}\Big)}\cdot\prod\limits_{g\in G}\frac{(-x_g)^{k_g}}{k_g!}\cdot e_{\sum\limits_{g\in G}k_gg}.
\end{eqnarray*}
From \eqref{age},\eqref{adno}, we have:
\begin{eqnarray*}
&&\sum\limits_{j=0}^2\lfloor\sum\limits_{g\in G}k_gF_g^{(j)}\rfloor-\sum\limits_{g\in G}k_g\\
&=&\sum\limits_{j=0}^2\sum\limits_{g\in G}k_gF_g^{(j)}-\sum\limits_{g\in G}k_g-\sum\limits_{j=0}^2\langle\sum\limits_{g\in G}k_gF_g^{(j)}\rangle\\
&=&\sum\limits_{g\in G}k_g(\sum\limits_{j=0}^2F_g^{(j)}-1)-\sum\limits_{j=0}^2F^{(j)}_{\sum\limits_{g\in G}k_gg}\\
&=&\sum\limits_{g\in G}k_g(a(g)-1)-a(\sum\limits_{g\in G}k_gg).
\end{eqnarray*}
Restricting to small variables $x=\sum\limits_{g\in G_s}x_ge_g$, we obtain:
\begin{eqnarray}
&&I_\mathcal{X}\big(x;z;\lambda_0,\lambda_1,\lambda_2\big)\nonumber\\
&=&z\sum\limits_{h\in G}(-z)^{-a(h)}\sum\limits_{\substack{k_g\in\mathbb{Z}_{\geqslant 0},\forall g\in G_s\\\sum\limits_{g\in G_s}k_gg=h}}\prod\limits_{j=0}^2\frac{\Gamma\big(\sum\limits_{g\in G_s}k_gF_g^{(j)}-\frac{\lambda_j}{z}\big)}{\Gamma\big(F_h^{(j)}-\frac{\lambda_j}{z}\big)}\nonumber\\
&&\quad\quad\cdot\prod\limits_{g\in G_s}\frac{(-x_g)^{k_g}}{k_g!}\cdot e_h.\label{21}
\end{eqnarray}
Furthermore, the ``small'' $I_\mathcal{X}$ can be written as:
\begin{eqnarray*}
&&I_\mathcal{X}\big(x;z;\lambda_0,\lambda_1,\lambda_2\big)\\
&=&z\cdot e_0-\sum\limits_{a(h)=1}\sum\limits_{\substack{k_g\in\mathbb{Z}_{\geqslant 0},\forall g_\in G_s\\\sum\limits_{g\in G_s}k_gg=h}}\prod\limits_{j=0}^2\frac{\Gamma\big(\sum\limits_{g\in G_s}k_gF_g^{(j)}\big)}{\Gamma\big(F_h^{(j)}\big)}\cdot\prod\limits_{g\in G_s}\frac{(-x_g)^{k_g}}{k_g!}\cdot e_h\\
&&\quad+O(\frac{1}{z}).
\end{eqnarray*}
Finally, the ``small'' $J_\mathcal{X}$ coincides with the ``small'' $I_\mathcal{X}$. More concretely, 
we have

\begin{prop}
Let $$X=\sum\limits_{g\in G_s}X_ge_g,\quad x=\sum\limits_{g\in G_s}x_ge_g,$$ 
then we have:
\begin{eqnarray}
J_\mathcal{X}\big(X;z;\lambda_0,\lambda_1,\lambda_2\big)=I_\mathcal{X}\big(x;z;\lambda_0,\lambda_1,\lambda_2\big),\label{J=I}
\end{eqnarray}
via the orbifold (closed) mirror map $X=X(x)$ given by
\begin{eqnarray}
&& X_h=-\sum\limits_{\substack{k_g\in\mathbb{Z}_{\geqslant 0},\forall g\in G_s\\\sum\limits_{g\in G_s}k_gg=h}}\prod\limits_{j=0}^2\frac{\Gamma\big(\sum\limits_{g\in G_s}k_gF_g^{(j)}\big)}{\Gamma\big(F_h^{(j)}\big)}\cdot\prod\limits_{g\in G_s}\frac{(-x_g)^{k_g}}{k_g!},\quad\forall h\in G_s.\label{orbclose}
\end{eqnarray}
\end{prop}

\subsection{GLSM for $\mathcal{X}$}

In this subsection, we find the charge vectors for $\mathcal{X}$.

First we take the non-equivariant limit on the ``small'' $I_\mathcal{X}$:
\begin{eqnarray*}
&&\lim_{(\lambda_0,\lambda_1,\lambda_2)\rightarrow 0}I_\mathcal{X}\big(x;z;\lambda_0,\lambda_1,\lambda_2\big)\\
&=&z\sigma_0(x)\cdot e_0-\sum\limits_{a(h)=1}\sigma_h(x)\cdot e_h+\frac{1}{z}\sum\limits_{a(h)=2}\sigma_h(x)\cdot e_h
\end{eqnarray*}
where
\begin{eqnarray*}
&&\sigma_h(x)\\
&=&\sum\limits_{\substack{k_g\in\mathbb{Z}_{\geqslant 0},\forall g\in G_s\\\sum\limits_{g\in G_s}k_gg=h}}\prod\limits_{j=0}^2\frac{\Gamma\big(\sum\limits_{g\in G_s}k_gF_g^{(j)}\big)}{\Gamma\big(\langle\sum\limits_{g\in G_s}k_gF_g^{(j)}\rangle\big)}\cdot\prod\limits_{g\in G_s}\frac{(-x_g)^{k_g}}{k_g!}\\
&=:&\sum\limits_{\substack{k_g\in\mathbb{Z}_{\geqslant 0},\forall g\in G_s\\\sum\limits_{g\in G_s}k_gg=h}}C(k_{g_1},\cdots,k_{g_s})\cdot\prod_{g\in G_s}x_g^{k_g},\quad h\in G
\end{eqnarray*}

Let $n_g$ be the order of $g$. By direct calculation, we have
\begin{eqnarray*}
&&\bigg(k_{g_l}+n_{g_l}\bigg)_{n_{g_l}}\cdot C(k_{g_1},\cdots,k_{g_l}+n_{g_l},\cdots,k_{g_s})\\
&=&\prod_{j=0}^2\bigg(-\sum\limits_{g\in G_s}k_gF_g^{(j)}\bigg)_{n_{g_l}F_{g_l}^{(j)}}\cdot C(k_{g_1},\cdots,k_{g_l},\cdots,k_{g_s}),\quad l=1,\cdots,s.
\end{eqnarray*}
So the recursion relations naturally give the following operators: $$\widetilde{\mathcal{D}}_l=\bigg(\Theta_{x_{g_l}}\bigg)_{n_{g_l}}-x_{g_l}^{n_{g_l}}\prod_{j=0}^2\bigg(-\sum\limits_{g\in G_s}F_g^{(j)}\Theta_{x_{g}}\bigg)_{n_{g_l}F_{g_l}^{(j)}},\quad l=1,\cdots,s.$$
Observe that if we set $y_{g_l}=x_{g_l}^{n_{g_l}}$, then $$\widetilde{\mathcal{D}}_l=\bigg(n_{g_l}\Theta_{y_{g_l}}\bigg)_{n_{g_l}}-y_{g_l}\prod_{j=0}^2\bigg(-\sum\limits_{g\in G_s}n_gF_g^{(j)}\Theta_{y_g}\bigg)_{n_{g_l}F_{g_l}^{(j)}},\quad l=1,\cdots,s.$$
So naturally, the charge vectors for $\mathcal{X}$ are $$\big(-n_{g_l}F_{g_l}^{(0)},-n_{g_l}F_{g_l}^{(1)},-n_{g_l}F_{g_l}^{(2)},0,\cdots,n_{g_l},\cdots,0\big),\quad l=1,\cdots,s,$$
which form a basis of $\mathbb{L}\otimes_{\mathbb{Z}}\mathbb{Q}$. So we have
\begin{eqnarray*}
&&y_g=\prod\limits_{j=0}^2z_j^{-n_gF_g^{(j)}}\cdot z_g^{n_g},\quad g\in G_s,
\end{eqnarray*}
or equivalently,
\begin{eqnarray}
&&x_g=\prod\limits_{j=0}^2z_j^{-F_g^{(j)}}\cdot z_g,\quad g\in G_s.\label{23}
\end{eqnarray}

\section{Orbifold disc potential for $[\mathbb{C}^3/G]$}

In this section we compute the orbifold disc invariants for $[\mathbb{C}^3/G]$
with boundary in some special D-branes using results from \cite{BC}.
Our main result is \eqref{eqn:OrbifoldDisc}.

\subsection{A family of D-branes in $[\mathbb{C}^3/G]$}

We first view $\mathbb{C}^3$ as an open chart of $\mathcal{O}_{\mathbb{P}^1}(-1)\oplus\mathcal{O}_{\mathbb{P}^1}(-1)$ as follows. 
Recall that the resolved conifold can be given local coordinates $(z_0,z_1,z_2)$ at $0\in\mathbb{P}^1$ and $(z_0',z_1',z_2')$ at $\infty\in\mathbb{P}^1$, 
which are related by 
$$z_0'=z_0^{-1},z_1'=z_0z_1,z_2'=z_0z_2.$$ 
Then  $\mathbb{C}^3$ can be naturally identified as 
the open chart of the resolved conifold at $0\in\mathbb{P}^1$, 
and $z_0$ is the coordinate of $\mathbb{P}^1$ near $0\in\mathbb{P}^1$.

As in \cite{Z06}, for each $A>0$,
define an anti-holomorphic involution on $\mathbb{P}^1$ by 
$$\sigma(z_0)=A\cdot\bar{z}_0^{-1}.$$ 
The following anti-holomorphic involution on 
$\mathcal{O}_{\mathbb{P}^1}(-1)\oplus\mathcal{O}_{\mathbb{P}^1}(-1)$ 
covers $\sigma$:  
$$\bar{\sigma}(z_0,z_1,z_2)
=(A\bar{z}_0^{-1},A^{-\frac{1}{2}}\bar{z}_0\bar{z}_2,A^{-\frac{1}{2}}\bar{z}_0\bar{z}_1).$$ 
The fixed locus of $\bar{\sigma}$ is a Lagrangian with topology $S^1\times\mathbb{C}$ 
given by 
$$L=\{(A^{\frac{1}{2}}e^{\sqrt{-1}\theta},u,e^{-\sqrt{-1}\theta}\cdot\bar{u})\mid \theta\in\mathbb{R},u\in\mathbb{C}\}\subset\mathbb{C}^3.$$

The action of $G$ can be extended to  $\mathcal{O}_{\mathbb{P}^1}(-1)\oplus\mathcal{O}_{\mathbb{P}^1}(-1)$ by 
$$g.(z_0',z_1',z_2'):=\big(z_0\cdot e^{-2\pi\sqrt{-1}F_g^{(0)}} ,z_1'\cdot e^{2\pi\sqrt{-1}(F_g^{(0)}+F_g^{(1)})},z_2'\cdot e^{2\pi\sqrt{-1}(F_g^{(0)}+F_g^{(2)})}\big).$$ 
One can check that the $G$-action commutes with the $\bar{\sigma}$-action. 
Hence the orbifold $[\mathcal{O}_{\mathbb{P}^1}(-1)\oplus\mathcal{O}_{\mathbb{P}^1}(-1)/G]$ 
admits a natural anti-holomorphic involution coming from $\bar{\sigma}$, 
whose fixed locus is a Lagrangian $\mathcal{L}_0=[L/G]\subset[\mathbb{C}^3/G]$.

\subsection{Torus action on $(\mathcal{X},\mathcal{L}_0)$}

As in \cite{BC}, let the $1$-torus $T=S^1$ act on $\mathbb{C}^3$ with weights
\begin{eqnarray}
(\lambda_0,\lambda_1,\lambda_2):=(\frac{1}{\vert G/G_0\vert},-a,a-\frac{1}{\vert G/G_0\vert}),\label{weight}
\end{eqnarray}
where $a\in\frac{1}{|G/G_0|}\cdot\mathbb Z$ is a free parameter, playing the role of framing.  
Then the torus action can be extended to the resolved conifold, 
and descend to $\mathcal{X}=[\mathbb{C}^3/G]$ 
since the $T$-action commutes with the $G$-action. 
In particular, $T$ acts on the coarse moduli space of the ``$z_0$-axis'' in $\mathcal{X}$ with weight $1$. 
Moreover, we can check that the $T$-action commutes with the $\bar{\sigma}$-action, 
and hence the $T$-action preserves $\mathcal{L}_0$ in $\mathcal{X}$.

\subsection{Orbifold open Gromov-Witten invariants for $(\mathcal{X},\mathcal{L}_0)$}

Using formal localization, the orbifold open Gromov-Witten invariant of 
$(\mathcal{X},\mathcal{L}_0)$ for bordered orbicurves of genus $g$ with 
$h$ boundary components, together with winding numbers $d_1,\cdots,d_h$ and 
insertions $e_{g'_1}$, $\cdots$, $e_{g'_m}$, is defined by \cite{BC}:
\begin{eqnarray}
&&\langle e_{g'_1},\cdots,e_{g'_m}\rangle_{g,m;d_1,\cdots,d_h}^{(\mathcal{X},\mathcal{L}_0)}\nonumber\\
&:=&\sum\limits_{(k_1,\cdots,k_h)\in G^h}\langle e_{g'_1},\cdots,e_{g'_m},\frac{D_{k_1}(d_1,a)\cdot e^\vee_{k_1}}{\frac{1}{d_1}-\psi},\cdots,\frac{D_{k_h}(d_h,a)\cdot e^\vee_{k_h}}{\frac{1}{d_h}-\psi}\rangle_{g,m+h}^\mathcal{X}.\nonumber
\end{eqnarray}
 Here the ``disc function'' $D_k(d,a)$ is defined as follows. For $d\in\mathbb{Z}_{>0}$ and $k\in G$, if
\begin{eqnarray}
F_k^{(0)}=\langle\frac{\vert G_0\vert}{\vert G\vert}d\rangle,\nonumber
\end{eqnarray}
then we define
\begin{eqnarray}
&&D_k(d,a)\nonumber\\
&:=&\frac{d^{-a(k)}}{\vert G_0\vert}\cdot\frac{1}{\Gamma\big(1+\lambda_0d-F_k^{(0)}\big)}\cdot\frac{\Gamma\big(-\lambda_2d+F_k^{(2)}\big)}{\Gamma\big(1+\lambda_1d-F_k^{(1)}\big)}\nonumber\\
&=&\frac{d^{-a(k)}}{\vert G_0\vert}\cdot\frac{1}{\Gamma\big(1+\frac{\vert G_0\vert}{\vert G\vert}d-F_k^{(0)}\big)}\cdot\frac{\Gamma\big(-(a-\frac{\vert G_0\vert}{\vert G\vert})d+F_k^{(2)}\big)}{\Gamma\big(1-ad-F_k^{(1)}\big)}\label{disc function}
\end{eqnarray}
otherwise we define $D_k(d,a)=0$.

The type $(g,h)$ orbifold open Gromov-Witten potential of $(\mathcal{X},\mathcal{L}_0)$ is defined by:
\begin{eqnarray}
&&\mathcal{F}_{g,h}^{(\mathcal{X},\mathcal{L}_0)}(X;w_1,\cdots,w_h;a)\nonumber\\
&=&\sum\limits_{\substack{m\\d_1,\cdots,d_h\geqslant 1}}\frac{1}{m!}\langle\underbrace{X,\cdots,X}_{m}\rangle_{g,m;d_1,\cdots,d_h}^{(\mathcal{X},\mathcal{L}_0)}\cdot\prod\limits_{j=1}^hw_j^{d_j},\label{prepotential}
\end{eqnarray}
with $X=\sum\limits_{g\in G}X_ge_g\in H^{*}_{orb}(\mathcal{X},\mathbb{C})$. Note that our definition \eqref{prepotential} differs from Definition $2$ in \cite{BC} by the factorial under $w_j^{d_j}$.

The orbifold disc potential $\mathcal{F}_{0,1}^{(\mathcal{X},\mathcal{L}_0)}$ can be expressed in terms of $J_\mathcal{X}$ and the disc functions. First note that, since $\{e_g\}_{g\in G}$ and $\{e^\vee_g\}_{g\in G}$ are dual to each other with respect to $(\cdot,\cdot)_{orb,T}$, it follows that $$J_\mathcal{X}(X;z;\lambda_0,\lambda_1,\lambda_2)=z\cdot e_0+X+\sum\limits_{g\in G}\sum\limits_{m=2}^\infty\frac{1}{m!}\langle\underbrace{X,\cdots,X}_{m},\frac{e^\vee_g}{z-\psi}\rangle_{0,m+1}^{\mathcal{X}}\cdot e_g.$$ So we have:
\begin{eqnarray*}
&&\mathcal{F}_{0,1}^{(\mathcal{X},\mathcal{L}_0)}(X;X_0;a)\\
&=&\sum\limits_{m=2}^\infty\sum\limits_{d=1}^\infty\frac{1}{m!}\langle\underbrace{X,\cdots,X}_{m}\rangle_{0,m;d}^{(\mathcal{X},\mathcal{L}_0)}\cdot X_0^d\\
&=&\sum\limits_{d=1}^\infty\sum\limits_{m=2}^\infty\sum\limits_{k\in G}\frac{1}{m!}\langle\underbrace{X,\cdots,X}_m,\frac{D_k(d,a)\cdot e^\vee_k}{\frac{1}{d}-\psi}\rangle^{\mathcal{X}}_{0,m+1}\cdot X_0^d\\
&=&\sum\limits_{d=1}^\infty\Big(J_\mathcal{X}(X;\frac{1}{d};\lambda_0,\lambda_1,\lambda_2)-\frac{1}{d}\cdot e_0-X,\sum\limits_{k\in G}D_k(d,a)\cdot e^\vee_k\Big)_{orb,T}\cdot X_0^d
\end{eqnarray*}
Therefore, it is natural to take care of the unstable terms as follows:
\begin{eqnarray*}
\langle\rangle_{0,0;d}^{(\mathcal{X},\mathcal{L}_0)}&:=&\frac{1}{d}D_0(d,a),\quad d\geqslant 1,\\
\langle e_g\rangle_{0,1;d}^{(\mathcal{X},\mathcal{L}_0)}&:=&D_g(d,a),\quad g\in G,d\geqslant 1.
\end{eqnarray*}
So including the unstable terms, we have
\begin{eqnarray} \label{eqn:OrbifoldDisc}
&&\mathcal{F}^{(\mathcal{X},\mathcal{L}_0)}_{0,1}(X,X_0;a)\nonumber\\
&=&\sum\limits_{d=1}^\infty\bigg(J_\mathcal{X}(X;\frac{1}{d};\lambda_0,\lambda_1,\lambda_2),\sum\limits_{h\in G}D_h(d,a)\cdot e_h^\vee\bigg)_{orb,T}X_0^d.\label{discpotential}
\end{eqnarray}
We will only consider the ``small'' disc potential, i.e. 
we will take  $X=\sum\limits_{g\in G_s}X_ge_g\in H^*_{orb}(\mathcal{X},\mathbb{C})$.

\subsection{GLSM for $(\mathcal{X},\mathcal{L}_0)$}

We propose that the orbifold open mirror map is 
\begin{eqnarray}
X_0=(-1)^{\lambda_2+1}x_0.\label{orbopen}
\end{eqnarray}
Then together with the orbifold (closed) mirror map $(67)$, we have:
\begin{eqnarray}
&&\mathcal F_{0,1}^{(\mathcal X,\mathcal L_0)}(X(x),X_0(x_0);a)\nonumber\\
&=&\frac{1}{\vert G_0\vert}\sum\limits_{\substack{d\in\mathbb{Z}_{\geqslant 1},k_g\in\mathbb{Z}_{\geqslant 1},\forall g\in G_s\\\langle\sum\limits_{g\in G_s}k_gF_g^{(0)}\rangle=\langle\frac{\vert G_0\vert}{\vert G\vert}d\rangle}}(-1)^{a(\sum\limits_{g\in G_s}k_gg)}\cdot\frac{\Gamma(\sum\limits_{g\in G_s}k_gF_g^{(0)}-\frac{\vert G_0\vert}{\vert G\vert}d)}{\Gamma(\langle\sum\limits_{g\in G_s}k_gF_g^{(0)}\rangle-\frac{\vert G_0\vert}{\vert G\vert}d)}\nonumber\\
&&\quad\cdot\frac{\Gamma(\sum\limits_{g\in G_s}k_gF_g^{(1)}+ad)}{\Gamma(\langle\sum\limits_{g\in G_s}k_gF_g^{(1)}\rangle+ad)}\cdot\frac{\Gamma(\sum\limits_{g\in G_s}k_gF_g^{(2)}-(a-\frac{\vert G_0\vert}{\vert G\vert})d)}{\Gamma(\langle\sum\limits_{g\in G_s}k_gF_g^{(2)}\rangle-(a-\frac{\vert G_0\vert}{\vert G\vert})d)}\nonumber\\
&&\quad\cdot\frac{1}{\Gamma\big(1+\frac{\vert G_0\vert}{\vert G\vert}d-\langle\sum\limits_{g\in G_s}k_gF_g^{(0)}\rangle\big)}\cdot\frac{\Gamma\big(-(a-\frac{\vert G_0\vert}{\vert G\vert})d+\langle\sum\limits_{g\in G_s}k_gF_g^{(2)}\rangle\big)}{\Gamma\big(1-ad-\langle\sum\limits_{g\in G_s}k_gF_g^{(1)}\rangle\big)}\nonumber\\
&&\quad\cdot\prod\limits_{g\in G_s}\frac{(-x_g)^{k_g}}{k_g!}\cdot\frac{((-1)^{1+\lambda_2}x_0)^d}{d}\label{findrestriction}\\
&=&\frac{1}{\vert G_0\vert}\sum\limits_{\substack{d\in\mathbb{Z}_{>0},k_g\in\mathbb{Z}_{\geqslant 0},\forall g\in G_s\\\langle\sum\limits_{g\in G_s}k_gF_g^{(0)}\rangle=\langle\frac{\vert G_0\vert}{\vert G\vert}d\rangle}}(-1)^{\lfloor\sum\limits_{g\in G_s}k_gF_g^{(2)}\rfloor}\nonumber\\
&&\quad\cdot\frac{1}{\Gamma\bigg(1-\sum\limits_{g\in G_s}k_gF_g^{(0)}+\lambda_0d\bigg)}\cdot\frac{\Gamma\bigg(\sum\limits_{g\in G_s}k_gF_g^{(2)}-\lambda_2d\bigg)}{\Gamma\bigg(1-\sum\limits_{g\in G_s}k_gF_g^{(1)}+\lambda_1d\bigg)}\nonumber\\
&&\quad\cdot\prod\limits_{g\in G_s}\frac{x_g^{k_g}}{k_g!}\cdot\frac{((-1)^{1+\lambda_2}x_0)^d}{d}\nonumber\\
&=:&\frac{1}{\vert G_0\vert}\sum\limits_{\substack{d\in\mathbb{Z}_{>0},k_g\in\mathbb{Z}_{\geqslant 0},\forall g\in G_s\\\langle\sum\limits_{g\in G_s}k_gF_g^{(0)}\rangle=\langle\frac{\vert G_0\vert}{\vert G\vert}d\rangle}}C(k_{g_1},\cdots,k_{g_s};d)\cdot\prod\limits_{g\in G_s}x_g^{k_g}\cdot x_0^d\nonumber
\end{eqnarray}
Let $n_g$ be the order of $g\in G$. Assume that $\lambda_1=-a\geqslant 0$, and then $\lambda_2<0$. By direct calculation, we have: 
\begin{eqnarray*}
&&\bigg(k_{g_l}+n_{g_l}\bigg)_{n_{g_l}}\cdot C(k_{g_1},\cdots,k_{g_l}+n_{g_l},\cdots,k_{g_s};d)\\
&=&\prod\limits_{j=0}^2\bigg(-\sum\limits_{g\in G_s}k_gF_g^{(j)}+\lambda_jd\bigg)_{n_{g_l}F_{g_l}^{(j)}}\cdot C(k_{g_1},\cdots,k_{g_s};d),\quad l=1,\cdots,s,\\
&&\bigg(d+\vert G/G_0\vert\bigg)_{\vert G/G_0\vert}\cdot\prod\limits_{j=0}^1\bigg(-\sum\limits_{g\in G_s}k_gF_g^{(j)}+\lambda_j(d+\vert G/G_0\vert)\bigg)_{\lambda_j\vert G/G_0\vert}\\
&&\quad\quad\cdot C(k_{g_1},\cdots,k_{g_s};d+\vert G/G_0\vert)\\
&=&\bigg(-d\bigg)_{-\lambda_2\vert G/G_0\vert}\cdot\bigg(-\sum\limits_{g\in G_s}k_gF_g^{(2)}+\lambda_2d\bigg)_{-\lambda_2\vert G/G_0\vert}\cdot C(k_{g_1},\cdots,k_{g_s};d).
\end{eqnarray*}
Then the recursion relation naturally give the following operators:
\begin{eqnarray*}
\widetilde{\mathcal{D}}_l&=&\bigg(\Theta_{x_{g_l}}\bigg)_{n_{g_l}}-x_{g_l}^{n_{g_l}}\prod_{j=0}^2\bigg(-\sum\limits_{g\in G_s}F_g^{(j)}\Theta_{x_{g}}\bigg)_{n_{g_l}F_{g_l}^{(j)}},\quad l=1,\cdots,s,\\
\widetilde{\mathcal{D}}_0&=&\bigg(\Theta_{x_0}\bigg)_{\vert G/G_0\vert}\cdot\prod\limits_{j=0}^1\bigg(-\sum\limits_{g\in G_s}F_g^{(j)}\Theta_{x_g}+\lambda_j\Theta_{x_0}\bigg)_{\lambda_j\vert G/G_0\vert}\\
&&\quad-x_0^{\vert G/G_0\vert}\bigg(-\Theta_{x_0}\bigg)_{\vert G/G_0\vert}\cdot\bigg(-\sum\limits_{g\in G_s}F_g^{(2)}\Theta_{x_g}+\lambda_2\Theta_{x_0}\bigg)_{-\lambda_2\vert G/G_0\vert}.
\end{eqnarray*}
Observe that if we set $y_{g_l}=x_{g_l}^{n_{g_l}}$ and $y_0=x_0^{\vert G/G_0\vert}$, then 
\begin{eqnarray*}
\widetilde{\mathcal{D}}_l&=&\bigg(n_{g_l}\Theta_{y_{g_l}}\bigg)_{n_{g_l}}-y_{g_l}\prod_{j=0}^2\bigg(-\sum\limits_{g\in G_s}n_gF_g^{(j)}\Theta_{y_g}\bigg)_{n_{g_l}F_{g_l}^{(j)}},\quad l=1,\cdots,s,\\
\widetilde{\mathcal D}_0&=&\bigg(\vert G/G_0\vert\Theta_{y_0}\bigg)_{\vert G/G_0\vert}\cdot\prod\limits_{j=0}^1\bigg(-\sum\limits_{g\in G_s}n_gF_g^{(j)}\Theta_{y_g}+\lambda_j\vert G/G_0\vert\Theta_{y_0}\bigg)_{\lambda_j\vert G/G_0\vert}\\
&&\quad-y_0\bigg(-\vert G/G_0\vert\Theta_{y_0}\bigg)_{\vert G/G_0\vert}\cdot\bigg(-\sum\limits_{g\in G_s}n_gF_g^{(2)}\Theta_{y_g}+\lambda_2\vert G/G_0\vert\Theta_{y_0}\bigg)_{-\lambda_2\vert G/G_0\vert}.
\end{eqnarray*}
So naturally, the charge vectors for $(\mathcal{X},\mathcal L_0)$ are 
\begin{displaymath}
\begin{array}{ccccccc|cc}
\hat l^{(g_1)}=\big(-n_{g_1}F_{g_1}^{(0)},&-n_{g_1}F_{g_1}^{(1)},&-n_{g_1}F_{g_1}^{(2)},&n_{g_1},&0,&\cdots,&0,&0,&0\big),\\
\hat l^{(g_2)}=\big(-n_{g_2}F_{g_2}^{(0)},&-n_{g_2}F_{g_2}^{(1)},&-n_{g_2}F_{g_2}^{(2)},&0,&n_{g_2},&\cdots,&0,&0,&0\big),\\
\cdots&\cdots&&&&&&&\\
\hat l^{(g_s)}=\big(-n_{g_s}F_{g_s}^{(0)},&-n_{g_s}F_{g_s}^{(1)},&-n_{g_s}F_{g_s}^{(2)},&0,&0,&\cdots,&n_{g_s},&0,&0\big),\\
\hat l^{(0)}=\big(\lambda_0\vert G/G_0\vert,&\lambda_1\vert G/G_0\vert,&\lambda_2\vert G/G_0\vert,&0,&0,&\cdots,&0,&\vert G/G_0\vert,&-\vert G/G_0\vert\big),
\end{array}
\end{displaymath}
We can also obtain the same charge vectors for the case $\lambda_1=-a<0$. Then we have
\begin{eqnarray*}
&&y_g=\prod\limits_{j=0}^2z_j^{-n_gF_g^{(j)}}\cdot z_g^{n_g},\quad g\in G_s,\\
&&y_0=\prod\limits_{j=0}^2z_j^{\lambda_j\vert G/G_0\vert}\cdot\bigg(\frac{z_+}{z_-}\bigg)^{\vert G/G_0\vert},
\end{eqnarray*}
or equivalently,
\begin{eqnarray}
&&x_g=\prod\limits_{j=0}^2z_j^{-F_g^{(j)}}\cdot z_g,\quad g\in G_s\label{23},\\
&&x_0=\prod\limits_{j=0}^2z_j^{\lambda_j}\cdot\frac{z_+}{z_-}.
\end{eqnarray}

\section{Comparison between the superpotential and the orbifold disc potential}

\subsection{Change of variables}

From \eqref{10}, we see that the matrix $(l^{(g)}_h)_{g,h\in G_s}$ is invertible. Hence we may find $(b_{gh})_{g,h\in G_s}$ such that
\begin{eqnarray}
\sum\limits_{h\in G_s}b_{g_1h}l^{(h)}_{g_2}=\sum\limits_{h\in G_s}l^{(g_1)}_hb_{hg_2}=\delta_{g_1,g_2}\label{26}
\end{eqnarray}

It is not difficult to check that  $\{\tilde l^{(g_1)},\cdots,\tilde l^{(g_s)},\tilde l^{(0)}\}$ and $\{\hat l^{(g_1)},\cdots,\hat l^{(g_s)},\hat l^{(0)}\}$ span the same linear subspace in $\mathbb Q^{s+1}$ if and only if 
\begin{eqnarray}
-a=\lambda_1=l_1^{(0)}+\sum\limits_{g\in G_s}l_g^{(0)}F_g^{(1)},\label{framing}
\end{eqnarray}
which gives an explicit correspondence between the framings $a$ and $f$. Via this correspondence, we have:
\begin{eqnarray*}
\prod\limits_{h\in G_s}q_h^{b_{gh}}&=&\prod\limits_{j=0}^2z_j^{-F_{g}^{(j)}}\cdot z_g,\\
q_0\cdot\prod\limits_{h\in G_s}q_h^{-\sum\limits_{g\in G_s}l_g^{(0)}b_{gh}}&=&\prod\limits_{j=0}^2z_j^{\lambda_j}\cdot\frac{z_+}{z_-}.
\end{eqnarray*}
and so it is natural to set:
\begin{eqnarray}
x_g&=&\prod\limits_{h\in G_s}q_h^{b_{gh}},\quad\forall g\in G_s,\label{27}\\
x_0&=&(-1)^{f+1+\lambda_2}q_0\cdot\prod_{g\in G_s}q_g^{-\sum\limits_{h\in G_s}l_{h}^{(0)}b_{hg}},\label{28}
\end{eqnarray}
where the phase factor $(-1)^{f+1+\lambda_2}$ is included for convenience of comparison. 

The authors of \cite{BKMP09} have conjectured the change of variables $X_0=X_0(q,q_0)$ as follows. Note that $$\Omega_0(q,q_0)=\log q_0-\sum\limits_{g\in G_s}\sum\limits_{h\in G_s}l_h^{(0)}b_{hg}\bigg(\log q_g-\Omega_g(q,q_0)\bigg),$$ and they conjectured that $$X_0=q_0\cdot\prod_{g\in G_s}q_g^{-\sum\limits_{h\in G_s}l_{h}^{(0)}b_{hg}},$$ which matches our result from GLSM via our propoed orbifold open mirror map $X_0=(-1)^{1+\lambda_2}x_0$.

\subsection{General Discussion}
Now we plug in the change of variables $(67),(73)$,\eqref{27},\eqref{28} to compare the ``small'' disc potential $\mathcal{F}_{0,1}^{(\mathcal{X},\mathcal{L}_0)}$ \eqref{discpotential} with the superpotential $W$ \eqref{f>0},\eqref{f<0} corresponding to $(Y,L_0)$. Recall that our brane $\mathcal{L}_0$ intersects the ``$z_0$-axis'' of $\mathcal{X}$ and $L_0$ intersects the non-compact toric curve in $Y$ given by $\overline{v_{i_1}v_{i_2}}\subset\overline{v_1v_2}$.

Via the correspondence of framings \eqref{framing}, the ``small'' disc potential is given by:
\begin{eqnarray}
&&\mathcal{F}_{0,1}^{(\mathcal{X},\mathcal{L}_0)}(X(x(\vec q)),X_0(x_0(\vec q,q_0));a)\nonumber\\
&=&\frac{1}{\vert G_0\vert}\sum_{\substack{d\in\mathbb{Z}_{>0},k_g\in\mathbb{Z}_{\geqslant 0},\forall g\in G_s\\\frac{\vert G_0\vert}{\vert G\vert}d-\sum\limits_{g\in G_s}k_gF_g^{(0)}\in\mathbb{Z}_{\geqslant 0}}}(-1)^{\lfloor\sum\limits_{g\in G_s}k_gF_g^{(2)}\rfloor}\nonumber\\
&&\quad\cdot\frac{1}{\Gamma\bigg(1+\frac{\vert G_0\vert}{\vert G\vert}d-\sum\limits_{g\in G_s}k_gF_g^{(0)}\bigg)}\cdot\frac{\Gamma\bigg(\sum\limits_{g\in G_s}k_gF_g^{(2)}-\Big(a-\frac{\vert G_0\vert}{\vert G\vert}\Big)d\bigg)}{\Gamma\bigg(1-ad-\sum\limits_{g\in G_s}k_gF_g^{(1)}\bigg)}\nonumber\\
&&\quad\cdot\prod\limits_{g\in G_s}\frac{q_g^{\sum\limits_{h\in G_s}[k_h-dl_h^{(0)}]b_{hg}}}{k_g!}\cdot\frac{[(-1)^fq_0]^d}{d},\label{find k_0}
\end{eqnarray}
Here we observe that:
\begin{eqnarray}
&&\bigg(\sum\limits_{g\in G_s}k_gF_g^{(2)}-\Big(a-\frac{\vert G_0\vert}{\vert G\vert}\Big)d\bigg)-\bigg(1-ad-\sum\limits_{g\in G_s}k_gF_g^{(1)}\bigg)\nonumber\\
&=&\sum\limits_{g\in G_s}k_g(F_g^{(1)}+F_g^{(2)})+\frac{\vert G_0\vert}{\vert G\vert}d-1\nonumber\\
&=&\sum\limits_{g\in G_s}k_g-\sum\limits_{g\in G_s}k_gF_g^{(0)}+\frac{\vert G_0\vert}{\vert G\vert}d-1\nonumber\\
&=&\bigg(\sum\limits_{g\in G_s}k_g-\lfloor\sum\limits_{g\in G_s}k_gF_g^{(0)}\rfloor\bigg)+\bigg(\frac{\vert G_0\vert}{\vert G\vert}d-\langle\sum\limits_{g\in G_s}k_gF_g^{(0)}\rangle\bigg)-1\nonumber\\
&\in&\mathbb{Z}_{\geqslant 0}.\label{restric}
\end{eqnarray}

Set
\begin{eqnarray}
\left\{\begin{array}{ccl}m_0&=&d,\\\\m_g&=&\sum\limits_{h\in G_s}[k_h-dl_h^{(0)}]b_{hg},\quad\forall g\in G_s.\end{array}\right.\label{29}
\end{eqnarray}
From \eqref{26}, we have
\begin{eqnarray}
\left\{\begin{array}{ccl}d&=&m_0\in\mathbb{Z}_{>0},\\\\k_h&=&m_0l_h^{(0)}+\langle \vec m,l_h\rangle\in\mathbb{Z}_{\geqslant 0},\quad\forall h\in G_s.\end{array}\right.\label{30}
\end{eqnarray}
In particular, from \eqref{11}, we have $$k_h=m_0\big[\delta_{h,i_0}+f\delta_{h,i_1}-(f+1)\delta_{h,i_2}\big]+\langle\vec  m,l_h\rangle.$$ So from \eqref{10}, we have for $j=0,1,2$: $$\sum\limits_{g\in G_s}k_gF_g^{(j)}=m_0\sum\limits_{g\in G_s}\big[\delta_{g,i_0}+f\delta_{g,i_1}-(f+1)\delta_{g,i_2}\big]F_g^{(j)}-\langle \vec m,l_j\rangle.$$ In particular, from Lemma \ref{lemma 2},
\begin{eqnarray}
&&\sum\limits_{g\in G_s}k_gF_g^{(0)}\nonumber\\
&=&\frac{\vert G_0\vert}{\vert G\vert}m_0(1-\delta_{i_0,0})-\langle \vec m,l_0\rangle\nonumber\\
&=&\frac{\vert G_0\vert}{\vert G\vert}d-m_0l_0^{(0)}-\langle\vec  m,l_0\rangle.\label{from Lemma 2}
\end{eqnarray}
 So we obtain that
\begin{eqnarray*}
&&\mathcal{F}_{0,1}^{(\mathcal{X},\mathcal{L}_0)}(X(x(\vec q)),X_0(x_0(\vec q,q_0));a)\\
&=&\frac{1}{\vert G_0\vert}\sum(-1)^{\lfloor m_0\sum\limits_{g\in G_s}l_g^{(0)}F_g^{(2)}\rfloor}\\
&&\quad\cdot\frac{1}{\Gamma\bigg(1+m_0l_0^{(0)}+\langle\vec  m,l_0\rangle\bigg)}\cdot\frac{\Gamma\bigg(m_0\Big(\sum\limits_{g\in G_s}l_g^{(0)}F_g^{(2)}-a+\frac{\vert G_0\vert}{\vert G\vert}\Big)-\langle\vec  m,l_2\rangle\bigg)}{\Gamma\bigg(1-m_0\Big(\sum\limits_{g\in G_s}l_g^{(0)}F_g^{(1)}+a\Big)+\langle\vec  m,l_1\rangle\bigg)}\\
&&\quad\cdot\prod\limits_{g\in G_s}\frac{[(-1)^{l_2^{(g)}}q_g]^{m_g}}{\Gamma\bigg(1+m_0l_g^{(0)}+\langle \vec m,l_g\rangle\bigg)}\cdot\frac{[(-1)^fq_0]^{m_0}}{m_0}.
\end{eqnarray*}

\subsection{Main theorems}

In this subsection, we compare $W$ \eqref{f>0},\eqref{f<0} with $F$ \eqref{analytic part}. Recall that $L_0$ intersects the non-compact toric curve in $Y$ given by $\overline{v_{i_1}v_{i_2}}\subset\overline{v_1v_2}$.

We observe that the domains of sum in \eqref{general f} and \eqref{analytic part} are disjoint if $i_1\neq 1$ and $i_2\neq 2$. So we only need to consider the cases $i_1=1$ or $i_2=2$, which puts restrictions on the location of $L_0$.

First we consider the effective case, i.e. $|G_0|=1$.
\begin{thm}\label{effective}
In the effective case, we have $$W(q,q_0;f)=\mathcal{F}_{0,1}^{(\mathcal{X},\mathcal{L}_0)}(X(x(\vec q)),X_0(x_0(\vec q,q_0));a).$$
\end{thm}
\begin{proof}
In this case, no integral points lie on the interior of the line segment $\overline{v_1v_2}$, and hence $i_1=1,i_2=2$. Recall that $G\neq\{0\}$ and hence $i_0\in G_s$. So we have
\begin{eqnarray*}
l_i^{(0)}=\left\{\begin{array}{cl}1,&i=i_0,\\f,&i=i_1=1,\\-(f+1),&i=i_2=2,\\0,&i\in\mathcal{S}\setminus\{i_0,1,2\},\end{array}\right.
\end{eqnarray*}
and the correspondence between framings is $$a=-f-F_{i_0}^{(1)}.$$ Then
\begin{eqnarray}
&&\mathcal{F}_{0,1}^{(\mathcal{X},\mathcal{L}_0)}(X(x(\vec q)),X_0(x_0(\vec q,q_0));a)\\
&=&\sum_{\substack{m_0\in\mathbb{Z}_{>0},\vec m=(m_g)_{g\in G_s}\in\mathbb{Z}^s_{\geqslant 0}\\m_0l_i^{(0)}+\langle\vec  m,l_i\rangle\in\mathbb{Z}_{\geqslant 0}, i\neq 1,2}}(-1)^{\lfloor m_0\sum\limits_{g\in G_s}l_g^{(0)}F_g^{(2)}\rfloor}\nonumber\\
&&\quad\cdot\frac{1}{\Gamma\bigg(1+m_0l_0^{(0)}+\langle\vec  m,l_0\rangle\bigg)}\cdot\frac{\Gamma\bigg((f+1)m_0-\langle\vec  m,l_2\rangle\bigg)}{\Gamma\bigg(1+fm_0+\langle\vec  m,l_1\rangle\bigg)}\nonumber\\
&&\quad\cdot\prod\limits_{g\in G_s}\frac{[(-1)^{l_2^{(g)}}q_g]^{m_g}}{\Gamma\bigg(1+m_0l_g^{(0)}+\langle\vec  m,l_g\rangle\bigg)}\cdot\frac{[(-1)^fq_0]^{m_0}}{m_0}.\label{F0}
\end{eqnarray}
We need to consider the implicit restriction on the domain of sum in \eqref{F0} more carefully. First observe that from \eqref{restric}, we have: $$\bigg((f+1)m_0-\langle\vec  m,l_2\rangle\bigg)-\bigg(1+fm_0+\langle\vec  m,l_1\rangle\bigg)\in\mathbb{Z}_{\geqslant 0}.$$ The $\lambda_1$- and $\lambda_2$-part of $I_\mathcal{X}$ in \eqref{findrestriction} are:
\begin{eqnarray*}
\frac{\Gamma(\sum\limits_{g\in G_s}k_gF_g^{(1)}+ad)}{\Gamma(\langle\sum\limits_{g\in G_s}k_gF_g^{(1)}\rangle+ad)}&=&\frac{\Gamma\bigg(-fm_0-\langle\vec  m,l_1\rangle\bigg)}{\Gamma\bigg(-fm_0-\lfloor m_0F_{i_0}^{(1)}\rfloor\bigg)},\\
\frac{\Gamma(\sum\limits_{g\in G_s}k_gF_g^{(2)}-(a-\frac{\vert G_0\vert}{\vert G\vert})d)}{\Gamma(\langle\sum\limits_{g\in G_s}k_gF_g^{(2)}\rangle-(a-\frac{\vert G_0\vert}{\vert G\vert})d)}&=&\frac{\Gamma\bigg((f+1)m_0-\langle \vec m,l_2\rangle\bigg)}{\Gamma\bigg((f+1)m_0-\lfloor m_0F_{i_0}^{(2)}\rfloor\bigg)}.
\end{eqnarray*}
So if $f\geqslant 0$, then $$-fm_0-\lfloor m_0F_{i_0}^{(1)}\rfloor\in\mathbb{Z}_{\leqslant 0}.$$ Therefore,
\begin{eqnarray}
-fm_0-\langle\vec  m,l_1\rangle\in\mathbb{Z}_{\leqslant 0},\textrm{ and then }(f+1)m_0-\langle\vec  m,l_2\rangle\in\mathbb{Z}_{>0}.\label{f0>0}
\end{eqnarray}
So from \eqref{f>0},\eqref{F0},\eqref{f0>0}, we have the required equality. If $f<0$, then $$(f+1)m_0-\lfloor m_0F^{(2)}_{i_0}\rfloor\in\mathbb{Z}_{\leqslant 0}.$$ Therefore,
\begin{eqnarray}
(f+1)m_0-\langle\vec  m,l_2\rangle\in\mathbb{Z}_{\leqslant 0},\textrm{ and then }fm_0+\langle\vec  m,l_1\rangle\in\mathbb{Z}_{\leqslant 0}.\label{f0<0}
\end{eqnarray}
So from \eqref{f<0},\eqref{F0},\eqref{f0<0}, we also have the required equality.
\end{proof}

Now we come to consider the ineffective case, i.e. $|G_0|\geqslant 2$.
\begin{thm}\label{ineffective}
In the ineffective case, if $L_0$ intersects the non-compact toric curve given by $\overline{v_{i_1}v_{i_2}}\subset\overline{v_1v_2}$ with $i_1=1$, then $$W(\vec q,q_0;f)=\vert G_0\vert\cdot\mathcal{F}_{0,1}^{(\mathcal{X},\mathcal{L}_0)}(X(x(\vec q)),X_0(x_0(\vec q,q_0));a),\quad\textrm{ for }f<0,$$
and if $L_0$ intersects the non-compact toric curve given by $\overline{v_{i_1}v_{i_2}}\subset\overline{v_1v_2}$ with $i_2=2$, then $$W(\vec q,q_0;f)=\vert G_0\vert\cdot\mathcal{F}_{0,1}^{(\mathcal{X},\mathcal{L}_0)}(X(x(\vec q)),X_0(x_0(\vec q,q_0));a),\quad\textrm{ for }f\geqslant 0.$$ The above equalities hold after moding out nonanalytic parts.
\end{thm}
\begin{proof}

In this case, there are integral points in the interior of $\overline{v_1v_2}$. So $L_0$ can be located on several different toric curves.

If $L_0$ is located on the toric curve given by $\overline{v_{i_1}v_{i_2}}$ with $i_1=1$, as illustrated below:

\begin{center}
\setlength{\unitlength}{1cm}
\begin{picture}(10.5,5)(-7.5,-1.5)
\drawline(-2,0)(0,-1)(1,1)(-2,0)
\drawline(0,-1)(-7,0)(2,3)
\dottedline{0.1}(1,1)(2,3)
\put(-7,-0.3){$v_0$}
\put(-2.5,0){$v_{i_0}$}
\put(0,-1.3){$v_{i_1}=v_1$}
\put(1.1,1){$v_{i_2}$}
\put(2,2.7){$v_2$}
\put(-7,0){\circle*{0.1}}
\put(2,3){\circle*{0.1}}
\put(-2,0){\circle*{0.1}}
\put(0,-1){\circle*{0.1}}
\put(1,1){\circle*{0.1}}
\put(0.5,0){$\backslash$}
\put(-3,-2){Figure 11}
\end{picture}
\end{center}

Then $i_0\in\{0,g_1,\cdots,g_s\},i_2\in G_s$, and therefore
\begin{eqnarray*}
l_i^{(0)}=\left\{\begin{array}{cl}1,&i=i_0,\\f,&i=i_1=1,\\-(f+1),&i=i_2,\\0,&i\in\mathcal{S}\setminus\{i_0,1,i_2\},\end{array}\right.
\end{eqnarray*}
and the correspondence between framings is $$a=-f-\sum\limits_{g\in G_s}l_g^{(0)}F_g^{(1)}.$$ So
\begin{eqnarray}
&&\mathcal{F}_{0,1}^{(\mathcal{X},\mathcal{L}_0)}(X(x(\vec q)),X_0(x_0(\vec q,q_0));a)\\
&=&\frac{1}{\vert G_0\vert}\sum_{\substack{m_0\in\mathbb{Z}_{>0},m_g\in\mathbb{Z}_{\geqslant 0},\forall g\in G_s\\m_0l_i^{(0)}+\langle\vec  m,l_i\rangle\in\mathbb{Z}_{\geqslant 0}, i\neq 1,2}}(-1)^{\lfloor m_0\sum\limits_{g\in G_s}l_g^{(0)}F_g^{(2)}\rfloor}\nonumber\\
&&\quad\cdot\frac{1}{\Gamma\bigg(1+m_0l_0^{(0)}+\langle\vec  m,l_0\rangle\bigg)}\cdot\frac{\Gamma\bigg(-\langle\vec  m,l_2\rangle\bigg)}{\Gamma\bigg(1+fm_0+\langle\vec  m,l_1\rangle\bigg)}\nonumber\\
&&\quad\cdot\prod\limits_{g\in G_s}\frac{[(-1)^{l_2^{(g)}}q_g]^{m_g}}{\Gamma\bigg(1+m_0l_g^{(0)}+\langle\vec  m,l_g\rangle\bigg)}\cdot\frac{[(-1)^fq_0]^{m_0}}{m_0}.\label{FL}
\end{eqnarray}
We need to consider the implicit restriction on the domain of sum in \eqref{FL} more carefully. First observe that from \eqref{restric}, we have: $$\bigg(-\langle\vec  m,l_2\rangle\bigg)-\bigg(1+fm_0+\langle\vec  m,l_1\rangle\bigg)\in\mathbb{Z}_{\geqslant 0}.$$ The $\lambda_1$- and $\lambda_2$-part of $I_\mathcal{X}$ in \eqref{findrestriction} are:
\begin{eqnarray*}
\frac{\Gamma(\sum\limits_{g\in G_s}k_gF_g^{(1)}+ad)}{\Gamma(\langle\sum\limits_{g\in G_s}k_gF_g^{(1)}\rangle+ad)}&=&\frac{\Gamma\bigg(-fm_0-\langle\vec  m,l_1\rangle\bigg)}{\Gamma\bigg(-fm_0-\lfloor m_0\sum\limits_{g\in G_s}l_g^{(0)}F_g^{(1)}\rfloor\bigg)},\\
\frac{\Gamma(\sum\limits_{g\in G_s}k_gF_g^{(2)}-(a-\frac{\vert G_0\vert}{\vert G\vert})d)}{\Gamma(\langle\sum\limits_{g\in G_s}k_gF_g^{(2)}\rangle-(a-\frac{\vert G_0\vert}{\vert G\vert})d)}&=&\frac{\Gamma\bigg(-\langle\vec  m,l_2\rangle\bigg)}{\Gamma\bigg(-\lfloor m_0\sum\limits_{g\in G_s}l_g^{(0)}F_g^{(2)}\rfloor\bigg)}.
\end{eqnarray*}
So we see that $-\lfloor m_0\sum\limits_{g\in G_s}l_g^{(0)}F_g^{(2)}\rfloor\in\mathbb{Z}_{\leqslant 0}$. Therefore,
\begin{eqnarray}
 -\langle \vec m,l_2\rangle\in\mathbb{Z}_{\leqslant 0},\textrm{ and then }fm_0+\langle\vec  m,l_1\rangle\in\mathbb{Z}_{<0}.\label{fL<0}
\end{eqnarray}
From \eqref{f<0},\eqref{FL},\eqref{fL<0}, the required equality holds for $f<0$, up to sign ambiguity. If $f\geqslant 0$, then the two domains of sum for $W$ and $F$ are disjoint.

Similarly, if $L_0$ is located on the toric curve given by $\overline{v_{i_1}v_{i_2}}$ with $i_2=2$, then we have the required equality for $f\geqslant 0$.
\end{proof}

For $\vert G_0\vert\geqslant 3$, $L_0$ can be located on the toric curve given by $\overline{v_{i_1}v_{i_2}}$ with $i_1\neq 1,i_2\neq 2$. However, in this case the domains of sum for $W$ and $F$ are disjoint.

\subsection{A conjecture on the effective case}

For the effective case, i.e. $|G_0|=1$, since $G$ acts effectively on the $z_0$-axis, it follows that $G$ is cyclic. From Lemma \eqref{lemma 2}, we have $F_{i_0}^{(0)}=\frac{1}{|G|}$ and hence $i_0$ is a generator of $G$. We shall set $g_1=i_0$. In our effective case, we have $i_1=1,i_2=1$, and we obtain
\begin{eqnarray*}
l_i^{(0)}=\left\{\begin{array}{cl}1,&i=g_1,\\f,&i=i_1=1,\\-(f+1),&i=i_2=2,\\0,&i\in\mathcal{S}\setminus\{g_1,1,2\},\end{array}\right.
\end{eqnarray*}
Note that from \eqref{find k_0}, we have $$k_0:=\frac{d}{|G|}-\sum\limits_{h\in G_s}k_hF_h^{(0)}\in\mathbb Z_{\geqslant 0}.$$ Then it is not difficult to see that
\begin{eqnarray*}
\left\{\begin{array}{ccl}m_g&=&\sum\limits_{h\in G_s}(b_{hg}-b_{g_1g}\cdot|G|F_h^{(0)})k_h-b_{g_1g}|G|k_0,\quad\forall g\in G_s,\\\\m_0&=&\sum\limits_{h\in G_s}|G|F_h^{(0)}k_h+|G|k_0,\end{array}\right.
\end{eqnarray*}
So if all the coefficients of $k_0,k_{g_1},\cdots,k_{g_s}$ are non-negative integers, then $m_g$'s are all non-negative integers. We can rewrite this sufficient condition as follows. Note that
\begin{eqnarray*}
\left\{\begin{array}{ccl}k_g&=&\sum\limits_{h\in G_s}l_g^{(h)}m_h+\delta_{g,g_1}m_0,\quad\forall g\in G_s,\\\\k_0&=&\sum\limits_{h\in G_s}l_0^{(h)}m_h.\end{array}\right.
\end{eqnarray*}
In the matrix form, we have
\ben
[m_{g_1},\cdots,m_{g_s},m_0]= [k_{g_1},\cdots,k_{g_s},k_0]
\begin{bmatrix} l_{g_1}^{(g_1)}&l_{g_2}^{(g_1)}&\cdots&l_{g_s}^{(g_1)}&l_{0}^{(g_1)} \\\vdots&\vdots&&\vdots&\vdots\\ l_{g_1}^{(g_s)}&l_{g_2}^{(g_s)}&\cdots&l_{g_s}^{(g_s)}&l_{0}^{(g_s)} \\ 1 &0& \cdots &0&0
\end{bmatrix}
\een
Now we express the required sufficient condition in the following conjecture:
\begin{conj}
The entries of the inverse of the above matrix are non-negative integers.
\end{conj}

By the general theory of toric geoemtry, it is not difficult to show the integrality. Unfortunately, we could not find a proof for the non-negativeness.

\section{Appendix}\label{Appendix}

In this appendix, we demonstrate our results by some examples. 
We shall let $\xi_n:=e^{\frac{2\pi\sqrt{-1}}{n}}$. 
For the convenience of the reader,
we repeat some details which have appeared in earlier sections.

\subsection{Example: $\mathbb{C}^3/\mathbb{Z}_3(1,1,1)$}
Let $\xi$ be the generator of $\mathbb{Z}_3$. The action is given by $$\xi.(z_0,z_1,z_2)=(\xi_3\cdot z_0,\xi_3\cdot z_1,\xi_3\cdot z_2).$$ Hence the small part of $\mathbb{Z}_3$ consists of a single element $\xi$, with $$F_\xi^{(0)}=F_\xi^{(1)}=F_\xi^{(2)}=\frac{1}{3}.$$ An integral basis of the lattice of invariants is 
\ben
[\varepsilon_0,\varepsilon_1,\varepsilon_2]
=[e_0,e_1,e_2]
\begin{bmatrix} 0&0&1 \\ 0 &-1&1 \\ 3 & 1 &1
\end{bmatrix}
\een
So let
\begin{eqnarray*}
\tilde{v}_0=\left[\begin{array}{c}0\\0\\1\end{array}\right],\tilde{v}_1=\left[\begin{array}{c}0\\-1\\1\end{array}\right],\tilde{v}_2=\left[\begin{array}{c}3\\1\\1\end{array}\right],\tilde{v}_{\xi}=\left[\begin{array}{c}1\\0\\1\end{array}\right],
\end{eqnarray*}
and the toric resolution $X$ is given by Figure 2.

Let $D_i$ be the corresponding toric divisor for $i\in\{0,1,2,\xi\}$. Then the intersection numbers are:
\begin{displaymath}
\begin{array}{c|cccc}
&D_0&D_1&D_2&D_{\xi}\\
\hline
D_\xi\cdot D_0&1&1&1&-3\\
D_\xi\cdot D_1&1&1&1&-3\\
D_\xi\cdot D_2&1&1&1&-3
\end{array}
\end{displaymath}
From the dual graph
\begin{center}
\setlength{\unitlength}{1cm}
\begin{picture}(0,6.5)(0,-3.5)
\drawline(-1,0)(0,0)(1,-1)(0,1)(0,0)
\drawline(-0.5,2.5)(0,1)
\drawline(1.5,-1.75)(1,-1)
\dottedline[.]{0.2}(-2,0)(-1,0)
\dottedline[.]{0.2}(-0.5,2.5)(-0.75,3.25)
\dottedline[.]{0.2}(1.5,-1.75)(2,-2.5)
\put(-1,0.5){$D_0$}
\put(-1,-0.5){$D_1$}
\put(1,0){$D_2$}
\put(0,0){$D_\xi$}
\put(0,-2.5){Figure 12}
\end{picture}
\end{center}
we can take $C_\xi = D_\xi \cdot D_0$ or $D_\xi \cdot D_1$ or $D_\xi \cdot D_2$,
and one can see the charge vector is $l^{(\xi)}=(1,1,1,-3)$, 
which is a basis of the lattice of invariants $$\mathbb{L}=\{(l_0,l_1,l_2,l_\xi)\in\mathbb{Z}^4\mid \sum\limits_{j=0}^2l_j\tilde{v}_j+l_\xi\tilde{v}_\xi=0\}.$$ Suppose that $L_0$ intersects the noncompact toric curve given by $\overline{v_1v_2}$ with framing $f\in\mathbb{Z}_{\geqslant 0}$. Then the charge vectors for $(X,L_0)$ are: 
\begin{eqnarray*}
\begin{array}{cccc|cc}
\tilde l^{(\xi)}=(1,&1,&1,&-3,&0,&0),\\
\tilde l^{(0)}=(0,&f,&-f-1,&1,&1,&-1).
\end{array}
\end{eqnarray*}
Then the associated variables are:
\begin{eqnarray*}
q_\xi&=&z_0z_1z_2\cdot z_\xi^{-3},\\
q_0&=&z_1^fz_2^{-f-1}\cdot z_\xi\cdot\frac{z_+}{z_-},
\end{eqnarray*}
and the corresponding differential operators are:
\begin{eqnarray*}
\mathcal D_1&=&\Theta_{q_\xi}\bigg(\Theta_{q_\xi}+f\Theta_{q_0}\bigg)\bigg(\Theta_{q_\xi}-(f+1)\Theta_{q_0}\bigg)-q_\xi\bigg(-3\Theta_{q_\xi}+\Theta_{q_0}\bigg)_3,\\
\mathcal D_0&=&\Theta_{q_0}\bigg(-3\Theta_{q_\xi}+\Theta_{q_0}\bigg)\bigg(\Theta_{q_\xi}+f\Theta_{q_0}\bigg)_f+q_0\Theta_{q_0}\bigg(\Theta_{q_\xi}-(f+1)\Theta_{q_0}\bigg)_{f+1}.
\end{eqnarray*}
So the open-close mirror map is given by:
\begin{eqnarray*}
\Omega_\xi&=&\log q_\xi+3\sum\limits_{m_\xi=1}^\infty\frac{(3m_\xi-1)!}{(m_\xi!)^3}(-q_\xi)^{m_\xi},\\
\Omega_0&=&\log q_0-\sum\limits_{m_\xi=1}^\infty\frac{(3m_\xi-1)!}{(m_\xi!)^3}(-q_\xi)^{m_\xi},
\end{eqnarray*}
and the superpotential is: $$W=\sum\limits_{\substack{m_0\geqslant 1,m_\xi\geqslant0,\\(f+1)m_0>m_\xi}}\frac{1}{\Gamma(1+m_\xi)}\cdot\frac{1}{\Gamma(1+m_0-3m_\xi)}\cdot\frac{\Gamma(m_0(f+1)-m_\xi)}{\Gamma(1+m_0f+m_\xi)}(-q_\xi)^{m_\xi}\frac{[(-1)^fq_0]^{m_0}}{m_0}.$$

For the orbifold $\mathcal X=[\mathbb C^3/\mathbb Z_3(1,1,1)]$, note that the small part of the group consists of $\xi$, with $$F_\xi^{(0)}=F_\xi^{(1)}=F_\xi^{(2)}=\frac{1}{3}.$$ So the charge vector for $\mathcal X$ is $(-1,-1,-1,3)$, which is a basis of $\mathbb L$. Note that $\mathcal L_0$ is located on the $z_0$-axis, and the corresponding weight is $$(\lambda_0,\lambda_1,\lambda_2)=(\frac{1}{3},-a,a-\frac{1}{3}).$$ So the charge vectors for $(\mathcal X,\mathcal L_0)$ are:
\begin{eqnarray*}
\begin{array}{cccccc|cc}
\hat l^{(\xi)}&=&(-1,&-1,&-1,&3,&0,&0),\\
\hat l^{(0)}&=&(1,&-3a,&3a-1,&0,&3,&-3).
\end{array}
\end{eqnarray*}
The corresponding variables are:
\begin{eqnarray*}
x_\xi&=&z_0^{-\frac{1}{3}}z_1^{-\frac{1}{3}}z_2^{-\frac{1}{3}}\cdot z_\xi,\\
x_0&=&z_0^{\frac{1}{3}}z_1^{-a}z_2^{a-\frac{1}{3}}\cdot\frac{z_+}{z_-},
\end{eqnarray*}
the orbifold open-close mirror map is given by:
\begin{eqnarray*}
X_\xi&=&-\sum\limits_{\substack{k_\xi\geqslant 0,\\3\mid k_\xi-1}}\bigg(\frac{\Gamma(\frac{k_\xi}{3})}{\Gamma(\frac{1}{3})}\bigg)^3\cdot\frac{(-x_\xi)^{k_\xi}}{k_{\xi}!},\\
X_0&=&(-1)^{1+\lambda_2}x_0,
\end{eqnarray*}
and the orbifold disc potential is given by:
\begin{eqnarray*}
&&\mathcal F_{0,1}^{(\mathcal X,\mathcal L_0)}(X(x),X_0(x_0);a)\\
&=&\sum\limits_{\substack{d\geqslant 1,k_\xi\geqslant 0,\\\langle\frac{k_\xi}{3}\rangle=\langle\frac{d}{3}\rangle}}(-1)^{\lfloor\frac{k_\xi}{3}\rfloor}\frac{1}{\Gamma(1-\frac{k_\xi}{3}+\lambda_0d)}\cdot\frac{\Gamma(\frac{k_\xi}{3}-\lambda_2d)}{\Gamma(1-\frac{k_\xi}{3}+\lambda_1d)}\cdot\frac{x_\xi^{k_\xi}}{k_\xi!}\cdot\frac{((-1)^{1+\lambda_2}x_0)^d}{d}
\end{eqnarray*}

To compare $W$ with $\mathcal F_{0,1}^{(\mathcal X,\mathcal L_0)}$, first note that $\{\tilde l^{(\xi)},\tilde l^{(0)}\}$ and $\{\hat l^{(\xi)},\hat l^{(0)}\}$ span the same linear subspace in $\mathbb Q^{6}$ if and only if $$a=-f-\frac{1}{3},$$ which gives the correspondence between $a$ and $f$, and so we obtain the change of variables: 
\begin{eqnarray*}
x_\xi&=&q_\xi^{-\frac{1}{3}},\\
x_0&=&(-1)^{1+\lambda_2+f}q_0q_\xi^{\frac{1}{3}},
\end{eqnarray*}
where the phase factor $(-1)^{1+\lambda_2+f}$ is included for convenience of comparison. Via the above identification, we have:
\begin{eqnarray*}
&&\mathcal F_{0,1}^{(\mathcal X,\mathcal L_0)}(X(x(q)),X_0(x_0(q,q_0));a)\\
&=&\sum\limits_{\substack{d\geqslant 1,k_\xi\geqslant 0,\\\frac{d}{3}-\frac{k_\xi}{3}\in\mathbb Z_{\geqslant 0}}}(-1)^{\lfloor\frac{k_\xi}{3}\rfloor}\frac{1}{\Gamma(1-\frac{k_\xi}{3}+\frac{d}{3})}\cdot\frac{\Gamma\bigg(\frac{k_\xi}{3}-(a-\frac{1}{3})d\bigg)}{\Gamma(1-\frac{k_\xi}{3}-ad)}\cdot\frac{q_\xi^{\frac{d-k_\xi}{3}}}{k_\xi!}\cdot\frac{((-1)^{f}q_0)^d}{d}
\end{eqnarray*}
Set 
\begin{eqnarray*}
\left\{\begin{array}{ccc}m_0&=&d\in\mathbb Z_{\geqslant 1},\\\\m_\xi&=&\frac{d-k_\xi}{3}\in\mathbb Z_{\geqslant 0},\end{array}\right.
\end{eqnarray*}
and then we have
\begin{eqnarray*}
&&\mathcal F_{0,1}^{(\mathcal X,\mathcal L_0)}(X(x(\vec q)),X_0(x_0(\vec q,q_0));a)\\
&=&\sum\limits_{\substack{m_0\geqslant 1,m_\xi\geqslant 0,\\m_0-3m_\xi\geqslant 0}}(-1)^{\lfloor\frac{m_0}{3}-m_\xi\rfloor}\frac{1}{\Gamma(1+m_\xi)}\cdot\frac{\Gamma\bigg((-a+\frac{2}{3})m_0-m_\xi\bigg)}{\Gamma\bigg(1-(a+\frac{1}{3})m_0+m_\xi\bigg)}\cdot\frac{q_\xi^{m_\xi}}{(m_0-3m_\xi)!}\cdot\frac{((-1)^{f}q_0)^{m_0}}{m_0}\\
&=&W(q,q_0;f)\quad\quad(\textrm{ via }a=-f-\frac{1}{3}).
\end{eqnarray*}
We can see that the last equality also holds via $a=f+\frac{2}{3}$. The choice $a=f+\frac{2}{3}$ was observed in \cite{BC}.

\subsection{Example: $\mathbb{C}^3/\mathbb{Z}_4(2,1,1)$}
Let $\xi$ be the generator of $\mathbb{Z}_4$. The action is given by $$\xi.(z_0,z_1,z_2)=(\xi_4^2\cdot z_0,\xi_4\cdot z_1,\xi_4\cdot z_2).$$ Hence the small part of $\mathbb{Z}_4$ consists of $\xi,2\xi$, with
\begin{eqnarray*}
(F_{\xi}^{(0)},F_{\xi}^{(1)},F_{\xi}^{(2)})&=&(\frac{1}{2},\frac{1}{4},\frac{1}{4}),\\
(F_{2\xi}^{(0)},F_{2\xi}^{(1)},F_{2\xi}^{(2)})&=&(0,\frac{1}{2},\frac{1}{2}).
\end{eqnarray*} An integral basis of the lattice of invariants is 
\ben
[\varepsilon_0,\varepsilon_1,\varepsilon_2]
=[e_0,e_1,e_2]
\begin{bmatrix} 0&0&1 \\ 0 &-1&1\\ 4 & 1 &1
\end{bmatrix}
\een
Let
\begin{eqnarray*}
\tilde{v}_0=\left[\begin{array}{c}0\\0\\1\end{array}\right],\tilde{v}_1=\left[\begin{array}{c}0\\-1\\1\end{array}\right],\tilde{v}_2=\left[\begin{array}{c}4\\1\\1\end{array}\right],\tilde{v}_{\xi}=\left[\begin{array}{c}1\\0\\1\end{array}\right],\tilde{v}_{2\xi}=\left[\begin{array}{c}2\\0\\1\end{array}\right].
\end{eqnarray*}
Then the toric resolution $X$ is given by Figure 3.

Let $D_i$ be the corresponding toric divisor for $i\in\{0,1,2,\xi,2\xi\}$. Then the intersection numbers are:
\begin{displaymath}
\begin{array}{c|ccccc}
&D_0&D_1&D_2&D_{\xi}&D_{2\xi}\\
\hline
D_\xi\cdot D_0&2&1&1&-4&0\\
D_\xi\cdot D_1&1&0&0&-2&1\\
D_\xi\cdot D_2&1&0&0&-2&1\\
D_\xi\cdot D_{2\xi}&0&1&1&0&-2
\end{array}
\end{displaymath}
From the dual graph 
\begin{center}
\setlength{\unitlength}{1cm}
\begin{picture}(0,10)(0,-3)
\drawline(-1,0)(0,0)(0,3)(1,0)(1,-1)(0,0)
\drawline(-0.5,5)(0,3)
\drawline(1,-1)(1.5,-2)
\drawline(1,0)(1.5,-1)
\dottedline[.]{0.2}(-2,0)(-1,0)
\dottedline[.]{0.2}(-1,7)(-0.5,5)
\dottedline[.]{0.2}(2,-3)(1.5,-2)
\dottedline[.]{0.2}(2,-2)(1.5,-1)
\put(-1,1){$D_0$}
\put(-1,-1){$D_1$}
\put(2,0){$D_2$}
\put(.3,0){$D_\xi$}
\put(1.2,-1.5){$D_{2\xi}$}
\put(-1,-2){Figure 13}
\end{picture}
\end{center}
we can choose $C_\xi = D_\xi \cdot D_1$ or $D_\xi \cdot D_2$ and 
$C_{2\xi} = D_\xi \cdot D_{2\xi}$,
and then the charge vectors are:
\begin{displaymath}
\begin{array}{ccccccc}
l^{(\xi)}&=&(1,&0,&0,&-2,&1),\\
l^{(2\xi)}&=&(0,&1,&1,&0,&-2),
\end{array}
\end{displaymath}
which form a basis of the lattice of invariants $$\mathbb{L}=\{(l_0,l_1,l_2,l_\xi,l_{2\xi})\in\mathbb{Z}^5\mid \sum\limits_{j=0}^2l_j\tilde{v}_j+l_\xi\tilde{v}_\xi+l_{2\xi}\tilde{v}_{2\xi}=0\}.$$

Suppose that $L_0$ intersects the noncompact toric curve given by $\overline{v_{2\xi}v_{2}}$ with $f\geqslant 0$. The charge vectors for $(X,L_0)$ are:
\begin{displaymath}
\begin{array}{ccccccc|cc}
\tilde l^{(\xi)}&=&(1,&0,&0,&-2,&1,&0,&0),\\
\tilde l^{(2\xi)}&=&(0,&1,&1,&0,&-2,&0,&0),\\
\tilde l^{(0)}&=&(0,&0,&-f-1,&1,&f,&1,&-1).
\end{array}
\end{displaymath}
Then the associated variables are:
\begin{eqnarray*}
q_\xi&=&z_0\cdot z_\xi^{-2}z_{2\xi},\\
q_{2\xi}&=&z_1z_2\cdot z_{2\xi}^{-2},\\
q_0&=&z_2^{-f-1}\cdot z_\xi z_{2\xi}^f\cdot\frac{z_+}{z_-},
\end{eqnarray*}
and the corresponding differential operators are:
\begin{eqnarray*}
\mathcal D_1&=&\Theta_{q_\xi}\bigg(\Theta_{q_\xi}-2\Theta_{q_{2\xi}}+f\Theta_{q_0}\bigg)-q_\xi\bigg(-2\Theta_{q_\xi}+\Theta_{q_0}\bigg)_2,\\
\mathcal D_2&=&\Theta_{q_{2\xi}}\bigg(\Theta_{q_{2\xi}}-(f+1)\Theta_{q_0}\bigg)-q_{2\xi}\bigg(\Theta_{q_\xi}-2\Theta_{q_{2\xi}}+f\Theta_{q_0}\bigg)_2,\\
\mathcal D_0&=&\Theta_{q_0}\bigg(-2\Theta_{q_\xi}+\Theta_{q_0}\bigg)\bigg(\Theta_{q_\xi}-2\Theta_{q_{2\xi}}+f\Theta_{q_0}\bigg)_f+q_0\Theta_{q_0}\bigg(\Theta_{q_{2\xi}}-(f+1)\Theta_{q_0}\bigg)_{f+1}.
\end{eqnarray*}
So the open-close mirror map is given by:
\begin{eqnarray*}
\Omega_\xi&=&\log q_\xi+2\sum\limits_{m_\xi>0,m_{2\xi}\geqslant 0}\frac{(2m_\xi-1)!}{m_\xi!(m_{2\xi}!)^2(m_\xi-2m_{2\xi})!}q_\xi^{m_\xi}q_{2\xi}^{m_{2\xi}}-\sum\limits_{m_{2\xi}>0}\frac{(2m_{2\xi}-1)!}{(m_{2\xi}!)^2}q_{2\xi}^{m_{2\xi}},\\
\Omega_{2\xi}&=&\log q_{2\xi}+2\sum\limits_{m_{2\xi}>0}\frac{(2m_{2\xi}-1)!}{(m_{2\xi}!)^2}q_{2\xi}^{m_{2\xi}},\\
\Omega_0&=&\log q_0-\sum\limits_{m_\xi>0,m_{2\xi}\geqslant 0}\frac{(2m_\xi-1)!}{m_\xi!(m_{2\xi}!)^2(m_\xi-2m_{2\xi})!}q_\xi^{m_\xi}q_{2\xi}^{m_{2\xi}}-f\sum\limits_{m_{2\xi}>0}\frac{(2m_{2\xi}-1)!}{(m_{2\xi}!)^2}q_{2\xi}^{m_{2\xi}},
\end{eqnarray*}
and the superpotential is:
\begin{eqnarray*}
W_{2\xi,2}&=&\sum\limits_{\substack{m_0\geqslant 1,\\m_\xi,m_{2\xi}\geqslant 0,\\(f+1)m_0-m_{2\xi}>0}}\frac{1}{\Gamma(1+m_\xi)\Gamma(1+m_{2\xi})}\\
&&\quad\cdot\frac{1}{\Gamma(1+m_0-2m_{\xi})}\cdot\frac{\Gamma(m_0(f+1)-m_{2\xi})}{\Gamma(1+m_0f+m_\xi-2m_{2\xi})}q_\xi^{m_\xi}(-q_{2\xi})^{m_{2\xi}}\frac{[(-1)^fq_0]^{m_0}}{m_0}.
\end{eqnarray*}

For the orbifold $\mathcal X=[\mathbb C^3/\mathbb Z_4(2,1,1)]$, note that the small part of the group consists of $\xi,2\xi$, with
\begin{eqnarray*}
(F_{\xi}^{(0)},F_{\xi}^{(1)},F_{\xi}^{(2)})&=&(\frac{1}{2},\frac{1}{4},\frac{1}{4}),\\
(F_{2\xi}^{(0)},F_{2\xi}^{(1)},F_{2\xi}^{(2)})&=&(0,\frac{1}{2},\frac{1}{2}).
\end{eqnarray*}
So the charge vectors for $\mathcal X$ are:
\begin{eqnarray*}
\begin{array}{ccccc}
(-2,&-1,&-1,&4,&0),\\
(0,&-1,&-1,&0,&2),
\end{array}
\end{eqnarray*}
which form a basis of $\mathbb L\otimes\mathbb Q$. Note that $\mathcal L_0$ is located on the $z_0$-axis, and the corresponding weight is $$(\lambda_0,\lambda_1,\lambda_2)=(\frac{1}{2},-a,a-\frac{1}{2}).$$ So the charge vectors for $(\mathcal X,\mathcal L_0)$ are:
\begin{eqnarray*}
\begin{array}{ccccccc|cc}
\hat l^{(\xi)}&=&(-2,&-1,&-1,&4,&0,&0,&0),\\
\hat l^{(2\xi)}&=&(0,&-1,&-1,&0,&2,&0,&0),\\
\hat l^{(0)}&=&(1,&-2a,&2a-1,&0,&0,&2,&-2).
\end{array}
\end{eqnarray*}
The corresponding variables are:
\begin{eqnarray*}
x_\xi&=&z_0^{-\frac{1}{2}}z_1^{-\frac{1}{4}}z_2^{-\frac{1}{4}}\cdot z_\xi,\\
x_{2\xi}&=&z_1^{-\frac{1}{2}}z_2^{-\frac{1}{2}}\cdot z_{2\xi},\\
x_0&=&z_0^{\frac{1}{2}}z_1^{-a}z_2^{a-\frac{1}{2}}\cdot\frac{z_+}{z_-},
\end{eqnarray*}
the orbifold open-close mirror map is given by:
\begin{eqnarray*}
X_\xi&=&-\sum\limits_{\substack{k_\xi,k_{2\xi}\geqslant 0,\\4\mid k_\xi+2k_{2\xi}-1}}\frac{\Gamma(\frac{k_\xi}{2})}{\Gamma(\frac{1}{2})}\bigg(\frac{\Gamma(\frac{k_\xi}{4}+\frac{k_{2\xi}}{2})}{\Gamma(\frac{1}{4})}\bigg)^2\cdot\frac{(-x_\xi)^{k_\xi}}{k_{\xi}!}\cdot\frac{(-x_{2\xi})^{k_{2\xi}}}{k_{2\xi}!},\\
X_{2\xi}&=&-\sum\limits_{\substack{k_{2\xi}\geqslant 0,\\2\mid k_{2\xi}-1}}\bigg(\frac{\Gamma(\frac{k_{2\xi}}{2})}{\Gamma(\frac{1}{2})}\bigg)^2\cdot\frac{(-x_{2\xi})^{k_{2\xi}}}{k_{2\xi}},\\
X_0&=&(-1)^{1+\lambda_2}x_0,
\end{eqnarray*}
and the orbifold disc potential is given by:
\begin{eqnarray*}
&&\mathcal F_{0,1}^{(\mathcal X,\mathcal L_0)}(X(x),X_0(x_0);a)\\
&=&\frac{1}{2}\sum\limits_{\substack{d\geqslant 1,k_\xi,k_{2\xi}\geqslant 0,\\\langle\frac{k_\xi}{2}\rangle=\langle\frac{d}{2}\rangle}}(-1)^{\lfloor\frac{k_\xi}{4}+\frac{k_{2\xi}}{2}\rfloor}\frac{1}{\Gamma(1-\frac{k_\xi}{2}+\lambda_0d)}\cdot\frac{\Gamma(\frac{k_\xi}{4}+\frac{k_{2\xi}}{2}-\lambda_2d)}{\Gamma(1-\frac{k_\xi}{4}-\frac{k_{2\xi}}{2}+\lambda_1d)}\\
&&\quad\quad\cdot\frac{x_{\xi}^{k_{\xi}}}{k_{\xi}!}\cdot\frac{x_{2\xi}^{k_{2\xi}}}{k_{2\xi}!}\cdot\frac{((-1)^{1+\lambda_2}x_0)^d}{d}.
\end{eqnarray*}

To compare $W$ with $\mathcal F_{0,1}^{(\mathcal X,\mathcal L_0)}$, first note that $\{\tilde l^{(\xi)},\tilde l^{(2\xi)},\tilde l^{(0)}\}$ and $\{\hat l^{(\xi)},\hat l^{(2\xi)},\hat l^{(0)}\}$ span the same linear subspace in $\mathbb Q^{7}$ if and only if $$a=-\frac{f}{2}-\frac{1}{4},$$ which gives the correspondence between $a$ and $f$, and so we obtain the change of variables: 
\begin{eqnarray*}
x_\xi&=&q_\xi^{-\frac{1}{2}}q_{2\xi}^{-\frac{1}{4}},\\
x_{2\xi}&=&q_{2\xi}^{-\frac{1}{2}},\\
x_0&=&(-1)^{1+\lambda_2+f}q_0q_\xi^{\frac{1}{2}}q_{2\xi}^{\frac{f}{2}+\frac{1}{4}},
\end{eqnarray*}
where the phase factor $(-1)^{1+\lambda_2+f}$ is included for convenience of comparison. Via the above identification, we have:
\begin{eqnarray*}
&&\mathcal F_{0,1}^{(\mathcal X,\mathcal L_0)}(X(x(\vec q)),X_0(x_0(\vec q,q_0));a)\\
&=&\frac{1}{2}\sum\limits_{\substack{d\geqslant 1,k_\xi,k_{2\xi}\geqslant 0,\\\frac{d}{2}-\frac{k_\xi}{2}\in\mathbb Z_{\geqslant 0}}}(-1)^{\lfloor\frac{k_\xi}{4}+\frac{k_{2\xi}}{2}\rfloor}\frac{1}{\Gamma(1-\frac{k_\xi}{2}+\frac{d}{2})}\cdot\frac{\Gamma\bigg(\frac{k_\xi}{4}+\frac{k_{2\xi}}{2}-(a-\frac{1}{2})d\bigg)}{\Gamma(1-\frac{k_\xi}{4}-\frac{k_{2\xi}}{2}-ad)}\\
&&\quad\quad\cdot\frac{q_\xi^{\frac{d-k_\xi}{2}}}{k_\xi!}\cdot\frac{q_{2\xi}^{(\frac{f}{2}+\frac{1}{4})d-\frac{k_\xi}{4}-\frac{k_{2\xi}}{2}}}{k_{2\xi}!}\cdot\frac{((-1)^{f}q_0)^d}{d}\\
&=&\frac{(-1)^f}{2}q_0\sum\limits_{k_{2\xi}\geqslant 0}(-1)^{\lfloor\frac{1}{4}+\frac{k_{2\xi}}{2}\rfloor}\frac{\Gamma(1+\frac{f}{2}+\frac{k_{2\xi}}{2})}{\Gamma(1+\frac{f}{2}-\frac{k_{2\xi}}{2})}\cdot\frac{q_{2\xi}^{\frac{f}{2}-\frac{k_{2\xi}}{2}}}{k_{2\xi}!}\\
&&\quad\quad+\frac{q_0^2q_\xi}{4}\sum\limits_{k_{2\xi}\geqslant 0}(-1)^{\lfloor\frac{k_{2\xi}}{2}\rfloor}\frac{\Gamma(\frac{3}{2}+f+\frac{k_{2\xi}}{2})}{\Gamma(\frac{3}{2}-f-\frac{k_{2\xi}}{2})}\cdot\frac{q_{2\xi}^{\frac{1}{2}+f-\frac{k_{2\xi}}{2}}}{k_{2\xi}!}\\
&&\quad\quad+\frac{q_0^2}{8}\sum\limits_{k_{2\xi}\geqslant 0}(-1)^{\lfloor\frac{1}{2}+\frac{k_{2\xi}}{2}\rfloor}\frac{\Gamma(2+f+\frac{k_{2\xi}}{2})}{\Gamma(1+f-\frac{k_{2\xi}}{2})}\cdot\frac{q_{2\xi}^{f-\frac{k_{2\xi}}{2}}}{k_{2\xi}!}+\cdots
\end{eqnarray*}
Set 
\begin{eqnarray*}
\left\{\begin{array}{ccl}m_0&=&d\in\mathbb Z_{\geqslant 1},\\\\m_\xi&=&\frac{d-k_\xi}{2}\in\mathbb Z_{\geqslant 0},\\\\m_{2\xi}&=&(\frac{f}{2}+\frac{1}{4})d-\frac{k_\xi}{4}-\frac{k_{2\xi}}{2},\end{array}\right.
\end{eqnarray*}
and then we have
\begin{eqnarray*}
&&\mathcal F_{0,1}^{(\mathcal X,\mathcal L_0)}(X(x(\vec q)),X_0(x_0(\vec q,q_0));a)\\
&=&\frac{1}{2}\sum\limits_{\substack{m_0\geqslant 1,m_\xi\geqslant 0,\\\cdots}}(-1)^{\lfloor(\frac{f}{2}+\frac{1}{4})m_0-m_{2\xi}\rfloor}\frac{1}{\Gamma(1+m_\xi)}\cdot\frac{\Gamma\bigg((f+1)m_0-m_{2\xi}\bigg)}{\Gamma\bigg(1+m_{2\xi}\bigg)}\\
&&\quad\quad\cdot\frac{q_\xi^{m_\xi}}{(m_0-2m_\xi)!}\cdot\frac{q_{2\xi}^{m_{2\xi}}}{(fm_0+m_\xi-2m_{2\xi})!}\cdot\frac{((-1)^{f}q_0)^{m_0}}{m_0}.
\end{eqnarray*}
We observe that $m_{2k_{2\xi}}$'s are not necessarily nonnegative integers. We choose to ignore the bad $m_{2\xi}$'s and consider only the analytic part of $\mathcal F_{0,1}^{(\mathcal X,\mathcal L_0)}(X(x(\vec q)),X_0(x_0(\vec q,q_0));a)$:
\begin{eqnarray*}
&&F(\vec q,q_0;f)\\
&=&\frac{1}{2}\sum\limits_{\substack{m_0\geqslant 1,\\m_\xi,m_{2\xi}\geqslant 0,\\(f+1)m_0-m_{2\xi}>0}}(-1)^{\lfloor(\frac{f}{2}+\frac{1}{4})m_0-m_{2\xi}\rfloor}\frac{1}{\Gamma(1+m_\xi)}\cdot\frac{\Gamma\bigg((f+1)m_0-m_{2\xi}\bigg)}{\Gamma\bigg(1+m_{2\xi}\bigg)}\\
&&\quad\quad\cdot\frac{q_\xi^{m_\xi}}{(m_0-2m_\xi)!}\cdot\frac{q_{2\xi}^{m_{2\xi}}}{(fm_0+m_\xi-2m_{2\xi})!}\cdot\frac{((-1)^{f}q_0)^{m_0}}{m_0}\\
&=&\frac{1}{2}W(q,q_0;f).
\end{eqnarray*}

The authors of \cite{BC} has shown a correspondence for disc potentials for $f=1,a=\frac{1}{2}$. Their choice of $a$ is different from ours. It is interesting to understand this discrepancy.

\subsection{Example: $\mathbb{C}^3/\mathbb{Z}_5(3,1,1)$}
Let $\xi$ be the generator of $\mathbb{Z}_5$. The action  is given by $$\xi.(z_0,z_1,z_2)=(\xi_5^3\cdot z_0,\xi_5\cdot z_1,\xi_5\cdot z_2).$$ Hence the small part of $\mathbb{Z}_5$ consists of $\xi,2\xi$, with
\begin{eqnarray*}
(F_{\xi}^{(0)},F_{\xi}^{(1)},F_{\xi}^{(2)})&=&(\frac{3}{5},\frac{1}{5},\frac{1}{5}),\\
(F_{2\xi}^{(0)},F_{2\xi}^{(1)},F_{2\xi}^{(2)})&=&(\frac{1}{5},\frac{2}{5},\frac{2}{5}).
\end{eqnarray*} An integral basis of the lattice of invariants is
\ben
[\varepsilon_0,\varepsilon_1,\varepsilon_2]
=[e_0,e_1,e_2]
\begin{bmatrix} 0&0&1 \\ 0 &-1&1 \\ 5 & 1 & 1
\end{bmatrix}
\een
So let
\begin{eqnarray*}
\tilde{v}_0=\left[\begin{array}{c}0\\0\\1\end{array}\right],\tilde{v}_1=\left[\begin{array}{c}0\\-1\\1\end{array}\right],\tilde{v}_2=\left[\begin{array}{c}5\\1\\1\end{array}\right],\tilde{v}_{\xi}=\left[\begin{array}{c}1\\0\\1\end{array}\right],\tilde{v}_{2\xi}=\left[\begin{array}{c}2\\0\\1\end{array}\right].
\end{eqnarray*}
and the toric resolution $X$ is given by Figure 4.

 $D_i$ be the corresponding toric divisor for $i\in\{0,1,2,\xi,2\xi\}$. Then the intersection numbers are:
\begin{displaymath}
\begin{array}{c|ccccc}
&D_0&D_1&D_2&D_{\xi}&D_{2\xi}\\
\hline
D_\xi\cdot D_0&3&1&1&-5&0\\
D_\xi\cdot D_1&1&0&0&-2&1\\
D_\xi\cdot D_2&1&0&0&-2&1\\
D_\xi\cdot D_{2\xi}&0&1&1&1&-3\\
D_{2\xi}\cdot D_1&0&1&1&1&-3\\
D_{2\xi}\cdot D_2&0&1&1&1&-3
\end{array}
\end{displaymath}
From the dual graph,
\begin{center}
\setlength{\unitlength}{1cm}
\begin{picture}(0,15)(0,-5)
\drawline(-1,0)(0,0)(1,-1)(2,-3)(1,0)(1,-1)
\drawline(2,-3)(2.5,-4.25)
\drawline(0,0)(0,4)(1,0)
\drawline(0,4)(-0.5,6.5)
\dottedline[.]{0.2}(-2,0)(-1,0)
\dottedline[.]{0.2}(-1,9)(-0.5,6.5)
\dottedline[.]{0.2}(3,-5.5)(2.5,-4.25)
\put(-1,1){$D_0$}
\put(-1,-1){$D_1$}
\put(2,0){$D_2$}
\put(0.2,0){$D_\xi$}
\put(1,-1){$D_{2\xi}$}
\put(-1,-3){Figure 14}
\end{picture}
\end{center}
we can choose $C_{2\xi}= D_{2\xi} \cdot D_1$ or $D_{2\xi} \cdot D_2$ or $D_{2\xi} \cdot D_\xi$,
and $C_\xi = D_\xi \cdot D_1$ or $D_\xi \cdot D_2$. Then the charge vectors are:
\begin{displaymath}
\begin{array}{ccccccc}
l^{(\xi)}&=&(1,&0,&0,&-2,&1),\\
l^{(2\xi)}&=&(0,&1,&1,&1,&-3),
\end{array}
\end{displaymath}
which form a basis of the lattice of invariants $$\mathbb{L}=\{(l_0,l_1,l_2,l_\xi,l_{2\xi})\in\mathbb{Z}^5\mid \sum\limits_{j=0}^2l_j\tilde{v}_j+l_\xi\tilde{v}_\xi+l_{2\xi}\tilde{v}_{2\xi}=0\}.$$

Suppose that $L_0$ intersects the noncompact toric curve given by $\overline{v_1v_2}$ with framing $f\in\mathbb{Z}_{\geqslant 0}$. The charge vectors for $(X,L_0)$ are:
\begin{displaymath}
\begin{array}{ccccccc|cc}
\tilde l^{(\xi)}&=&(1,&0,&0,&-2,&1,&0,&0),\\
\tilde l^{(2\xi)}&=&(0,&1,&1,&1,&-3,&0,&0),\\
\tilde l^{(0)}&=&(0,&f,&-f-1,&0,&1,&1,&-1).
\end{array}
\end{displaymath}
Then the associated variables are:
\begin{eqnarray*}
q_\xi&=&z_0\cdot z_\xi^{-2}z_{2\xi},\\
q_{2\xi}&=&z_1z_2\cdot z_\xi z_{2\xi}^{-3},\\
q_0&=&z_1^fz_2^{-f-1}\cdot z_{2\xi}\cdot\frac{z_+}{z_-},
\end{eqnarray*}
and the corresponding differential operators are:
\begin{eqnarray*}
\mathcal D_1&=&\Theta_{q_\xi}\bigg(\Theta_{q_\xi}-3\Theta_{q_{2\xi}}+\Theta_{q_0}\bigg)-q_\xi\bigg(-2\Theta_{q_\xi}+\Theta_{q_{2\xi}}\bigg)_2,\\
\mathcal D_2&=&\bigg(\Theta_{q_{2\xi}}+f\Theta_{q_0}\bigg)\bigg(\Theta_{q_{2\xi}}-(f+1)\Theta_{q_0}\bigg)\bigg(-2\Theta_{q_\xi}+\Theta_{q_{2\xi}}\bigg)-q_{2\xi}\bigg(\Theta_{q_\xi}-3\Theta_{q_{2\xi}}+\Theta_{q_0}\bigg)_3,\\
\mathcal D_0&=&\Theta_{q_0}\bigg(\Theta_{q_{2\xi}}+f\Theta_{q_0}\bigg)_f\bigg(\Theta_{q_\xi}-3\Theta_{q_{2\xi}}+\Theta_{q_0}\bigg)+q_0\Theta_{q_0}\bigg(\Theta_{q_{2\xi}}-(f+1)\Theta_{q_0}\bigg)_{f+1}.
\end{eqnarray*}
So the open-close mirror map is given by:
\begin{eqnarray*}
\Omega_\xi&=&\log q_\xi+2\sum\limits_{2m_\xi>m_{2\xi}\geqslant 0}\frac{(2m_\xi-m_{2\xi}-1)!}{m_\xi!(m_{2\xi})^2(m_\xi-3m_{2\xi})!}q_\xi^{m_\xi}(-q_{2\xi})^{m_{2\xi}}\\
&&\quad\quad-\sum\limits_{3m_{2\xi}>m_\xi\geqslant 0}\frac{(3m_{2\xi}-m_\xi-1)!}{m_\xi!(m_{2\xi})^2(m_{2\xi}-2m_\xi)!}(-q_\xi)^{m_\xi}(-q_{2\xi})^{m_{2\xi}},\\
\Omega_{2\xi}&=&\log q_{2\xi}-\sum\limits_{2m_\xi>m_{2\xi}\geqslant 0}\frac{(2m_\xi-m_{2\xi}-1)!}{m_\xi!(m_{2\xi})^2(m_\xi-3m_{2\xi})!}q_\xi^{m_\xi}(-q_{2\xi})^{m_{2\xi}}\\
&&\quad\quad+3\sum\limits_{3m_{2\xi}>m_\xi\geqslant 0}\frac{(3m_{2\xi}-m_\xi-1)!}{m_\xi!(m_{2\xi})^2(m_{2\xi}-2m_\xi)!}(-q_\xi)^{m_\xi}(-q_{2\xi})^{m_{2\xi}},\\
\Omega_0&=&\log q_0-\sum\limits_{3m_{2\xi}>m_\xi\geqslant 0}\frac{(3m_{2\xi}-m_\xi-1)!}{m_\xi!(m_{2\xi})^2(m_{2\xi}-2m_\xi)!}(-q_\xi)^{m_\xi}(-q_{2\xi})^{m_{2\xi}},
\end{eqnarray*}
and the superpotential is:
\begin{eqnarray*}
W&=&\sum\limits_{\substack{m_0\geqslant 1,m_\xi,m_{2\xi}\geqslant 0,\\(f+1)m_0>m_{2\xi}}}\frac{1}{\Gamma(1+m_\xi)\Gamma(1-2m_\xi+m_{2\xi})}\\
&&\quad\quad\cdot\frac{1}{\Gamma(1+m_0+m_\xi-3m_{2\xi})}\cdot\frac{\Gamma\bigg(m_0(f+1)-m_{2\xi})\bigg)}{\Gamma(1+m_0f+m_{2\xi})}\cdot q_\xi^{m_\xi}(-q_{2\xi})^{m_{2\xi}}\frac{[(-1)^fq_0]^{m_0}}{m_0}.
\end{eqnarray*}

For the orbifold $\mathcal X=[\mathbb C^3/\mathbb Z_5(3,1,1)]$, the small part of the group consists of $\xi,2\xi$, with
\begin{eqnarray*}
(F_{\xi}^{(0)},F_{\xi}^{(1)},F_{\xi}^{(2)})&=&(\frac{3}{5},\frac{1}{5},\frac{1}{5}),\\
(F_{2\xi}^{(0)},F_{2\xi}^{(1)},F_{2\xi}^{(2)})&=&(\frac{1}{5},\frac{2}{5},\frac{2}{5}).
\end{eqnarray*}
So the charge vectors for $\mathcal X$ are:
\begin{eqnarray*}
\begin{array}{ccccc}
(-3,&-1,&-1,&5,&0),\\
(-1,&-2,&-2,&0,&5),
\end{array}
\end{eqnarray*}
which form a basis of $\mathbb L\otimes\mathbb Q$. Note that $\mathcal L_0$ is located on the $z_0$-axis, and the corresponding weight is $$(\lambda_0,\lambda_1,\lambda_2)=(\frac{1}{5},-a,a-\frac{1}{5}).$$ So the charge vectors for $(\mathcal X,\mathcal L_0)$ are:
\begin{eqnarray*}
\begin{array}{ccccccc|cc}
\hat l^{(\xi)}&=&(-3,&-1,&-1,&5,&0,&0,&0),\\
\hat l^{(2\xi)}&=&(-1,&-2,&-2,&0,&5,&0,&0),\\
\hat l^{(0)}&=&(1,&-5a,&5a-1,&0,&0,&5,&-5).
\end{array}
\end{eqnarray*}
The corresponding variables are:
\begin{eqnarray*}
x_\xi&=&z_0^{-\frac{3}{5}}z_1^{-\frac{1}{5}}z_2^{-\frac{1}{5}}\cdot z_\xi,\\
x_{2\xi}&=&z_0^{-\frac{1}{5}}z_1^{-\frac{2}{5}}z_2^{-\frac{2}{5}}\cdot z_{2\xi},\\
x_0&=&z_0^{\frac{1}{5}}z_1^{-a}z_2^{a-\frac{1}{5}}\cdot\frac{z_+}{z_-},
\end{eqnarray*}
the orbifold open-close mirror map is given by:
\begin{eqnarray*}
X_\xi&=&-\sum\limits_{\substack{k_\xi,k_{2\xi}\geqslant 0,\\5\mid k_\xi+2k_{2\xi}-1}}\frac{\Gamma(\frac{3}{5}k_\xi+\frac{1}{5}k_{2\xi})}{\Gamma(\frac{3}{5})}\bigg(\frac{\Gamma(\frac{1}{5}k_\xi+\frac{2}{5}k_{2\xi})}{\Gamma(\frac{1}{5})}\bigg)^2\cdot\frac{(-x_\xi)^{k_\xi}}{k_{\xi}!}\cdot\frac{(-x_{2\xi})^{k_{2\xi}}}{k_{2\xi}!},\\
X_{2\xi}&=&-\sum\limits_{\substack{k_\xi,k_{2\xi}\geqslant 0,\\5\mid k_\xi+2k_{2\xi}-2}}\frac{\Gamma(\frac{3}{5}k_\xi+\frac{1}{5}k_{2\xi})}{\Gamma(\frac{1}{5})}\bigg(\frac{\Gamma(\frac{1}{5}k_\xi+\frac{2}{5}k_{2\xi})}{\Gamma(\frac{2}{5})}\bigg)^2\cdot\frac{(-x_\xi)^{k_\xi}}{k_{\xi}!}\cdot\frac{(-x_{2\xi})^{k_{2\xi}}}{k_{2\xi}!},\\
X_0&=&(-1)^{1+\lambda_2}x_0,
\end{eqnarray*}
and the orbifold disc potential is given by:
\begin{eqnarray*}
&&\mathcal F_{0,1}^{(\mathcal X,\mathcal L_0)}(X(x),X_0(x_0);a)\\
&=&\sum\limits_{\substack{d\geqslant 1,k_\xi,k_{2\xi}\geqslant 0,\\\langle\frac{3}{5}k_\xi+\frac{1}{5}k_{2\xi}\rangle=\langle\frac{d}{5}\rangle}}(-1)^{\lfloor\frac{1}{5}k_\xi+\frac{2}{5}k_{2\xi}\rfloor}\frac{1}{\Gamma(1-\frac{3}{5}k_\xi-\frac{1}{5}k_{2\xi}+\lambda_0d)}\cdot\frac{\Gamma(\frac{1}{5}k_\xi+\frac{2}{5}k_{2\xi}-\lambda_2d)}{\Gamma(1-\frac{1}{5}k_\xi-\frac{2}{5}k_{2\xi}+\lambda_1d)}\\
&&\quad\quad\cdot\frac{x_{\xi}^{k_{\xi}}}{k_{\xi}!}\cdot\frac{x_{2\xi}^{k_{2\xi}}}{k_{2\xi}!}\cdot\frac{((-1)^{1+\lambda_2}x_0)^d}{d}.
\end{eqnarray*}

To compare $W$ with $\mathcal F_{0,1}^{(\mathcal X,\mathcal L_0)}$, first note that $\{\tilde l^{(\xi)},\{\tilde l^{(2\xi)},\tilde l^{(0)}\}$ and $\{\hat l^{(\xi)},\hat l^{(2\xi)},\hat l^{(0)}\}$ span the same linear subspace in $\mathbb Q^{7}$ if and only if $$a=-f-\frac{2}{5},$$ which gives the correspondence between $a$ and $f$, and so we obtain the change of variables: 
\begin{eqnarray*}
x_\xi&=&q_\xi^{-\frac{3}{5}}q_{2\xi}^{-\frac{1}{5}},\\
x_{2\xi}&=&q_\xi^{-\frac{1}{5}}q_{2\xi}^{-\frac{2}{5}},\\
x_0&=&(-1)^{1+\lambda_2+f}q_0q_\xi^{\frac{1}{5}}q_{2\xi}^{\frac{2}{5}},
\end{eqnarray*}
where the phase factor $(-1)^{1+\lambda_2+f}$ is included for convenience of comparison. Via the above identification, we have:
\begin{eqnarray*}
&&\mathcal F_{0,1}^{(\mathcal X,\mathcal L_0)}(X(x(\vec q)),X_0(x_0(\vec q,q_0));a)\\
&=&\sum\limits_{\substack{d\geqslant 1,k_\xi,k_{2\xi}\geqslant 0,\\\frac{d}{5}-\frac{3}{5}k_\xi-\frac{1}{5}k_{2\xi}\in\mathbb Z_{\geqslant 0}}}(-1)^{\lfloor\frac{1}{5}k_\xi+\frac{2}{5}k_{2\xi}\rfloor}\frac{1}{\Gamma(1-\frac{3}{5}k_\xi-\frac{1}{5}k_{2\xi}+\frac{d}{5})}\cdot\frac{\Gamma\bigg(\frac{1}{5}k_\xi+\frac{2}{5}k_{2\xi}-(a-\frac{1}{5})d\bigg)}{\Gamma(1-\frac{1}{5}k_\xi-\frac{2}{5}k_{2\xi}-ad)}\\
&&\quad\quad\cdot\frac{q_\xi^{\frac{d}{5}-\frac{3}{5}k_\xi-\frac{1}{5}k_{2\xi}}}{k_\xi!}\cdot\frac{q_{2\xi}^{\frac{2}{5}d-\frac{1}{5}k_\xi-\frac{2}{5}k_{2\xi}}}{k_{2\xi}!}\cdot\frac{((-1)^{f}q_0)^d}{d}.
\end{eqnarray*}
Set 
\begin{eqnarray*}
\left\{\begin{array}{ccl}m_0&=&d\in\mathbb Z_{\geqslant 1},\\\\m_\xi&=&\frac{d}{5}-\frac{3}{5}k_\xi-\frac{1}{5}k_{2\xi}\in\mathbb Z_{\geqslant 0},\\\\m_{2\xi}&=&\frac{2d}{5}-\frac{1}{5}k_\xi-\frac{2}{5}k_{2\xi}=2m_\xi+k_\xi\in\mathbb Z_{\geqslant 0},\end{array}\right.
\end{eqnarray*}
and then we have
\begin{eqnarray*}
&&\mathcal F_{0,1}^{(\mathcal X,\mathcal L_0)}(X(x(\vec q)),X_0(x_0(\vec q,q_0));a)\\
&=&\sum\limits_{m_0\geqslant 1,m_{2\xi}\geqslant 2m_\xi\geqslant 0}(-1)^{\lfloor\frac{2}{5}m_0-m_{2\xi}\rfloor}\frac{1}{\Gamma(1+m_\xi)}\cdot\frac{\Gamma\bigg((f+1)m_0-m_{2\xi}\bigg)}{\Gamma(1+fm_0+m_{2\xi})}\\
&&\quad\quad\cdot\frac{q_\xi^{m_\xi}}{(m_{2\xi}-2m_\xi)!}\cdot\frac{q_{2\xi}^{m_{2\xi}}}{(m_0+m_\xi-3m_{2\xi})!}\cdot\frac{((-1)^fq_0)^{m_0}}{m_0}\\
&=&W(q,q_0;f).
\end{eqnarray*}

\subsection{Example: $\mathbb{C}^3/\mathbb{Z}_6(1,2,3)$} 
In this example, the group acts on the $z_0$-axis effectively, and we test our conjecture by showing the corresponding inverse matrices.

Let $\xi$ be the generator of $\mathbb{Z}_6$. The action is given by $$\xi.(z_0,z_1,z_2)=(\xi_6\cdot z_0,\xi_6^2\cdot z_1,\xi_6^3\cdot z_2).$$ Hence the small part of $\mathbb{Z}_6$ consists of $\xi,2\xi,3\xi,4\xi$,with
\begin{eqnarray*}
(F_{\xi}^{(0)},F_{\xi}^{(1)},F_{\xi}^{(2)})&=&(\frac{1}{6},\frac{1}{3},\frac{1}{2}),\\
(F_{2\xi}^{(0)},F_{2\xi}^{(1)},F_{2\xi}^{(2)})&=&(\frac{1}{3},\frac{2}{3},0).\\
(F_{3\xi}^{(0)},F_{3\xi}^{(1)},F_{3\xi}^{(2)})&=&(\frac{1}{2},0,\frac{1}{2}).\\
(F_{4\xi}^{(0)},F_{4\xi}^{(1)},F_{4\xi}^{(2)})&=&(\frac{2}{3},\frac{1}{3},0).
\end{eqnarray*} 
An integral basis of the lattice of invariants is
\ben
[\varepsilon_0,\varepsilon_1,\varepsilon_2]
=[e_0,e_1,e_2]
\begin{bmatrix} 0&0&1 \\ 3 &0&1 \\ 0 & 2 &1
\end{bmatrix}
\een
This is not a basis obtained from our procedure. We invite the readers to apply our procedure and compare the resulting charge vectors with those demonstrated here.

We set
\begin{eqnarray*}
\tilde{w}_0=\left[\begin{array}{c}0\\0\\1\end{array}\right],\tilde{w}_1=\left[\begin{array}{c}3\\0\\1\end{array}\right],\tilde{w}_2=\left[\begin{array}{c}0\\2\\1\end{array}\right],\tilde{v}_{\xi}=\left[\begin{array}{c}1\\1\\1\end{array}\right],\tilde{v}_{2\xi}=\left[\begin{array}{c}2\\0\\1\end{array}\right],\tilde{v}_{3\xi}=\left[\begin{array}{c}0\\1\\1\end{array}\right],\tilde{v}_{4\xi}=\left[\begin{array}{c}1\\0\\1\end{array}\right].
\end{eqnarray*}
There are five different toric resolutions of $\mathbb{C}^3/\mathbb{Z}_6(1,2,3)$:
\begin{center}
\setlength{\unitlength}{1cm}
\begin{picture}(13,10)(0,-4.5)
\drawline(0,0)(1,1)(1,0)(2,0)(1,1)(3,0)(0,0)(0,2)(3,0)(1,1)(0,2)(0,1)(1,1)
\multiput(0,0)(0,1){3}{\circle*{0.15}}
\multiput(0,0)(1,0){4}{\circle*{0.15}}
\put(1,1){\circle*{0.15}}
\put(2,1){$I$}
\put(-0.2,-0.3){$v_0$}
\put(2.8,-0.3){$v_1$}
\put(-0.4,1.9){$v_2$}
\put(1,1){$v_\xi$}
\put(1.8,-0.3){$v_{2\xi}$}
\put(-0.5,0.9){$v_{3\xi}$}
\put(0.8,-0.3){$v_{4\xi}$}
\put(3.5,1){\vector(1,0){1}} \put(4.5,1){\vector(-1,0){1}}

\drawline(5,0)(8,0)(5,2)(5,0)(6,0)(5,1)(6,1)(6,0)(7,0)(5,2)(6,1)(8,0)
\put(7,1){$II$}
\multiput(5,0)(1,0){4}{\circle*{0.15}}
\multiput(5,0)(0,1){3}{\circle*{0.15}}
\put(6,1){\circle*{0.15}}
\put(8.5,1){\vector(1,0){1}} \put(9.5,1){\vector(-1,0){1}}
\put(8.5,2){\vector(1,1){1}} \put(9.5,3){\vector(-1,-1){1}}
\drawline(10,3)(13,3)(10,5)(10,3)(10,4)(11,3)(10,5)(12,3)(13,3)(11,4)(11,3)
\put(12,4){$III$}
\multiput(10,3)(1,0){4}{\circle*{0.15}}
\multiput(10,3)(0,1){3}{\circle*{0.15}}
\put(11,4){\circle*{0.15}}
\drawline(10,0)(13,0)(10,2)(10,0)(11,0)(10,1)(12,0)(10,2)(10,1)(11,1)(13,0)
\put(12,1){$IV$}
\multiput(10,0)(1,0){4}{\circle*{0.15}}
\multiput(10,0)(0,1){3}{\circle*{0.15}}
\put(11,1){\circle*{0.15}}
\put(11,-0.5){\vector(0,-1){1}} \put(11,-1.5){\vector(0,1){1}}

\drawline(10,-4)(13,-4)(10,-2)(10,-4)(11,-4)(10,-3)(12,-4)(13,-4)(10,-3)(11,-3)(13,-4)(11,-3)(10,-2)
\put(12,-3){$V$}
\multiput(10,-4)(1,0){4}{\circle*{0.15}}
\multiput(10,-4)(0,1){3}{\circle*{0.15}}
\put(11,-3){\circle*{0.15}}
\put(3,-3){Figure 15}
\end{picture}
\end{center}
For each toric resolution, 
the corresponding charge vectors form a basis of the lattice of invariants 
$$\mathbb{L}=\{(l_{0},l_{1},l_{2},l_{\xi},l_{2\xi},l_{3\xi},l_{4\xi})\in\mathbb{Z}^7
\mid\sum\limits_{j=0}^2l_{j}\tilde{w}_j+\sum\limits_{j=1}^4l_{j\xi}\tilde{v}_{j\xi}=0\}.$$ 
Let $D_i$ be the corresponding toric divisor for $i\in\{0,1,2,\xi,2\xi,3\xi,4\xi\}$. 
Then we may calculate the intersection numbers and the charge vectors as follows:
\begin{description}
\item[Phase $I$] The intersection numbers are:
\begin{displaymath}
\begin{array}{c|ccccccc}
&D_0&D_1&D_2&D_{\xi}&D_{2\xi}&D_{3\xi}&D_{4\xi}\\
\hline
D_\xi.D_0&-1&0&0&-1&0&1&1\\
D_\xi.D_1&0&0&1&-2&1&0&0\\
D_\xi.D_2&0&1&1&-3&0&1&0\\
D_\xi.D_{2\xi}&0&1&0&0&-2&0&1\\
D_\xi.D_{3\xi}&1&0&1&0&0&-2&0\\
D_\xi.D_{4\xi}&1&0&0&0&1&0&-2
\end{array}
\end{displaymath}
From the dual graph,
\begin{center}
\setlength{\unitlength}{1cm}
\begin{picture}(0,10)(0,-1.5)
\drawline(0,-1)(0,0)(1,0)(2,1)(3,6)(-1,2)(-1,1)(0,0)
\drawline(1,-1)(1,0)
\drawline(2,0)(2,1)
\drawline(3,6)(4,7.5)
\drawline(-2,2)(-1,2)
\drawline(-2,1)(-1,1)
\dottedline[.]{0.2}(0,-1)(0,-1.5)
\dottedline[.]{0.2}(1,-1)(1,-1.5)
\dottedline[.]{0.2}(2,0)(2,-1)
\dottedline[.]{0.2}(-2,1)(-2.5,1)
\dottedline[.]{0.2}(-2,2)(-2.5,2)
\dottedline[.]{0.2}(4,7.5)(4.5,8.25)
\put(0.2,1.2){$D_\xi$}
\put(0.2,-0.5){$D_{4\xi}$}
\put(-1.8,1.2){$D_{3\xi}$}
\put(1.2,-0.3){$D_{2\xi}$}
\put(-1.5,0.2){$D_0$}
\put(2.5,0.7){$D_1$}
\put(-1.5,2.5){$D_2$}
\put(-2,-1){Figure 16}
\end{picture}
\end{center}
we can choose $C_{2\xi}=D_\xi\cdot D_{2\xi},C_{3\xi}=D_\xi\cdot D_{3\xi},C_{4\xi}=D_{4\xi}\cdot D_{4\xi}$ and $C_\xi=D_\xi\cdot D_0$. Then the charge vectors are:
\begin{displaymath}
\begin{array}{ccccccccc}
l^{(\xi)}&=&(-1,&0,&0,&-1,&0,&1,&1),\\
l^{(2\xi)}&=&(0,&1,&0,&0,&-2,&0,&1),\\
l^{(3\xi)}&=&(1,&0,&1,&0,&0,&-2,&0),\\
l^{(4\xi)}&=&(1,&0,&0,&0,&1,&0,&-2).
\end{array}
\end{displaymath}
and our conjecture is tested by:
\ben
\begin{bmatrix} -1&0&1&1&-1 \\ 0 &-2&0&1&0 \\ 0 & 0 & -2&0&1\\0&1&0&-2&1\\1&0&0&0&0
\end{bmatrix}^{-1}=\begin{bmatrix}0&0&0&0&1\\2&0&1&1&2\\3&1&1&2&3\\4&1&2&2&4\\6&2&3&4&6\end{bmatrix}
\een

\item[Phase $II$] The intersection numbers are:
\begin{displaymath}
\begin{array}{c|ccccccc}
&D_0&D_1&D_2&D_\xi&D_{2\xi}&D_{3\xi}&D_{4\xi}\\
\hline
D_{3\xi}.D_{4\xi}&1&0&0&1&0&-1&-1\\
D_\xi.D_1&0&0&1&-2&1&0&0\\
D_\xi.D_2&0&1&1&-3&0&1&0\\
D_\xi.D_{2\xi}&0&1&0&0&-2&0&1\\
D_\xi.D_{3\xi}&0&0&1&-1&0&-1&1\\
D_\xi.D_{4\xi}&0&0&0&-1&1&1&-1
\end{array}
\end{displaymath}
From the dual graph,
\begin{center}
\setlength{\unitlength}{1cm}
\begin{picture}(0,9)(0,-2.5)
\drawline(-1,-2)(-1,-1)(0,0)(1,0)(2,1)(4,5)(0,1)(0,0)
\drawline(-2,-1)(-1,-1)
\drawline(-1,1)(0,1)
\drawline(4.5,5.75)(4,5)
\drawline(2,1)(2,0)
\drawline(1,0)(1,-1)
\dottedline[.]{0.2}(-1,-2)(-1,-2.5)
\dottedline[.]{0.2}(1,-1)(1,-1.5)
\dottedline[.]{0.2}(2,0)(2,-0.5)
\dottedline[.]{0.2}(5,6.5)(4.5,5.75)
\dottedline[.]{0.2}(-1,1)(-1.5,1)
\dottedline[.]{0.2}(-2,-1)(-2.5,-1)
\put(0.2,0.2){$D_\xi$}
\put(-0.5,-0.8){$D_{4\xi}$}
\put(1.2,-0.3){$D_{2\xi}$}
\put(2.2,0.7){$D_{1}$}
\put(-1.7,-1.8){$D_{0}$}
\put(-1.5,0.2){$D_{3\xi}$}
\put(-1,1.2){$D_{2}$}
\put(-3,-2.5){Figure 17}
\end{picture}
\end{center}
we can choose 
\begin{eqnarray*}
&&C_\xi=D_\xi\cdot D_{3\xi},C_{2\xi}=D_\xi\cdot D_{2\xi},C_{3\xi}=D_{3\xi}\cdot D_{4\xi},C_{4\xi}=D_{4\xi}\cdot D_\xi,\\
&\textrm{or,}&\\
&&C_\xi=D_\xi\cdot D_{4\xi},C_{2\xi}=D_\xi\cdot D_{2\xi},C_{3\xi}=D_{3\xi}\cdot D_\xi,C_{4\xi}=D_{4\xi}\cdot D_{3\xi}.
\end{eqnarray*}
We can see that the two choices give the same set of charge vectors. With respect to the first choice, the charge vectors are indexed by:
\begin{displaymath}
\begin{array}{ccccccccc}
l^{(\xi)}&=&(0,&0,&1,&-1,&0,&-1,&1),\\
l^{(2\xi)}&=&(0,&1,&0,&0,&-2,&0,&1),\\
l^{(3\xi)}&=&(1,&0,&0,&1,&0,&-1,&-1),\\
l^{(4\xi)}&=&(0,&0,&0,&-1,&1,&1,&-1).
\end{array}
\end{displaymath}
and our conjecture is tested by:
\ben
\begin{bmatrix} -1&0&-1&1&0 \\ 0 &-2&0&1&0 \\ 1 & 0 & -1&-1&1\\-1&1&1&-1&0\\1&0&0&0&0
\end{bmatrix}^{-1}=\begin{bmatrix}0&0&0&0&1\\1&0&0&1&2\\1&1&0&2&3\\2&1&0&2&4\\3&2&1&4&6\end{bmatrix}
\een

\item[Phase $III$] The intersection numbers are:
\begin{displaymath}
\begin{array}{c|ccccccc}
&D_0&D_1&D_2&D_\xi&D_{2\xi}&D_{3\xi}&D_{4\xi}\\
\hline
D_{3\xi}.D_{4\xi}&1&0&1&0&0&-2&0\\
D_\xi.D_1&0&0&1&-2&1&0&0\\
D_\xi.D_2&0&1&2&-4&0&0&1\\
D_\xi.D_{2\xi}&0&1&0&0&-2&0&1\\
D_{4\xi}.D_{2}&0&0&-1&1&0&1&-1\\
D_\xi.D_{4\xi}&0&0&1&-2&1&0&0
\end{array}
\end{displaymath}
From the dual graph,
\begin{center}
\setlength{\unitlength}{1cm}
\begin{picture}(0,9)(0,-3.5)
\drawline(-3,-3)(-3,-2)(-2,-1)(0,0)(1,0)(2,1)(3,3)(0,0)
\drawline(-4,-2)(-3,-2)
\drawline(-3,-1)(-2,-1)
\drawline(1,0)(1,-1)
\drawline(2,1)(2,0)
\drawline(3,3)(4,4.5)
\dottedline[.]{0.2}(-3,-3)(-3,-3.5)
\dottedline[.]{0.2}(1,-1)(1,-1.5)
\dottedline[.]{0.2}(2,0)(2,-0.5)
\dottedline[.]{0.2}(4,4.5)(4.5,5.25)
\dottedline[.]{0.2}(-3,-1)(-3.5,-1)
\dottedline[.]{0.2}(-4,-2)(-4.5,-2)
\put(0.5,0.2){$D_\xi$}
\put(-1,-0.8){$D_{4\xi}$}
\put(-4,-2.5){$D_0$}
\put(-3.5,-1.7){$D_{3\xi}$}
\put(-2,0.2){$D_2$}
\put(1.2,-0.3){$D_{2\xi}$}
\put(2.2,0.7){$D_1$}
\put(0,-3){Figure 18}
\end{picture}
\end{center}
we can choose $$C_\xi=D_\xi\cdot D_{4\xi},C_{2\xi}=D_{\xi}\cdot D_{2\xi},C_{3\xi}=D_{3\xi}\cdot D_{4\xi},C_{4\xi}=D_{4\xi}\cdot D_2.$$ Then the charge vectors are:
\begin{displaymath}
\begin{array}{ccccccccc}
l^{(\xi)}&=&(0,&0,&1,&-2,&1,&0,&0),\\
l^{(2\xi)}&=&(0,&1,&0,&0,&-2,&0,&1),\\
l^{(3\xi)}&=&(1,&0,&1,&0,&0,&-2,&0),\\
l^{(4\xi)}&=&(0,&0,&-1,&1,&0,&1,&-1).
\end{array}
\end{displaymath}
and our conjecture is tested by:
\ben
\begin{bmatrix} -2&1&0&0&0 \\ 0 &-2&0&1&0 \\ 0 & 0 & -2&0&1\\1&0&1&-1&0\\1&0&0&0&0
\end{bmatrix}^{-1}=\begin{bmatrix}0&0&0&0&1\\1&0&0&0&2\\2&1&0&1&3\\2&1&0&0&4\\4&2&1&2&6\end{bmatrix}
\een

\item[Phase $IV$] The intersection numbers are:
\begin{displaymath}
\begin{array}{c|ccccccc}
&D_0&D_1&D_2&D_\xi&D_{2\xi}&D_{3\xi}&D_{4\xi}\\
\hline
D_{3\xi}.D_{4\xi}&1&0&0&0&1&0&-2\\
D_\xi.D_1&0&0&1&-2&1&0&0\\
D_\xi.D_2&0&1&1&-3&0&1&0\\
D_\xi.D_{2\xi}&0&1&0&-1&-1&1&0\\
D_\xi.D_{3\xi}&0&0&1&-2&1&0&0\\
D_{2\xi}.D_{3\xi}&0&0&0&1&-1&-1&1
\end{array}
\end{displaymath}
From the dual graph,
\begin{center}
\setlength{\unitlength}{1cm}
\begin{picture}(0,10)(0,-4.5)
\drawline(-2,-4)(-2,-3)(-1,-2)(0,0)(1,1)(2,3)(0,1)(0,0)
\drawline(-3,-3)(-2,-3)
\drawline(-1,-2)(-1,-3)
\drawline(1,1)(1,0)
\drawline(2,3)(2.5,3.75)
\drawline(0,1)(-1,1)
\dottedline[.]{0.2}(-2,-4)(-2,-4.5)
\dottedline[.]{0.2}(-1,-3)(-1,-3.5)
\dottedline[.]{0.2}(1,0)(1,-0.5)
\dottedline[.]{0.2}(2.5,3.75)(3,4.5)
\dottedline[.]{0.2}(-1,1)(-1.5,1)
\dottedline[.]{0.2}(-3,-3)(-3.5,-3)
\put(0.1,0.7){$D_\xi$}
\put(1.5,1.8){$D_1$}
\put(-0.5,-1.5){$D_{2\xi}$}
\put(-1.8,-3.5){$D_{4\xi}$}
\put(-3,-3.5){$D_0$}
\put(-2,-1.5){$D_{3\xi}$}
\put(-1,1.8){$D_2$}
\put(0,-4){Figure 19}
\end{picture}
\end{center}
we can choose 
\begin{eqnarray*}
&&C_\xi=D_\xi\cdot D_{2\xi},C_{2\xi}=D_{2\xi}\cdot D_{3\xi},C_{3\xi}=D_{3\xi}\cdot D_\xi,C_{4\xi}=D_{4\xi}\cdot D_{3\xi},\\
&\textrm{or,}&\\
&&C_\xi=D_\xi\cdot D_{3\xi},C_{2\xi}=D_{2\xi}\cdot D_\xi,C_{3\xi}=D_{3\xi}\cdot D_{2\xi},C_{4\xi}=D_{4\xi}\cdot D_{3\xi}.
\end{eqnarray*}
With respect to the first choice, the charge vectors are indexed by:
\begin{displaymath}
\begin{array}{ccccccccc}
l^{(\xi)}&=&(0,&1,&0,&-1,&-1,&1,&0),\\
l^{(2\xi)}&=&(0,&0,&0,&1,&-1,&-1,&1),\\
l^{(3\xi)}&=&(0,&0,&1,&-2,&1,&0,&0),\\
l^{(4\xi)}&=&(1,&0,&0,&0,&1,&0,&-2).
\end{array}
\end{displaymath}
and our conjecture is tested by:
\ben
\begin{bmatrix} -1&-1&1&0&0 \\1 &-1&-1&1&0 \\ -2 & 1 & 0&0&0\\0&1&0&-2&1\\1&0&0&0&0
\end{bmatrix}^{-1}=\begin{bmatrix}0&0&0&0&1\\0&0&1&0&2\\1&0&1&0&3\\1&1&2&0&4\\2&2&3&1&6\end{bmatrix}
\een

\item[Phase $V$] The intersection numbers are:
\begin{displaymath}
\begin{array}{c|ccccccc}
&D_0&D_1&D_2&D_\xi&D_{2\xi}&D_{3\xi}&D_{4\xi}\\
\hline
D_{3\xi}.D_{4\xi}&1&0&0&0&1&0&-2\\
D_\xi.D_1&0&1&1&-3&0&1&0\\
D_\xi.D_2&0&1&1&-3&0&1&0\\
D_{3\xi}.D_1&0&-1&0&1&1&-1&0\\
D_\xi.D_{3\xi}&0&1&1&-3&0&1&0\\
D_{2\xi}.D_{3\xi}&0&1&0&0&-2&0&1
\end{array}
\end{displaymath}
From the dual graph,
\begin{center}
\setlength{\unitlength}{1cm}
\begin{picture}(0,11)(0,-4.5)
\drawline(-2,-4)(-2,-3)(-1,-2)(0,0)(1,3)(2,5)(1,4)(1,3)
\drawline(-1,-3)(-1,-2)
\drawline(0,-1)(0,0)
\drawline(2.5,5.75)(2,5)
\drawline(0,4)(1,4)
\drawline(-3,-3)(-2,-3)
\dottedline[.]{0.2}(-2,-4)(-2,-4.5)
\dottedline[.]{0.2}(-1,-3)(-1,-3.5)
\dottedline[.]{0.2}(0,-1)(0,-1.5)
\dottedline[.]{0.2}(2.5,5.75)(3,6.5)
\dottedline[.]{0.2}(0,4)(-0.5,4)
\dottedline[.]{0.2}(-3,-3)(-3.5,-3)
\put(1,0){$D_1$}
\put(-0.7,-1.8){$D_{2\xi}$}
\put(-1.7,-3.8){$D_{4\xi}$}
\put(-2.7,-3.8){$D_0$}
\put(-2,0){$D_{3\xi}$}
\put(1.1,3.8){$D_\xi$}
\put(0,-4){Figure 20}
\end{picture}
\end{center}
we can choose 
$$C_\xi=D_\xi\cdot D_{3\xi},C_{2\xi}=D_{2\xi}\cdot D_{3\xi},C_{3\xi}=D_{3\xi}\cdot D_1,C_{4\xi}=D_{4\xi}\cdot D_{3\xi}.$$ Then the charge vectors are:
\begin{displaymath}
\begin{array}{ccccccccc}
l^{(\xi)}&=&(0,&1,&1,&-3,&0,&1,&0),\\
l^{(2\xi)}&=&(0,&1,&0,&0,&-2,&0,&1),\\
l^{(3\xi)}&=&(0,&-1,&0,&1,&1,&-1,&0),\\
l^{(4\xi)}&=&(1,&0,&0,&0,&1,&0,&-2).
\end{array}
\end{displaymath}
\end{description}
and our conjecture is tested by:
\ben
\begin{bmatrix} -3&0&1&0&0 \\ 0 &-2&0&1&0 \\ 1 & 1 & -1&0&0\\0&1&0&-2&1\\1&0&0&0&0
\end{bmatrix}^{-1}=\begin{bmatrix}0&0&0&0&1\\1&0&1&0&2\\1&0&0&0&3\\2&1&2&0&4\\3&2&3&1&6\end{bmatrix}
\een

\paragraph{\textbf{Acknowledgement}} The first author would like to thank Renzo Cavalieri for helpful discussions. The second author is partially supported by NSFC grant 11171174.


\begin{thebibliography}{99}


\bibitem{AGV}
Abramovich, D., Graber, T., Vistoli, A.: Gromov-Witten theory of Deligne-Mumford stacks. \emph{Amer. J. Math.}, \textbf{130} (5), 1337--1398(2008) 



\bibitem{AV}
Aganagic, M., Vafa, C.: Mirror symmetry, D-branes and counting holomorphic discs. Preprint, hep-th/0012041

\bibitem{BC}
Brini, A., Cavalieri, R.: Open orbifold Gromov-Witten invariants of $[\mathbb C^3/\mathbb Z_n]$: localization and mirror symmetry. \emph{Selecta Math. (N.S.)}, \textbf{17}(4), 879--933 (2011)


\bibitem{BG}
Bryan, J., Graber, T.: The crepant resolution conjecture. Algebraic geometry--Seattle 2005. Part 1, 2--42, Proc. Sympos. Pure Math., \textbf{80}, Part 1, Amer. Math. Soc., Providence, RI, (2009)

\bibitem{BKMP09}
Bouchard, V., Klemm, A., Mari\~no, M., Pasquetti, S.: Remodeling the B-model. \emph{Comm. Math. Phys.}, \textbf{287}(1), 117--178 (2009)




\bibitem{CK}
Cox, D., Katz, S.: Mirror symmetry and algebraic geometry. Mathematical Surveys and Monographs, \textbf{68}. American Mathematical Society, Providence, RI, (1999)

\bibitem{CKYZ}
Chiang, T.-M., Klemm, A., Yau, S.-T., Zaslow, E.: Local mirror symmetry: calculations and interpretations. \emph{Adv. Theor. Math. Phys.}, \textbf{3}(3), 495--565 (1999)


\bibitem{CCIT}
Coates, T., Corti, A., Iritani, H., Tseng, H.-H.: Computing genus-zero twisted Gromov-Witten invariants. \emph{Duke Math. J.}, \textbf{147}(3), 377--438 (2009) 


\bibitem{CR}
Chen, W., Ruan, Y,: Orbifold Gromov-Witten theory. Orbifolds in mathematics and physics (Madison, WI, 2001), 25--85, Contemp. Math., \textbf{310}, Amer. Math. Soc., Providence, RI, (2002)

\bibitem{dOFS}
De la Ossa, X., Florea, B., Skarke, H.: D-branes on noncompact Calabi-Yau manifolds: K-theory and monodromy. \emph{Neclear Phys. B}, \textbf{644}(1-2), 170--200 (2002)

\bibitem{FL}
Fang, B., Liu, C.-C.: Open Gromov-Witten invariants of toric Calabi-Yau 3-folds. \emph{Comm. Math. Phys.}, \textbf{323}(1), 285--328 (2013)


\bibitem{GP}
Graber, T., Pandharipande, R.: Localization of virtual classes. \emph{Invent. Math.}, \textbf{135}(2), 487--518 (1999)


\bibitem{HKTY}
Hosono, S., Klemm, A., Theisen, S., Yau, S.-T.: Mirror symmetry, mirror map and applications to complete intersection Calabi-Yau spaces. \emph{Nuclear Phys. B}, \textbf{433}(3), 501--552 (1995)

\bibitem{JK}
Jarvis, T., Kimura, T.: Orbifold quantum cohomology of the classifying space of a finite group. Orbifolds in mathematics and physics (Madison, WI, 2001), 123--134, 
Contemp. Math., \textbf{310}, Amer. Math. Soc., Providence, RI, (2002)

\bibitem{LLLZ}
Li, J., Liu, C.-C., Liu, K., Zhou, J.: A mathematical theory of the topological vertex. \emph{Geom. Topol.}, \textbf{13}(1), 527--621 (2009)

\bibitem{MOP}
Markushevich, D., Olshanetsky, M., Perelomov, A.: Description of a class of superstring compactifications related to semi-simple lie algebras. \emph{Comm. Math. Phys.}, \textbf{111}(2), 247--274 (1987)

\bibitem{MT}
Ram Murty, M., Thain, N.: Pick's theorem via Minkowski's theorem. \emph{Amer. Math. Monthly}, \textbf{114}(8), 732--736 (2007) 


\bibitem{Reid}
Reid, M.: La correspondance de McKay. Séminaire Bourbaki, Vol. 1999/2000. Ast\'erisque \textbf{276}, 53--72 (2002)


\bibitem{Ru}
Ruan, Y.: The cohomology ring of crepant resolutions of orbifolds. Gromov-Witten theory of spin curves and orbifolds, 117--126, Contemp. Math., \textbf{403}, Amer. Math. Soc., Providence, RI, (2006) 

\bibitem{S}
Stienstra, J.: GKZ hypergeometric structures. Arithmetic and geometry around hypergeometric functions, 313--371, Progr. Math., \textbf{260}, Birkhäuser, Basel, (2007)

\bibitem{W}
Witten, E.: Phases of $N=2$ theories in two dimensions. \emph{Nuclear Phys. B}, \textbf{403}(1-2), 159--222 (1993)

\bibitem{Z06}
Zhou, J.: Some open Gromov-Witten invariants of the resolved conifold. Inspired by S. S. Chern, 487--513, Nankai Tracts Math., \textbf{11}, World Sci. Publ., Hackensack, NJ, (2006) 

\bibitem{Z07}
Zhou, J.: On computations of Hurwitz-Hodge integrals. arXiv:0710.1679

\bibitem{Z08}
Zhou, J.: Crepant resolution conjecture in all genera for type A singularities. math.AG/0811.2023

\bibitem{Z11}
Zhou, J.: Crepant resolutions of Gorenstein singularities and linear sigma model. unpublished (2011)
\end{thebibliography}
\end{document}